\documentclass[11pt,letterpaper]{article}

\usepackage[backend=biber,
style=alphabetic,
sorting=nyt,
minalphanames=3,
maxbibnames=99]{biblatex}

\usepackage{amsmath,amsthm,thmtools,amsfonts,amssymb}
\usepackage{mathrsfs}
\usepackage{mathtools}
\usepackage{physics}
\usepackage{braket}
\usepackage{array}
\allowdisplaybreaks

\usepackage[margin=1in]{geometry}
\usepackage{xspace}
\usepackage{xcolor}
\usepackage{subcaption}
\usepackage{enumitem}
\usepackage{microtype}
\usepackage{csquotes}

\usepackage{algorithm}
\usepackage{algpseudocode}

\usepackage{hyperref}
\hypersetup{
    colorlinks=true,
    linkcolor=black,
    citecolor=black,
    filecolor=black,
    urlcolor=black,
}
\usepackage[capitalize]{cleveref}
\crefname{step}{step}{steps}
\Crefname{step}{Step}{Steps}

\usepackage[skins,many]{tcolorbox}

\makeatletter
\newcounter{HALG@line}
\renewcommand{\theHALG@line}{\thealgorithm.\arabic{ALG@line}}
\makeatother
\newtheorem{theorem}{Theorem}[section]
\newtheorem{lemma}[theorem]{Lemma}
\newtheorem{corollary}[theorem]{Corollary}

\theoremstyle{definition}
\newtheorem{definition}[theorem]{Definition}
\theoremstyle{remark}

\newcommand{\local}{LOCAL\xspace}

\DeclareMathOperator{\dist}{dist} %
\newcommand{\neighborhood}{\NN}
\DeclareMathOperator{\poly}{poly} %
\newcommand{\ceil}[1]{\left \lceil #1 \right \rceil}
\newcommand{\floor}[1]{\left \lfloor #1 \right \rfloor}

\DeclareMathOperator{\diam}{diam}

\newcommand{\nats}{\mathbb{N}}

\newcommand{\reals}{\mathbb{R}}

\renewcommand{\AA}{\mathcal{A}}

\newcommand{\CC}{\mathcal{C}}

\newcommand{\EE}{\mathcal{E}}

\newcommand{\HH}{\mathcal{H}}

\newcommand{\LL}{\mathcal{L}}

\newcommand{\NN}{\mathcal{N}}

\newcommand{\VV}{\mathcal{V}}

\newcommand{\st}{\ \middle| \ }

\newcommand{\problem}{\Pi}
\newcommand{\maxDeg}{\Delta}

\newcommand{\lblnodes}{\ell_{\operatorname{nodes}}}
\newcommand{\lbledges}{\ell_{\operatorname{half-edges}}}
\newcommand{\pot}{\Psi}
\newcommand{\graphPot}{\operatorname{Pot}}

\newcommand{\imp}{\operatorname{Imp}}
\newcommand{\impRatio}{\operatorname{IR}}

\makeatletter
\DeclareMathOperator{\regularlog}{log}
\renewcommand{\log}{\protect\@ifstar{\regularlog^*}{\regularlog}}
\makeatother

\newcommand{\Vin}{\VV_{\operatorname{in}}}
\newcommand{\Ein}{\EE_{\operatorname{in}}}
\newcommand{\Vout}{\VV_{\operatorname{out}}}
\newcommand{\Eout}{\EE_{\operatorname{out}}}
\newcommand{\inptLbl}{\ell_{\operatorname{in}}}
\newcommand{\outLbl}{\ell_{\operatorname{out}}}

\newtcolorbox{myframe}[2][]{%
	breakable,enhanced,colback=white,colframe=black,coltitle=black,
	sharp corners,boxrule=0.4pt,
	fonttitle=\itshape,
	attach boxed title to top left={yshift=-0.3\baselineskip-0.4pt,xshift=2mm},
	boxed title style={tile,size=minimal,left=0.5mm,right=0.5mm,
		colback=white,before upper=\strut},
	title=#2,#1
}

\newcommand{\IR}{\operatorname{IR}\xspace}

\begin{document}
\begin{flushleft}
    \huge\bf
    Distributed Algorithms for Potential Problems
\end{flushleft}
\smallskip

\newcommand{\myemail}[1]{\,$\cdot$\, {\small #1}}
\newcommand{\myaff}[1]{\,$\cdot$\, {\small #1}\par\medskip}
\newenvironment{myabstract}{\list{}{\listparindent 1.5em \itemindent \listparindent \leftmargin 0cm \rightmargin 0cm \parsep 0pt} \item\relax}{\endlist}
\newenvironment{mycover}{\list{}{\listparindent 0pt \itemindent \listparindent \leftmargin 0cm \rightmargin 1.5cm \parsep 0pt} \raggedright \item\relax}{\endlist}

\begin{mycover}

\textbf{Alkida Balliu}
\myaff{Gran Sasso Science Institute, Italy}

\textbf{Thomas Boudier}
\myaff{Gran Sasso Science Institute, Italy}

\textbf{Francesco d'Amore}
\myaff{Gran Sasso Science Institute, Italy}

\textbf{Fabian Kuhn}
\myaff{University of Freiburg, Germany}

\textbf{Dennis Olivetti}
\myaff{Gran Sasso Science Institute, Italy}

\textbf{Gustav Schmid}
\myaff{University of Freiburg, Germany}

\textbf{Jukka Suomela}
\myaff{Aalto University, Finland}

\bigskip
\end{mycover}

\begin{myabstract}
\noindent\textbf{Abstract.}
In this work, we present a fast distributed algorithm for \emph{local potential problems}: these are graph problems where the task is to find a locally optimal solution where no node can unilaterally improve the utility in its local neighborhood by changing its own label. A simple example of such a problem is the task of finding a \emph{locally optimal cut}, i.e., a cut where for each node at least half of its incident edges are cut edges.

The distributed round complexity of the locally optimal cut problem has been wide open; the problem is known to require $\Omega(\log n)$ rounds in the deterministic LOCAL model and $\Omega(\log \log n)$ rounds in the randomized LOCAL model, but the only known upper bound is the trivial brute-force solution of $O(n)$ rounds. Locally optimal cut in constant-degree graphs is perhaps the simplest example of a \emph{locally checkable labeling} problem for which there is still such a large gap between current upper and lower bounds.

We show that in constant-degree graphs, \emph{all} local potential problems, including locally optimal cut, can be solved in $\log^{O(1)} n$ rounds, both in the deterministic and randomized LOCAL models. In particular, the deterministic round complexity of the locally optimal cut problem is now settled to $\log^{\Theta(1)} n$.

Our algorithms also apply to the general case of graphs of maximum degree $\Delta$. For the special case of locally optimal cut, we obtain a randomized algorithm that runs in $O(\Delta^{2} \log^{6} n)$ rounds, which can be derandomized at polylogarithmic cost with standard techniques. Furthermore, we show that a dependence in $\Delta$ is necessary: we prove a lower bound of $\Omega(\min\{\Delta,\sqrt{n}\})$ rounds, even in the quantum-LOCAL model; in particular, there is no polylogarithmic-round algorithm for the general case.
\end{myabstract}
 \thispagestyle{empty}
\setcounter{page}{0}
\newpage

\section{Introduction}

In this work, we present a distributed algorithm that finds a locally optimal cut in any constant-degree graph in a polylogarithmic number of rounds (see \cref{ssec:intro-loc}); this is the first nontrivial distributed algorithm for this problem. Moreover, our algorithm generalizes to all distributed graph problems that are defined with a local potential function (see \cref{ssec:intro-lpp}).

\subsection{Locally optimal cuts}\label{ssec:intro-loc}

In a graph $G = (V,E)$, a \emph{cut} is a labeling $\ell: V \to \{-1,+1\}$, an edge $\{u,v\}$ is called a \emph{cut edge} if $\ell(u) \ne \ell(v)$, and the \emph{size} of the cut is the total number of cut edges. Maximizing the size of the cut is a classic NP-hard problem.

We say that a cut is \emph{locally optimal} if for every node at least half of its incident edges are cut edges. Put otherwise, changing the label of $v$ from $\ell(v)$ to $-\ell(v)$ does not increase the size of the cut. Such a cut always exists and is trivial to find in the centralized setting: start with any labeling, find any node that violates the constraint (we call these \emph{unhappy} nodes) and change its label, and repeat until all nodes are happy. Note that changing the label of an unhappy node $v$ will make node $v$ happy but it might make some of its neighbors unhappy. However, this process has to converge after at most $\abs{E}$ steps, since the size of the cut increases in each step by at least one.

\paragraph{Distributed graph algorithms.}

While it is trivial to find a locally optimal cut with a centralized algorithm, in this work we are interested in \emph{distributed} algorithms for finding a locally optimal cut. We work in the usual LOCAL model \cite{linial-1992-locality-in-distributed-graph-algorithms,peleg-2000-distributed-computing-a-locality-sensitive} of distributed computing: Each node $v \in V$ represents a computer and each edge $\{u,v\}$ represents a communication link between two computers. Initially each node is only aware of its own unique identifier, its degree, and the total number of nodes $n$. Computation proceeds in synchronous rounds; in each round each node can send a message to each neighbor, receive a message from each neighbor, and update its own state. Eventually each node $v$ has to stop and output its own part of the solution, here $\ell(v)$. The \emph{running time} of the algorithm is $T(n)$ if in any $n$-node graph all nodes stop after at most $T(n)$ rounds.

It is good to note that in this model, time and distance are interchangeable: in $T$ rounds all nodes can gather their radius-$T$ neighborhood, and a $T$-round algorithm is a function that maps radius-$T$ neighborhoods to local outputs. In particular, this means that in $O(n)$ rounds we can solve \emph{any} graph problem by brute force; the key question is which problems admit a nontrivial $o(n)$-round algorithm.

\paragraph{Distributed algorithms for locally optimal cuts.}

It is not hard to turn the trivial centralized algorithm into a valid distributed algorithm for solving the locally optimal cut problem. For example, identify all unhappy nodes, then find a maximal independent set of unhappy nodes, and flip all of them in parallel. Unfortunately, this is not going to be an efficient algorithm. Consider the following graph and a cut $\ell$ indicated with the colors white and black; here the node marked with an arrow is the only unhappy node, as only $1$ of its $3$ incident edges is a cut edge:
\begin{center}
    \includegraphics[page=1,scale=0.8]{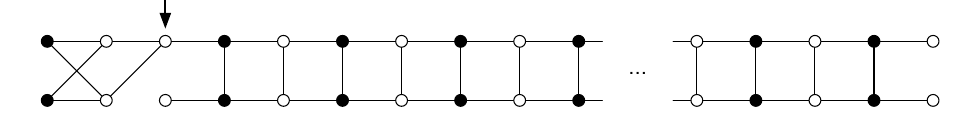}
\end{center}
If we now change the label of this node, we arrive at the following configuration where we again have exactly one unhappy node:
\begin{center}
    \includegraphics[page=2,scale=0.8]{figs.pdf}
\end{center}
Flipping that node, we again move unhappiness around:
\begin{center}
    \includegraphics[page=3,scale=0.8]{figs.pdf}
\end{center}
Clearly, this will take $\Omega(n)$ rounds until all unhappiness is resolved (even if we put aside the time needed to find a maximal independent set), so it is no better than the trivial brute-force solution that solves any graph problem in $O(n)$ rounds. Can we do much better than this?

\paragraph{Broader context: locally checkable labelings.}

The problem of finding a locally optimal cut is non-trivial even when restricting to \emph{bounded-degree graphs}, i.e., when there is some fixed constant $\maxDeg$ and our goal is to solve locally optimal cut in any graph with maximum degree at most $\maxDeg$. We note that e.g.\ in the above apparently-hard instance we had $\maxDeg = 3$, so small values of $\maxDeg$ do not make the problem trivial.

For any fixed $\maxDeg$, the locally optimal cut problem is an example of a \emph{locally checkable labeling} problem (LCL). LCL problems were first defined by \textcite{naor-stockmeyer-1995-what-can-be-computed-locally}, and there is a large body of work that has developed a complexity theory of LCL problems in the LOCAL model
\cite{
        cole-vishkin-1986-deterministic-coin-tossing-with,
        naor-1991-a-lower-bound-on-probabilistic-algorithms-for,
        linial-1992-locality-in-distributed-graph-algorithms,
        naor-stockmeyer-1995-what-can-be-computed-locally,
        brandt-fischer-etal-2016-a-lower-bound-for-the,
        fischer-ghaffari-2017-sublogarithmic-distributed,
        ghaffari-harris-kuhn-2018-on-derandomizing-local,
        balliu-hirvonen-etal-2018-new-classes-of-distributed,
        chang-pettie-2019-a-time-hierarchy-theorem-for-the,
        chang-kopelowitz-pettie-2019-an-exponential-separation,
        rozhon-ghaffari-2020-polylogarithmic-time-deterministic,
        balliu-brandt-etal-2020-how-much-does-randomness-help,
        balliu-brandt-etal-2021-almost-global-problems-in-the,
        dahal-d-amore-etal-2023-brief-announcement-distributed,
        chang-2020-the-complexity-landscape-of-distributed,
        balliu-brandt-etal-2021-almost-global-problems-in-the,grunau-rozhon-brandt-2022-the-landscape-of-distributed,
        suomela-2020-landscape-of-locality-invited-talk%
}.

Most natural LCL problems are by now very well-understood, especially for the deterministic LOCAL model. Indeed, for some restricted families of LCLs and restricted families of graphs, there is now a \emph{complete} classification of the round complexity of all such problems, see e.g.\ \cite{
    balliu-brandt-etal-2020-classification-of-distributed,
    chang-studeny-suomela-2023-distributed-graph-problems,
    brandt-hirvonen-etal-2017-lcl-problems-on-grids,
    balliu-brandt-etal-2022-efficient-classification-of,balliu-brandt-etal-2019-the-distributed-complexity-of%
}.
For example, there is a broad class of \emph{symmetry-breaking problems}, all of which have round complexity $\Theta(\log^* n)$ \cite{cole-vishkin-1986-deterministic-coin-tossing-with,
    goldberg-plotkin-shannon-1988-parallel-symmetry,
    linial-1992-locality-in-distributed-graph-algorithms,
    naor-1991-a-lower-bound-on-probabilistic-algorithms-for%
}, and many other problems \cite{ghaffari-hirvonen-etal-2020-improved-distributed-degree} have turned out to be asymptotically as hard as the \emph{sinkless orientation} problem \cite{brandt-fischer-etal-2016-a-lower-bound-for-the}, which has round complexity $\Theta(\log n)$ in the deterministic LOCAL model \cite{
    brandt-fischer-etal-2016-a-lower-bound-for-the,
    chang-kopelowitz-pettie-2019-an-exponential-separation,
    ghaffari-su-2017-distributed-degree-splitting-edge%
}.

Locally optimal cut is perhaps the most elementary LCL problem whose distributed complexity is currently wide open. We do not even know if it admits any nontrivial algorithm with round complexity $o(n)$, not even for the smallest nontrivial degree bound $\maxDeg = 3$. (The case of $\maxDeg = 2$ is fully understood, as it is in essence equivalent to the weak-$2$-coloring problem, which can be solved in $O(\log^* n)$ rounds, and this is also tight, see e.g.\ \cite{balliu-hirvonen-etal-2019-hardness-of-minimal-symmetry}.)

\paragraph{Prior work.}

What makes locally optimal cuts particularly intriguing is that many commonly-used lower bound techniques fail. In particular, one cannot use existential graph-theoretic arguments (in the spirit of e.g.\ the lower bound for coloring trees in \cite{linial-1992-locality-in-distributed-graph-algorithms} or the lower bound for approximate coloring in \cite{coiteux-roy-d-amore-etal-2024-no-distributed-quantum}), since a locally optimal cut always exists in any graph.

It is, however, possible to construct a nontrivial reduction that connects locally optimal cuts with the sinkless orientation problem \cite{balliu-hirvonen-etal-2019-locality-of-not-so-weak}, and this gives the following lower bounds: we need at least $\Omega(\log n)$ rounds in the deterministic LOCAL model, and at least $\Omega(\log \log n)$ rounds in the randomized LOCAL model.

This is essentially \emph{everything} that is known about locally optimal cuts before this work. There are no upper bounds better than the trivial $O(n)$. Furthermore, the lower bounds do not hold in stronger models such as quantum-LOCAL \cite{gavoille-kosowski-markiewicz-2009-what-can-be-observed,arfaoui-fraigniaud-2014-what-can-be-computed-without,akbari-coiteux-roy-etal-2025-online-locality-meets,balliu-brandt-etal-2025-distributed-quantum-advantage,balliu2026,BalliuCC0ELMOS25}---there we do not even know if the true complexity of the problem is closer to $O(1)$ or $\Omega(n)$.

We note that locally optimal cuts have been studied e.g.\ in the context of descriptive combinatorics also under the names of \emph{unfriendly colorings} and \emph{unfriendly partitions} \cite{shelah-milner-2012-graphs-with-no-unfriendly-partitions,conley-2014-measure-theoretic-unfriendly-colorings,conley-marks-unger-2020-measurable-realizations-of,conley-tamuz-2020-unfriendly-colorings-of-graphs-with}.

\paragraph{Our contribution.}

We present the first nontrivial upper bound for locally optimal cuts. We show that, in bounded-degree graphs, the problem can be solved in $\log^{O(1)} n$ rounds, both in the deterministic and randomized LOCAL models. In particular, the complexity of locally optimal cut in the deterministic LOCAL model is now known to be $\log^{\Theta(1)} n$ rounds.

In the case where the maximum degree $\Delta$ of the graph is not constant, our algorithm requires $O(\maxDeg^{2} \log^{6} n)$ randomized rounds (see \cref{cor:algorithm:locally-optimal-cut} for the main result), which can be derandomized at polylogarithmic cost by using standard techniques \cite{ghaffari-harris-kuhn-2018-on-derandomizing-local,rozhon-ghaffari-2020-polylogarithmic-time-deterministic,ghaffari2024ND}, obtaining a deterministic runtime of $O(\maxDeg^{2} \log^{8}(n) \poly(\log \log n))$ (\cref{cor:analysis:det-complexity-locally-optimal-cut}). Furthermore, we show that $\Omega(\min\{\maxDeg,\sqrt{n}\})$ rounds are necessary, even for randomized and quantum algorithms; in particular, there is no polylogarithmic-round algorithm for the general case (see \cref{thm:locLB} for the result).

\paragraph{Open questions.}

The complexity of the problem is not completely settled yet.
For bounded-degree graphs, the situation is the following:
In the deterministic LOCAL model there is now a polynomial (in $\log n$) gap between the upper and lower bounds. In the randomized LOCAL model, there is still an exponential gap between the upper and lower bounds, and in the quantum-LOCAL model we are lacking any nontrivial lower bounds.
For unbounded-degree graphs, there is still a polynomial gap between the upper and lower bounds, even in the quantum-LOCAL model.

\subsection{Local potential problems}\label{ssec:intro-lpp}

\paragraph{Local potential problems.}

Locally optimal cut is a perfect example of a much more general problem family that we call \emph{local potential problems}. These are constraint satisfaction problems where we have 
\begin{enumerate}[noitemsep]
    \item finitely many possible node labels,
    \item a \emph{local} definition of ``happy'' vs.\ ``unhappy'' nodes,
    \item a \emph{local} potential function that associates a value to each node or edge, and
    \item we can turn ``unhappy'' nodes into ``happy'' nodes with a local change so that we also locally improve the total potential.
\end{enumerate}
For example, in the locally optimal cut problem, our node labels are $\{-1, +1\}$. A node is ``happy'' if at least half of its incident edges are cut edges (this is a local definition in the sense that it only depends on the constant-radius neighborhood of the node). We assign local potential $1$ to each cut edge and potential $0$ to all other edges (this is a local potential in the sense that it only depends on the constant-radius neighborhood of the edge). Now we can always turn unhappy nodes into happy nodes by flipping their label, and this will improve the total potential around the node.

All local potential problems share the commonality that a feasible solution always exists and it is trivial to find with a centralized algorithm that turns unhappy nodes into happy nodes: as the total potential always strictly increases, and there are only finitely many possible solutions, the centralized algorithm will always converge to an all-happy solution.

\paragraph{Locally optimal problems.}

Local potential problems are also closely connected with \emph{local optimization problems}; these are optimization problems where we have
\begin{enumerate}[noitemsep]
    \item finitely many possible node labels and
    \item a \emph{local} utility function that associates a value to each node or edge,
\end{enumerate}
and the task is to maximize the sum of local utilities. For example, the maximum cut problem is such an optimization problem; the local utility function assigns value $1$ to each cut edge and $0$ to all other edges.

Now consider the task of finding a \emph{locally optimal solution} where we cannot improve the total utility by changing the value of any single node. This is exactly what we have in the locally optimal cut problem. Hence it is an example of \emph{finding a locally optimal solution to a local optimization problem}, which is a mouthful; we will use the term \emph{locally optimal problem} for such a task.

It may be helpful to interpret locally optimal problems from a game-theoretic perspective. The goal is to find a \emph{stable} solution or \emph{equilibrium} where no node has an incentive to change its label so that it could improve utility in its local neighborhood.

\paragraph{Reductions and generalizations.}

Assume we have an efficient algorithm for solving locally optimal problems. Then we can also solve local potential problems equally fast. To see this, note that we can simply ignore the notion of happiness first, and interpret local potentials as utilities. We find a solution where no node can unilaterally improve the values of the local potential functions. Then by definition all nodes will also be happy (as unhappy nodes would have a local move that improves potentials). Hence it is sufficient to design efficient algorithms for locally optimal problems.

At first the definition of locally optimal problems may feel very restrictive, as we only consider solutions where no single node can change its own label to improve the potential. In many problems we would like to make more complicated local changes (e.g.\ consider the maximum independent set problem, and the task of finding a locally optimal solution where there is no place where removing one node would let us add two nodes). However, in the bounded-degree setting we can easily encode such problems in our formalism. To see this, assume we have a graph problem with $k$ labels. Then the local label of a node $v$ can be a vector of $k^{\maxDeg+1}$ labels $\ell_0, \dotsc, \ell_\maxDeg$, where $\ell_0$ is the ``base'' label of node $v$, and $\ell_i$ indicates that node $v$ wants to add $\ell_i$ modulo $k$ to the label of its $i$th neighbor. This way changing the \emph{physical} label of one node allows us to change the \emph{logical} labels of all nodes in its radius-$1$ neighborhood. This scheme of course generalizes also to any constant radius. Hence, without loss of generality, it suffices to consider solutions where no single node can unilaterally improve the total utility.

\paragraph{Our contribution.}

In this work we design an algorithm that solves any locally optimal problem in $\log^{O(1)} n$ rounds in bounded-degree graphs in the deterministic LOCAL model. As a corollary, we can also solve any local potential problem in the same time. We also analyze the dependence on the maximum degree $\maxDeg$, so that the results are directly applicable beyond the bounded-degree settings. 
In general, the randomized complexity of our algorithm is \(\maxDeg^{O(1)} (\Lambda / \lambda)^{2} \log^{6} (\Lambda n / \lambda)\), where \(\Lambda\) is the \emph{maximum improvement} associated to the problem, and \(\lambda\) is the \emph{minimum improvement} associated to the problem: \(\Lambda\) and \(\lambda\) are defined by the sole description of the problem and the value of \(\maxDeg\) (see \cref{sec:prelim,sec:alg} for formal definitions, and \cref{thm:algorithm:lop} for the main result). 
Derandomizing the algorithm comes at a polylogarithmic cost, so the deterministic complexity is \(\maxDeg^{O(1)} (\Lambda / \lambda)^{2} \log^{6} (\Lambda n / \lambda) \cdot \tilde{O}(\log^2(n))\) (\cref{cor:analysis:det-complexity}).
In particular, for any constant $\maxDeg$, it holds that \(\Lambda = O(1)\) and \(\lambda = O(1)\), and hence we have an $O(\log^6 n)$-round randomized algorithm and an \(\tilde{O}(\log^ 8 n)\)-round deterministic algorithm.
Our result for locally optimal cuts now follows as a special case of this.

Our work also makes it possible to directly translate many centralized optimization algorithms that are based on local search into nontrivial distributed algorithms. Our work also gives a new, simple way to argue that a given LCL problem $\problem$ can be solved in polylogarithmic time: simply try to come up with a local potential function or a local utility function that allows us to interpret $\problem$ as a local potential problem or as a locally optimal problem.

Conversely, super-polylogarithmic lower bounds on the distributed round complexity now directly imply that an LCL problem \emph{cannot} be represented with the help of a local potential. To give a simple example, stable matchings require a linear number of rounds, even in bounded-degree bipartite graphs \cite{floreen-kaski-etal-2010-almost-stable-matchings-by}, so this problem cannot be represented with local potential functions. Results of this form may be of interest also in a much broader context outside distributed computing.

\paragraph{An application: defective coloring.}

Key examples of local potential problems include \emph{defective coloring} problems, see e.g.\ \cite[Chapter~6]{barenboim-elkin-2013-distributed-graph-coloring}, \cite{chung-pettie-su-2017-distributed-algorithms-for-the,fischer-ghaffari-2017-sublogarithmic-distributed}. Assume that we have $c$ colors and we want to have defect at most $d$, that is, each node can have at most $d$ neighbors with the same color. If $c(d+1) > \maxDeg$, this is a local potential problem. Here the potential function is the number of properly colored edges. If a node is unhappy (it has more than $d$ neighbors with the same color), it can find a color that makes it happy. Our algorithm can be used to solve this problem for any constants $c$, $d$ and $\maxDeg$ with $c(d+1) > \maxDeg$ in a polylogarithmic number of rounds.

\paragraph{Future work and open questions.}

We believe the notions of local potential problems and locally optimal problems will inspire plenty of follow-up work; this can also serve as a bridge between distributed constraint satisfaction problems (e.g.\ LCL) and distributed optimization and approximation problems.

A key open question is understanding the distributed complexity landscape of the family of local potential problems. When \(\maxDeg = O(1)\), our work shows that there are local potential problems with deterministic round complexity $\log^{\Theta(1)} n$, and by prior work there are also local potential problems with round complexity $\Theta(\log^* n)$ and $O(1)$. However, one question that is currently wide open is what happens when we step outside the deterministic LOCAL model. For example, does randomness ever help with any of these problems? The only lower bound we have is the exponentially lower $\Omega(\log\log n)$ for locally optimal cuts; is this the correct complexity? If yes, do all local potential problems admit $\poly(\log\log n)$ randomized algorithms?
Do any of these problems admit a distributed quantum advantage?

Also, questions related to decidability are wide open. For example, given the specification of a local potential problem, is it decidable whether it can be solved in $O(\log^* n)$ deterministic rounds or whether it requires $\Omega(\log n)$ deterministic rounds?
(We remind the reader of the well-known fact that, for LCL problems, there is a gap in the deterministic complexity landscape between $O(\log^* n)$ and $\Omega(\log n)$.)

\subsection{New techniques and key ideas}

\paragraph{Informal intuition.}

Our algorithm is based on the idea that we start with an arbitrary labeling and repeatedly improve it. However, if we only fix nodes that are unhappy, we will end up in a long dependency chain, similar to what we discussed in \cref{ssec:intro-loc}.

Intuitively, the key challenge in our example was that the solution was very fragile---all nodes were \emph{barely happy}, and a minor perturbation in their local neighborhood may cause them to flip, which leads to long cascades that we must avoid.

What our algorithm does differently is that we \emph{push it further, beyond the minimal requirement of happiness}. Our algorithm tries to find local neighborhoods where the value of the potential can be improved by changing \emph{multiple} nodes. For example, the algorithm might notice that the shaded area here admits a modification that improves the total potential a lot:
\begin{center}
    \includegraphics[page=4,scale=0.8]{figs.pdf}
\end{center}
Such a change may at first look counterproductive, as we did not make any nodes happy, yet we increased the number of unhappy nodes. However, if we now start fixing unhappy nodes near the left end of the graph, such a fixing process will not propagate deep inside the shaded region. Intuitively, the nodes in the shaded region are so happy that minor perturbations will not result in catastrophic cascading chains of correction.

\paragraph{Key properties (\cref{sec:analysis}).}

Let us now try to identify a key combinatorial property of locally optimal cuts (and more generally, local potential problems), focusing on bounded-degree graphs for simplicity. Consider first a sequential algorithm that eliminates all unhappy nodes one by one, let $X$ be the set of all nodes that were flipped in this process at least once. Note that in general, $X$ might span the entire input graph.

We prove that for every node $u \in X$, we can find in the $O(\log n)$-radius neighborhood of $u$ a set of nodes $A$ that we can flip so that we get a significant improvement in the total potential. For example, in the above illustration, the changes that we did in the shaded area improved the total potential significantly.

Conversely, this means that if there is no such improving set $A$ in the $O(\log n)$-radius neighborhood of node $u$, we know that we are safe: the sequential flipping procedure will never reach node~$u$.

The goal of the algorithm is to reach a situation in which all nodes can declare that they are safe. Then they are not only currently happy, but also cannot be influenced by any future changes.

\paragraph{Algorithm (\cref{sec:alg}).}

Given the above property, the key challenge is to orchestrate the process of making all nodes safe in a parallel and distributed manner.

Our algorithm uses the MPX clustering algorithm \cite{miller-peng-xu-2013-parallel-graph-decompositions-using, elkin-neiman-2022-distributed-strong-diameter-network} to divide the nodes into polylogarithmic-radius clusters with small boundaries. Each cluster can make nodes far from boundaries safe; we can simply compute for each cluster the result of repeatedly finding for each internal node an improving set $A$, and flipping $A$ until all internal nodes are safe.

The main technical challenge is taking care of nodes close to boundaries that may not be safe and not even happy. To this end, we repeat the randomized clustering process $O(\log n)$ times, to ensure that each node will be w.h.p.\ at least once far from a boundary and hence made safe.

Now we need to ensure that we do not undo what we gained in the previous phase. In particular, consider a node $u$ that was deep inside a cluster in phase $1$, and was made safe, but it is close to a cluster boundary in phase $2$. We need to ensure that improvements in phase $2$ do not tamper with the property that node $u$ was already safe.

Here we make use of the idea that we can have different ``degrees'' of safety: In phase $1$, we are more aggressive in making nodes safe than in phase $2$; put otherwise, in phase $2$ we require a larger relative gain when we are finding improving sets.

This is sufficient to ensure that node $u$ is protected from a chain of improvements in phase $2$. To see this, we consider a sequence of improvements that we did in phase $2$ in the vicinity of node $u$. The union of these strongly improving sets we found in phase $2$ would give a weakly improving set that we would have found already in phase $1$ around node $u$, but as $u$ was made very safe, such a set cannot exist. Hence phase-$2$ improvements will not tamper with phase-$1$ gains.

A careful choice of parameters allows us to push this reasoning throughout all $O(\log n)$ phases, showing that each node is made safe and hence in particular happy during the process, and this property is never lost.

Our algorithm is randomized, but it can be derandomized (with polylogarithmic cost) by using standard techniques \cite{ghaffari-harris-kuhn-2018-on-derandomizing-local,rozhon-ghaffari-2020-polylogarithmic-time-deterministic,ghaffari2024ND}.

\paragraph{Lower bound (\cref{sec:lb}).}

If we parameterize the running time of our algorithm both as a function of maximum degree $\maxDeg$ and the number of nodes $n$, we have an algorithm for locally optimal cuts that runs in $O(\maxDeg^{2} \log^{6} n)$ rounds; in particular, the dependence on $n$ is only polylogarithmic, but the dependence on $\maxDeg$ is polynomial.

In \cref{sec:lb} we prove, through an indistinguishability argument, that there is no (deterministic or randomized) algorithm for locally optimal cuts that solves the problem in time that is only polylogarithmic in both $n$ and $\maxDeg$; we prove a lower bound of $\Omega(\min\{\maxDeg,\sqrt{n}\})$ rounds.

For the high-level idea of the lower bound, observe that in the locally optimal cut problem, the output ($-1$ or $+1$) of a node $v$ is determined if the outputs of all neighbors of $v$ are fixed and if the number of $-1$ and $+1$ outputs among the neighbors is not equal. In this case, $v$ must output the minority value among the neighbors. Hence, if $v$ has odd degree, the output of $v$ is determined by the outputs of the neighbors and thus, if two odd-degree nodes $v$ and $v'$ have the same set of neighbors, then $v$ and $v'$ must have the same output in any valid solution. Based on this observation, we build a path-like lower bound graph for which any valid locally optimal cut solution corresponds to a valid proper $2$-coloring of a path, a problem that is known to have a round complexity lower bound that is linear in the length of the path~\cite{linial-1992-locality-in-distributed-graph-algorithms}.

The construction consists of a sequence of layers such that each layer forms an independent set and consecutive layers are connected by a complete bipartite subgraph. Note that all nodes in the same layer then have the same set of neighbors. The layers are constructed such that all nodes have an odd degree and hence all nodes in the same layer must output the same value. Further, the layers are chosen in such a way that for any two consecutive layers, the outputs must differ. As a consequence, any locally optimal cut solution on the layered graph induces a $2$-coloring of the underlying path. Finally, we make sure that the construction is symmetric enough so that the $2$-coloring can only be produced if the endpoints of the path coordinate with each other. The claimed lower bound then follows because for a sequence of length $k$, we need $\Theta(k^2)$ nodes and a maximum degree $\maxDeg$ of order $\Theta(k)$.

Since the indistinguishability argument has been proved to work also in the quantum-LOCAL model (and, actually, even the stronger non-signaling model \cite{gavoille-kosowski-markiewicz-2009-what-can-be-observed,arfaoui-fraigniaud-2014-what-can-be-computed-without,coiteux-roy-d-amore-etal-2024-no-distributed-quantum}), the same lower bound applies to quantum algorithms as well.

\paragraph{Notes on the definitions (\cref{sec:prelim}).}

So far in the introduction, we have used the locally optimal cut problem as a running example, and there it is natural to relate it with the max-cut problem and use the number of cut edges as the potential or utility function; in particular, unhappy nodes can be flipped to \emph{increase} the potential.
In the technical parts, it is more convenient to use a convention where unhappy nodes can be flipped to \emph{decrease} the potential. Hence the formal definition of potential problems in \cref{sec:prelim} follows this convention.
\section{Preliminaries}\label{sec:prelim}
In this section, we provide some basic graph definitions, and then we define the family of problems that we consider.
\subsection{Notation and definitions}
In this paper, we consider simple undirected graphs unless otherwise stated.
For a graph \(G=(V,E)\), we denote its set of nodes by $V$ and its set of edges by $E$. If $G$ is not clear from the context we use \(V(G)\) to denote the set of nodes of $G$ and \(E(G)\) to denote the set of edges of $G$.
For any two nodes \(u,v\in V(G)\) the distance between \(
u\) and \(v\) is the number of hops in a shortest path between \(u\) and \(v\), and is denoted by \(\dist(u,v)\).
For any node \(v \in V(G)\), the radius-\(r\) neighborhood \(\NN_r[v]\) of \(v\) is the set of all nodes at distance at most \(r\) from \(v\).
Let \(A \subseteq V(G)\) be a subset of nodes, and \(B \subseteq E(G)\) a subset of edges.
By \(\neighborhood_{r}[A]\) we denote the set of all nodes at distance at most \(r\) from any node in \(A\), and by \(\neighborhood_{r}[B]\) we denote the set of all nodes that are endpoints of edges in \(B\) or are at distance at most \(r\) from an endpoint of an edge in \(B\).
Similarly, by \(\neighborhood_r(v)\) we denote the radius-\(r\) (open) neighborhood of \(v\), which is defined simply as \(\neighborhood_r(v) = \neighborhood_r[v] \setminus \{v\}\).
When we omit the subscript \(r\), we assume that \(r = 1\). 
For any two graphs \(G,H\), the intersection between \(G\) and \(H\) is the graph \(G \cap H = (V(G) \cap V(H), E(G) \cap E(H))\), and the union between \(G\) and \(H\) is the graph \(G \cup H = (V(G) \cup V(H), E(G) \cup E(H))\).
The difference between two graphs \(G\) and \(H\) is the graph \(G \setminus H = (V(G), E(G) \setminus E(H))\). 

For any graph \(G = (V,E)\) and any node \(v\in V\), if \(\{v,u\} \in E\) is an edge incident to \(v\), we say that the pair \((v,\{v,u\})\) is a half-edge incident to \(v\). 
We denote the set of all half-edges in the graph \(G\) by \(\HH(G)\). By $[x]$ we denote the set $\{1,\ldots,x\}$.

\begin{definition}[Labeled graph]
	Let \(\VV\) and \(\EE\) be two sets of labels.
	We say that a graph \(G = (V,E)\) is \((\VV,\EE)\)-labeled if there exist two labeling functions \(\lblnodes \colon V \to \VV\) and \(\lbledges \colon \HH(G) \to \EE \) that assign labels to the nodes and half-edges of \(G\), respectively.
    The pair \(\ell = (\lblnodes,\lbledges)\) is said to be the \emph{labeling} of \(G\).
\end{definition}

\paragraph{LOCAL model.}
In this paper, we consider the \local model of distributed computation~\cite{linial-1992-locality-in-distributed-graph-algorithms}.
We are given a distributed network of \(n\) nodes/processors, each of which is capable of arbitrary (but finite) computation.
Nodes are connected by edges which represent communication links, and are able to send messages of arbitrarily large (but finite) size to each other: note that nodes can discriminate between different communication links through \emph{ports}, which are numbered from \(1\) to \(\maxDeg\), where \(\maxDeg\) is the maximum degree of the network, that are adversarially assigned.
Each node is also given a unique identifier in the set \([n^c]\) for some constant \(c \ge 1\), which is, again, adversarially assigned. 
The communication network is represented by a simple undirected graph \(G = (V,E)\) where \(V\) is the set of nodes/processors and \(E\) is the set of edges/communication links.
All nodes run the same algorithm and computation proceeds in synchronous rounds.
In each round, each node sends messages to its neighbors, receives messages from all of its neighbors, and performs local computation.
The input of a problem is given in the form of a \((\VV,\EE)\)-labeled graph \(G = (V,E)\) where \(\VV\) and \(\EE\) are sets of labels.
The output of a node is an assignment of labels to the node itself and/or to its incident half-edges: note that this information can be stored in the node through the use of port numbers.
When a node determines its own output, we say that the node terminates its computation.
The \emph{running time} of a \local algorithm is the number of communication rounds it takes for all nodes to produce their output.
We also refer to the running time of an algorithm as the \emph{locality} of the algorithm.
The \emph{complexity} of a problem is the minimum running time (across all possible algorithms) needed to solve the problem.
We can also assume that nodes are equipped with some source of randomness (independently of other nodes), such as an arbitrarily long random bit string, which they can use during the computation.
In this case, we talk about the \emph{randomized} \local model as opposed to the \emph{deterministic} \local model.
In case of randomized computation, we require that the problem is solved with high probability (w.h.p.), that is, with probability at least \(1 - 1/n\).
We remark that in case of randomized \local, we assume that unique identifiers are generated by the nodes uniformly at random.
This can be done with probability \(1 - 1/n^c\) for arbitrarily large constant \(c\) by generating identifiers of length \(O(\log n)\) bits, and thus does not affect the success probability of an algorithm.

\subsection{Problems of interest}

We first define the notion of \emph{centered graph}, that is, a graph with a distinguished node that we call the \emph{center}.

\begin{definition}[Centered graph]
    Let \(H = (V,E)\) be any graph, and let \(v \in V\) be any node of \(H\).
    The pair \((H,v)\) is said to be a \emph{centered graph} and \(v\) is said to be the \emph{center} of \((H,v)\).
	We say that the pair \((H,v)\) has \emph{radius} \(r\) if all nodes in \(H\) are at distance at most \(r\) from \(v\) and all edges in \(H\) have at least one endpoint at distance at most \(r-1\) from \(v\). 
\end{definition}

\begin{definition}[Labeled centered graph]
	Let \(\VV,\EE\) be two finite sets of labels, and \(r, \maxDeg \in \nats\).
	We denote by \(\LL(\VV,\EE,r,\maxDeg)\) the family of all centered graphs \((H,v)\) such that \(H\) is a \((\VV,\EE)\)-labeled graph of degree at most \(\maxDeg\) and has radius at most \(r\).
\end{definition}

We will focus on problems that ask to produce labelings over graphs that satisfy some constraints in all radius-\(r\) neighborhoods of every node.
Consider a centered graph \((H,v)\). 

\begin{definition}[Set of constraints]\label{def:preliminaries:set-of-constraints}
    Let \(r, \maxDeg \in \nats\) be constants, \(I\) a finite set of integers, and
    \((\VV,\EE)\) a tuple of finite label sets.
    Let \(\CC\) be a finite set of centered graphs \(\{(H_i,v_{i})\}_{i \in I}\) where each \(H_i\) is a \((\VV,\EE)\)-labeled graph of degree at most \(\maxDeg\) and \((H_i,v_{i})\) has radius at most \(r\) for all \(i \in I\).
    Then, \(\CC\) is said to be an \((r,\maxDeg)\)-\emph{set of constraints} over \((\VV,\EE)\).
\end{definition}

In the following, we define what it means to satisfy a set of constraints.
For a \((\VV,\EE)\)-labeled graph \(G = (V,E)\) and a node \(v \in V\), we denote by \(G_r(v)\) the \((\VV,\EE)\)-labeled subgraph of \(G\) that contains all nodes at distance at most \(r\) from \(v\) and all edges that have at least one endpoint at distance at most \(r-1\) from \(v\).

\begin{definition}[Labeled graph satisfying a set of constraints]\label{def:preliminaries:constraint-satisfaction}
    Let \(G = (V,E)\) be a \((\VV,\EE)\)-labeled graph for some finite sets \(\VV\) and \(\EE\) of labels.
    Let \(\CC\) be an \((r,\maxDeg)\)-set of constraints over \((\VV,\EE)\) for some finite \(r,\maxDeg \in \nats\).
    We say that the graph \(G\) satisfies \(\CC\) if, for every node \(v \in V\), it holds that \((G_r(v), v) \in \CC\).
\end{definition}

Before introducing the family of problems considered in this paper, we first define the family of problems called locally checkable labelings.
\begin{definition}[Locally checkable labeling (LCL) problems]\label{def:preliminaries:lcl-problems}
    Let \(r, \maxDeg \in \nats\) be constants.
    A locally checkable labeling (LCL) problem \(\problem\) is a tuple \((\Vin, \Ein, \Vout, \Eout, \CC)\) where \(\Vin\), \(\Ein\), \(\Vout\), and \(\Eout\) are finite label sets and \(\CC\) is an \((r,\maxDeg)\)-set of constraints over \((\Vin \times \Vout, \Ein \times \Eout)\).
\end{definition}

\emph{Solving} an LCL problem \(\problem\) means the following.
We are given as input a \((\Vin,\Ein)\)-labeled graph \(G = (V,E)\).
For any node \(v \in V\) and half-edge \((v,e) \in \HH(G)\), let us denote the input label of \(v\) by \(\inptLbl(v)\), and the input label of \((v,e)\) by \(\inptLbl((v,e))\). We want to produce an output labeling on \(G\) so that we obtain a \((\Vout,\Eout)\)-labeled graph.
Let us denote by \(\outLbl(v)\) and by \(\outLbl((v,e))\) the output labels of \(v\) and \((v,e)\), respectively.
Consider the \((\Vin \times \Vout, \Ein \times \Eout)\)-labeled graph \(G\) where the labeling is defined as follows: the label of each node \(v\) is \(\lblnodes(v) = (\inptLbl(v),\outLbl(v))\), and the label of each half-edge \((v,e)\) is \(\lbledges((v,e)) = (\inptLbl((v,e)),\outLbl((v,e)))\).
Now, the \((\Vin \times \Vout, \Ein \times \Eout)\)-labeled graph \(G\) must satisfy the \((r,\maxDeg)\)-set of constraints \(\CC\) over \((\Vin \times \Vout, \Ein \times \Eout)\) according to \cref{def:preliminaries:constraint-satisfaction}.

With these notions in mind, we can define the family of Locally Optimal Problems (LOPs) that we consider in this paper.

\begin{definition}[Locally Optimal Problems]\label{def:preliminaries:lop}
	A Locally Optimal Problem (LOP) is a pair \((\problem, \pot)\) where:
	\begin{itemize}
		\item \(\problem = (\Vin, \Ein, \Vout, \Eout, \CC)\) is an LCL problem as defined in \Cref{def:preliminaries:lcl-problems}, where \(\CC\) is a \((2r,\maxDeg)\)-set of constraints over \((\Vin \times \Vout, \Ein \times \Eout)\).
		\item \(\pot \colon \LL(\Vin \times \Vout, \Ein \times \Eout, r, \maxDeg) \to \reals_{\ge 0}\) is a function that assigns a non-negative real value to every \((\Vin \times \Vout, \Ein \times \Eout)\)-labeled centered graph \((H,v)\).
	\end{itemize}
    Moreover, \(\CC\) must satisfy the following property, which in essence says that a correct solution does not allow local improvements.
	For any \((\Vin \times \Vout, \Ein \times \Eout)\)-labeled graph \(G = (V,E)\), let us denote by \(\ell = (\lblnodes,\lbledges)\) the labeling of \(G\).
    We say that the potential associated to \((G,\ell)\) is \(\graphPot(G,\ell) = \sum_{v \in V }\pot_\ell(G_r(v),v)\), where \(\pot_\ell(G_r(v),v)\) is the evaluation of \(\pot\) on \((G_r(v),v)\) that is labeled by \(\ell\).
    For any node \(v \in V\), let the cost of $v$ be \(c_v := \pot_\ell(G_r(v),v)\).
    Then, the centered graph \((G_{2r}(v),v)\) labeled with $\ell$ belongs to \(\CC\) if and only if we cannot change the labeling of $v$ to improve the potential of a node in \((G_{2r}(v),v)\). Formally, there is no labeling \(\ell'\) with the following properties, where, for any node \(u\in V\), \(c_u' := \pot_{\ell'}(G_r(u),u)\):
    \begin{itemize}[noitemsep]
        \item Labeling $\ell'$ differs from \(\ell\) only on the \emph{output labels} of \(v\) and of its incident half-edges.
		\item According to labeling $\ell'$, \(G_{2r}(v)\) is still a \((\Vin \times \Vout, \Ein \times \Eout)\)-labeled graph.
		\item It holds that \( c_{v}' < c_{v} \), and that \(\graphPot(G_r(v),\ell') < \graphPot(G_r(v),\ell)\).
	\end{itemize}
\end{definition}
In other words, a centered graph $(G_{2r}(v),v)$ is contained in $\CC$ if and only if there is no alternative labeling for $v$ where:
\begin{itemize}[noitemsep]
    \item $\pot$ gives a strictly better value on node $v$, and
    \item the overall potential is strictly better (observe that the changes on $v$ can only affect nodes within distance $r$ from $v$, and their neighborhood is contained in the radius-$2r$ neighborhood of $v$).
\end{itemize}

The reader might wonder why we do not consider a slightly broader definition of LOPs, where we allow the relabeling to take place on all nodes and half-edges of \(G_r(v)\). In \Cref{def:preliminaries:lop-alternative} we provide such a more general definition, and in \cref{lemma:lop-alternative-to-lop} we will show that, given an algorithm for solving LOPs, we can use such an algorithm to solve problems in this broader family.

\begin{definition}[Generalized Locally Optimal Problems]\label{def:preliminaries:lop-alternative}
	A Generalized Locally Optimal Problem (GLOP) is a pair \((\problem, \pot)\) where:
	\begin{itemize}
		\item \(\problem = (\Vin, \Ein, \Vout, \Eout, \CC)\) is an LCL problem as defined in \Cref{def:preliminaries:lcl-problems}, where \(\CC\) is a \((3r,\maxDeg)\)-set of constraints over \((\Vin \times \Vout, \Ein \times \Eout)\).
		\item \(\pot \colon \LL(\Vin \times \Vout, \Ein \times \Eout, r, \maxDeg) \to \reals_{\ge 0}\) is a function that assigns a non-negative real value to every \((\Vin \times \Vout, \Ein \times \Eout)\)-labeled centered graph \((H,v)\).
	\end{itemize}
    Moreover, \(\CC\) must satisfy the following property.
	For any \((\Vin \times \Vout, \Ein \times \Eout)\)-labeled graph \(G = (V,E)\), let us denote by \(\ell = (\lblnodes,\lbledges)\) the labeling of \(G\).
    We say that the potential associated to \((G,\ell)\) is \(\graphPot(G,\ell) = \sum_{v \in V }\pot_\ell(G_r(v),v)\), where \(\pot_\ell(G_r(v),v)\) is the evaluation of \(\pot\) on \((G_r(v),v)\) that is labeled by \(\ell\).
    The additional requirement that $\CC$ needs to satisfy is that, if the centered graph \((G_{3r}(v),v)\) labeled with $\ell$ \emph{does not belong} to \(\CC\), then there must exist a labeling \(\ell'\) with the following properties.
    \begin{itemize}[noitemsep]
        \item Labeling $\ell'$ differs from \(\ell\) only on the \emph{output labels} of \(G_r(v)\).
        \item According to the labeling $\ell'$, \(G_{3r}(v)\) is still a \((\Vin \times \Vout, \Ein \times \Eout)\)-labeled graph.

        \item The centered graph \((G_{3r}(v),v)\) labeled with $\ell'$ belongs to $\CC$. 
        \item Let $c_v := \pot_\ell(G_r(v),v)$ and $c'_v := \pot_{\ell'}(G_r(v),v)$. It holds that $c'_v < c_v$.
		\item It holds that \(\graphPot(G_{2r}(v),\ell') < \graphPot(G_{2r}(v),\ell)\).
	\end{itemize}
\end{definition}
In other words, if a centered graph $(G_{3r}(v),v)$ labeled with $\ell$ is not contained in $\CC$, then the outputs on $(G_{r}(v),v)$ can be changed such that we obtain a centered graph in $\CC$, and the obtained potential is strictly better, globally and at node $v$. (Observe that the changes on $G_{r}(v)$ can only affect nodes within distance $2r$ from $v$, and their neighborhood is contained in the radius-$3r$ neighborhood of $v$.)

We now prove that we can use an algorithm designed to solve LOPs to solve GLOPs.

\begin{lemma}\label{lemma:lop-alternative-to-lop}
    Consider a GLOP \((\problem, \pot)\) as defined in \Cref{def:preliminaries:lop-alternative}. 
    Then, there exists an LOP \((\problem', \pot')\) as defined in \Cref{def:preliminaries:lop} such that any solution to \((\problem', \pot')\) can be converted into a solution of \((\problem, \pot)\) by an \(O(1)\)-round \local algorithm.
\end{lemma}
\begin{proof}
    Let \((\problem, \pot)\) be a GLOP as defined in \Cref{def:preliminaries:lop-alternative}, where the constraint of $\problem$ contains centered neighborhoods of radius $3r$. In the following, we define \((\problem', \pot')\). We start by defining the labels of $\problem'$.
    \begin{itemize}
        \item The sets of input labels are the same, i.e., \(\Vin' = \Vin\), \(\Ein' = \Ein\).
        \item Let $d_{r,\maxDeg}$ be an upper bound on the number of nodes and half-edges contained in a centered graph of maximum degree $\maxDeg$ and radius $r$. The sets of output labels are as follows. 
        \begin{itemize}
            \item \(\Vout' = \{\bot\}\).
            \item \(\Eout' = [\maxDeg] \times [|\Vout|]^{d_{r,\maxDeg}} \times [|\Eout|]^{d_{r,\maxDeg}}\).
        \end{itemize}
        \item Let $r' := 2r$. The radius of the centered graphs in $\CC'$ is \(2r' = 4r\) and the maximum degree is \(\maxDeg' = \maxDeg\).
    \end{itemize}

    For any \((\Vin' \times \Vout', \Ein' \times \Eout')\)-labeled graph \(G = (V,E)\), let us denote by \(\ell = (\lblnodes,\lbledges)\) the labeling of \(G\). For each node $v$, we now define a function $f$ that takes as input a centered graph $(G_{r}(v),v)$ labeled with $\ell$, and produces a labeling. By abusing the notation, with $f$ we also denote the obtained labeling. The function $f$ will be a partial function, that is, in some cases, the function $f$ will not be able to label $(G_{r}(v),v)$.

    Let $(p_1,X_1,Y_1), \ldots, (p_{\deg(v)},X_{\deg(v)},Y_{\deg(v)})$ be the labels assigned by $\ell$ to the half-edges incident to $v$.
    The values $p_1,\ldots,p_{\deg(v)}$ are called \emph{ports} of $v$.
    If $\exists i,j ~:~ X_i \neq X_j \text{ or }Y_i \neq Y_j$, then $f((G_{r}(v),v))$ is undefined. Let $X = X_1$ and $Y = Y_1$.
    Consider a BFS visit on $G_{r}(v)$ where, the order in which the neighbors of some node $u$ are visited, is according to the ports given as output by $u$. If $u$ has ports with values not in $[\deg(u)]$, or two incident ports with the same value, then $f$ is undefined.
    The labeling $f$ on $(G_{r}(v),v)$ is then obtained by assigning, to the $j$-th visited half-edge, the $j$-th element of $Y$, and to the $j$-th visited node, the $j$-th element of $X$.
    For each node $v$, we say that \emph{$f$ succeeds at $v$} if $f((G_r(v),v))$ is defined.

    We now define a function $g$ that takes as input a centered graph $(G_{2r}(v),v)$ labeled with $\ell$ and returns a labeling for $(G_{r}(v),v)$ that uses the labels of $\problem$. By abusing the notation, with $g$ we also denote the obtained labeling. Again, the function $g$ will be a partial function. For each node $u$ in $(G_{r}(v),v)$, compute $f((G_{r}(u),u))$. If $f$ is undefined in any of the cases in which we applied it, then $g$ is undefined on $(G_{2r}(v),v)$. For each node (resp.\ half-edge) $z$ in $(G_{r}(v),v)$, the labeling $g$ for $\problem$ is defined as follows. Let $a$ be the array containing the output that each node in $(G_{r}(v),v)$ produced for $z$, according to $f$. Then, $g$ labels $z$ with the $(\sum_{x \in a} x \pmod{|\Vout|})$-th label in $\Vout$ (resp.\ $(\sum_{x \in a} x \pmod{|\Eout|})$-th label in $\Eout$), according to some arbitrary but fixed order.
    For each node $v$, we say that \emph{$g$ succeeds at $v$} if $g((G_{2r}(v),v))$ is defined.
   
    The potential function $\pot'$ according to $\ell$ is defined as follows.
  \[
  \pot'_\ell(G_{r'}(v),v):=
  \begin{cases}
  M' & \text{if $f$ does not succeed at $v$,}\\
  M  & \text{if $f$ succeeds at $v$ but $g$ does not succeed at $v$,}\\
  \pot_g(G_r(v),v) & \text{if $g$ succeeds at $v$.}
  \end{cases}
  \]
    where $M'$ and $M$ are arbitrary values strictly larger than any value returned by $\pot$, and such that $M' > M$.

    The constraint $\CC'$ contains exactly those $\ell$-labeled centered graphs $(G_{2r'}(v),v)$ such that there is no labeling $\ell'$, that differs only on the output labels of $v$ and its incident half-edges, that satisfies $\pot'_{\ell'}(G_{r'}(v),v) < \pot'_\ell(G_{r'}(v),v)$ and \(\graphPot(G_{r'}(v),\ell') < \graphPot(G_{r'}(v),\ell)\). This condition ensures that the LOP \((\problem', \pot')\) that we defined satisfies the requirements of \Cref{def:preliminaries:lop}.

    Now, we prove that any solution for $(\problem',\pot')$ can be converted into a solution for $(\problem,\pot)$ in $O(1)$ rounds.

We first observe that, in any valid solution for $(\problem',\pot')$, the potential on each node $v$ must be strictly less than $M'$.
Assume for a contradiction that some node has potential $M'$. Then there exists some node $v$ with some local issue (either an invalid port assignment, or inconsistent arrays assigned to the half-edges of $v$), and $f$ does not succeed at $v$.
Consider the following improvement:
\begin{itemize}
    \item If the ports at $v$ are invalid, $v$ considers some arbitrary order of its incident half-edges, and outputs, for each $1 \le i \le \deg(v)$, $(i,X,Y)$ on its $i$th half-edge, using the same arbitrary arrays $X,Y$ on all incident half-edges. 
    \item If the ports at $v$ are valid, for each $1 \le i \le \deg(v)$, let $e_i$ be the incident half-edge that currently has port value $i$. Node $v$ outputs $(i,X,Y)$ on $e_i$, using the same arbitrary arrays $X,Y$ on all incident half-edges. Note that, with this change, the port assignment remains unchanged.
\end{itemize}
Observe that, in both cases, the potential of $v$, after this change, drops from $M'$ to at most $M$. Moreover, since the ports on $v$ are not changed if already correct, we get that, if before the change $f$
  succeeded at some node $u$ with $\dist(u,v)\le r$, then after the change $f$ still succeeds at $u$ (and its evaluation is unchanged). Hence, if the ports of $v$ were valid before the change, then nodes in $G_{r'}(v)\setminus G_r(v)$ are unaffected by the change, while if the ports of $v$ were invalid, then their potential can only stay the same or improve. Moreover, for all nodes within distance $r$ from $v$, $g$ was undefined before the change (since $f$ was undefined at $v$), and hence had potential at least $M$. After the change, each such node has potential no larger than before. Hence, no node worsens its potential, while $v$ improves from $M'$ to at most $M$, a contradiction.

We now prove that, in any valid solution for $(\problem',\pot')$, the potential on each node $v$ must be strictly less than $M$. We already know that the potential on each node is at most $M$. Suppose for a contradiction that there is some node $v$ with potential exactly $M$. Then, $f$ succeeds at $v$, but $g$ does not. Hence, there exists some node $u$, within distance $r$ from $v$, where $f$ does not succeed. This would imply that $u$ has potential $M' > M$, a contradiction.

The algorithm that converts a solution for $(\problem',\pot')$ into a solution for $(\problem,\pot)$ is defined as follows.
Each node $v$ gathers its $2r$-hop neighborhood and computes $g((G_{2r}(v),v))$, and outputs the label that $g$ assigns to $v$ and its incident half-edges. This algorithm requires $O(1)$ rounds, since $r$ is constant.
Let $\ell'$ be the labeling of the given solution for $(\problem',\pot')$, and let $\ell$ be the labeling for $\problem$ obtained by decoding $\ell'$ with $g$.
Assume for a contradiction that $\ell$ is not a solution
  for $\problem$. Then there exists some node $v$ satisfying that $(G_{3r}(v),v)\notin\CC$. By the definition of GLOP, this means that there exists a relabeling $\hat\ell$ of $G_r(v)$ such that, according to  $\hat\ell$,
  $(G_{3r}(v),v)\in\CC$, 
  $c'_v < c_v$ (where $c_v:=\pot_{\ell}(G_r(v),v)$ and $c'_v:=\pot_{\hat\ell}(G_r(v),v)$), and
  $\graphPot(G_{2r}(v),\hat\ell)<\graphPot(G_{2r}(v),\ell)$.
  Now modify $\ell'$ as follows: change only the output arrays $(X,Y)$ on node $v$, such that $g$ produces exactly the labels of $\hat\ell$ on $G_r(v)$ (this is possible, since the output of $g$ is a modular sum). Let $\ell''$ be the obtained labeling.
  Since all nodes have potential strictly smaller than $M$, $g$ succeeds at every node. Therefore
  $\pot'_{\ell''}(G_{r'}(v),v)=c'_v<c_v=\pot'_{\ell'}(G_{r'}(v),v)$,
  and, because only nodes in $G_{r'}(v)$ are affected and the labeling decoded by $g$ is exactly $\hat\ell$ on $G_r(v)$,
  $\graphPot(G_{r'}(v),\ell'')<\graphPot(G_{r'}(v),\ell')$.
  Hence $v$ would have a valid local improving move in $\problem'$, contradicting that
  $\ell'$ is a solution to $\problem'$.
\end{proof}

The definition of locally checkable labeling problems inherently asks for the maximum degree \(\maxDeg\) of the input graph to be a constant.
However, note that all definitions in this section generalize to the case where \(\maxDeg = \maxDeg(n)\) is a function of \(n\).
In our analysis, we will consider the general case where \(\maxDeg\) is indeed a function of \(n\). %
\section{The algorithm}\label{sec:alg}

\subsection{Preliminary definitions and results}

A key notion in our algorithm is that of an \emph{improving set}.
As already described in the introduction, the main insight of our approach is that if a node $u$ is part of some sequential flipping process that obtains a globally correct solution, we can detect this in $u$'s $O(\log n)$ neighborhood (this only holds for \(\maxDeg = O(1)\), but we will consider the more general case).
What we will find in $u$'s neighborhood is an improving set with a good improving ratio.
We will later prove this formally, so to state our algorithm, we give the necessary definitions here.

\begin{definition}[Improving set]\label{def:algorithm:improving-subgraph}
    Let \((\problem, \pot)\) be an LOP as defined in \Cref{def:preliminaries:lop}.
    Consider any \((\Vin \times \Vout, \Ein \times \Eout)\)-labeled graph \(G = (V,E)\) not necessarily satisfying the constraints of \(\problem\).
    Let \(\ell_1\) denote the labeling of \(G\). For a subset $A\subseteq V$ of nodes, we denote by \emph{relabeling} of $A$ a new output label assignment to the nodes of $A$ and to their incident half-edges.
    Consider any subset $A$ that can be relabeled, obtaining a new labeling \(\ell_2\) of \(G\), such that \(\graphPot(G,\ell_2) < \graphPot(G, \ell_1)\).
    We say that the pair \((A,\ell_2)\) is an \emph{improving set} of \(G\) with respect to \(\ell_1\).
    The \emph{improvement} of \((A,\ell_2)\) with respect to \(\ell_1\) is defined as the difference \(\imp(A,\ell_1,\ell_2):= \graphPot(G,\ell_1) - \graphPot(G,\ell_2)\).
    The \emph{improving ratio} (IR) of \((A,\ell_2)\) with respect to \(\ell_1\) is defined as the ratio \(\impRatio(A,\ell_1, \ell_2) := \imp(A,\ell_1,\ell_2) / \abs{A}\) if $|A|>0$. For an empty set $A$ of nodes we define $\impRatio(A, \ell_1, \ell_2):=0$, for any $\ell_1,\ell_2$.
    By \emph{diameter} of an improving set \((A, \ell_A)\), we denote the \emph{weak diameter} of the subgraph \(G[A]\) induced by \(A\).
\end{definition}

\begin{definition}[Minimal improving set]\label{def:algorithm:minimal-improving-subgraph}
    An improving set \((A,\ell_2)\) is said to be minimal if there is no improving set \((A',\ell_2')\) such that \(A'\subseteq A\) and \(\impRatio(A',\ell_1,\ell_2') > \impRatio(A,\ell_1,\ell_2)\).
\end{definition}

Notice that every improving set \((A,\ell_2)\) contains a minimal improving set, as one can iteratively pick a subset of nodes (possibly, the same set) and a new labeling that increases the improving ratio until no such subset exists.
Note that this process must terminate with at least one node, as the empty graph has improving ratio \(0\). 

From a top-down view, our algorithm does nothing else than relabeling a lot of improving sets.
So the changes that our algorithm does to the labeling can be fully described as a sequence of improving sets.

\begin{definition}[Sequence of \(\beta\)-improving sets]\label{def:algorithm:sequence-of-improving-subgraphs}
    Let \((\problem, \pot)\) be an LOP.
    Consider any \((\Vin \times \Vout, \Ein \times \Eout)\)-labeled graph \(G = (V,E)\) not necessarily satisfying the constraints of \(\problem\), and let \(\ell_0\) denote the labeling of \(G\).
    Let \(\beta > 0\).
    A sequence of \(\beta\)-improving sets in \(G\) w.r.t.\ \(\ell_0\) is a sequence of pairs \((A_1,\ell_1), (A_2,\ell_2), \ldots, (A_k,\ell_k)\) such that:
    \begin{itemize}[noitemsep]
        \item For every \(1 \le i \le k\), the pair \((A_i,\ell_i)\) is a minimal improving set of \(G\) w.r.t.\ the labeling \(\ell_{i-1}\).
        \item For every \(1 \le i \le k\), it holds that \(\impRatio(A_i,\ell_{i-1},\ell_i) \ge \beta\).
    \end{itemize}
\end{definition}

To find improving sets and figure out which improving sets to relabel, our algorithm will first compute a clustering of the graph and then find a good sequence of improving sets in every cluster.
Note that finding such a sequence inside a cluster is easy if the cluster has small diameter, as we can just brute force this computation.

\paragraph{MPX clustering.}
In this paragraph, we define and describe the MPX clustering algorithm, which is a key component of our algorithm, following the paper that first introduced it~\cite{miller-peng-xu-2013-parallel-graph-decompositions-using}.

Let us define the notion of a \((\rho,d)\)-decomposition of a graph \(G = (V,E)\).
\begin{definition}[\((\rho,d)\)-decomposition of a graph]\label{def:preliminaries:decomposition}
    A \((\rho,d)\)-decomposition of a graph \(G = (V,E)\) is a partition of the nodes \(V\) into sets \(C_1, \ldots, C_k\), which we call \emph{clusters}, such that:
    \begin{itemize}[noitemsep]
        \item For every \(i\), the strong diameter of each \(C_i\) is at most \(d\).
        \item The number of edges whose endpoints belong to different clusters is at most \(\rho |E|\).
    \end{itemize}
\end{definition}
Before stating the properties of the MPX clustering algorithm, we report a useful lemma. 

\begin{lemma}[Lemma 4.4 of \cite{miller-peng-xu-2013-parallel-graph-decompositions-using}]\label{lem:MPX-second-smallest}
    Let \(d_1 \le \ldots \le d_n\) be arbitrary values, and let \(\lambda_1, \ldots, \lambda_n\) be independent random variables such that \(\lambda_i \sim \text{Exp}(\beta)\) for each \(i\), for some \(\beta > 0\).
    Then, the probability that the smallest and the second-smallest values of \(d_i - \lambda_i\) are within distance \(c\) of each other is at most \(O(\beta c)\).
\end{lemma}

The MPX clustering algorithm achieves the following result.
\begin{lemma}[\cite{miller-peng-xu-2013-parallel-graph-decompositions-using}]\label{lemma:preliminaries:mpx}
    There exists a randomized \local algorithm that computes a \((\rho,d)\)-decomposition of a graph \(G = (V,E)\) of \(n\) nodes in \(O(\log n / \rho)\) rounds w.h.p., for some \(d = O(\log n / \rho)\).
    Let \(k = 1 / (c \rho)\), where \(c\) is a large enough constant.
    Furthermore, it holds that for each \(v \in V\),  the probability that the radius-\(k\) neighborhood of \(v\) is fully contained in some cluster is at least \(1/2\).
\end{lemma}

\begin{proof}
We briefly show how the MPX algorithm works.
Every node \(u\) of the input graph \(G = (V,E)\) is assigned, independently of other nodes, an exponential random variable \(\lambda_u \sim \text{Exp}(\rho/2)\).
We say that \(\lambda_u\) is the \emph{shift} of \(u\).
The \emph{shifting distance} from a node \(u\) to a node \(v\) is defined as \(\dist_{\text{shift}}(u,v) = \dist(u,v) - \lambda_u\): note that, with probability 1, \(\dist_{\text{shift}}(u,v) \neq \dist_{\text{shift}}(v,u)\).
Each node \(u\) that, for any other node $v\neq u$, satisfies \(\dist_{\text{shift}}(u,u) < \dist_{\text{shift}}(v,u)\), will be the \emph{center} of the cluster \(C_u\). 
We show that at least one such node exists.

Consider any node \(v \in V\), and take the node \(u\) that realizes \(\min_{v' \in V} \dist_{\text{shift}}(v',v)\).
By contradiction, suppose there exists a node \(w\) such that \( \dist_{\text{shift}}(u,u) >  \dist_{\text{shift}}(w,u)\).
Then, we have that \( - \lambda_u > \dist(w,u) - \lambda_w\).
This implies that \( \dist_{\text{shift}}(u,v) = \dist(u,v) - \lambda_u > \dist(u,v) + \dist(u,w) - \lambda_w \ge \dist(w,v) - \lambda_w = \dist_{\text{shift}}(w,v)\) by the triangle inequality.
Hence, \(u\) does not realize the minimum of \(\min_{v' \in V} \dist_{\text{shift}}(v',v)\), reaching a contradiction.

Now, we say that each node \(v\) joins the cluster of the node \(u\) that minimizes the shifting distance from \(u\) to \(v\), that is, \(\min_{v' \in V} \dist_{\text{shift}}(v',v)\).
Let \(d = O(\log n / \rho)\).
In \cite{miller-peng-xu-2013-parallel-graph-decompositions-using}, it is shown that this procedure produces a \((\rho,d)\)-decomposition of \(G\) with high probability.
While the original procedure is defined for the PRAM model, \textcite{elkin-neiman-2022-distributed-strong-diameter-network}
showed that it can be implemented in the \local model through an \(O(\log n / \rho)\)-round randomized \local algorithm.

Now let two nodes $v, v'$, such that $\dist(v,v') \le  k$. 
Let $u_v$ be the center of the cluster that $v$ joins, that is the node realizing \(\min_{u \in V} \dist_{\text{shift}}(u,v)\).
Let \(d_1 \le \ldots \le d_n\) be the distances from \(v\) to the nodes \(u_1, \ldots, u_n \in V\).
Consider \(\lambda_i\) to be the exponential random variable associated to the node \(u_i\).
By \cref{lem:MPX-second-smallest}, we have that, with probability \(O(\rho k / 2)\), it holds that the smallest value and the second-smallest value in the family \(\{(d_i - \lambda_i)\}_{i \in [n]}\) are within distance \(k\) of each other.
Let \(k = 1/(c\rho)\) for a large enough constant \(c\) so that the probability \(O(\rho k / 2)\) is strictly less than \(1/2\).

Hence, \(\dist_{\text{shift}}(u,v') = \dist(u,v') - \lambda_u \le \dist(u,v) + \dist(v,v') - \lambda_u \le k + \dist_{\text{shift}}(u,v)\).
At the same time, for any other \(u' \neq u\), 
\(\dist_{\text{shift}}(u',v') = \dist(u',v') - \lambda_{u'} \ge \dist(u',v) - \dist(v,v') - \lambda_{u'}\ge \dist_{\text{shift}}(u',v) - k\).
Now notice that \(\Pr(\cup_{u' \neq u} \{\abs{\dist_{\text{shift}}(u,v) - \dist_{\text{shift}}(u',v)} \le k\}) < 1/2\) by \cref{lem:MPX-second-smallest}, hence the thesis.
\end{proof}

\subsection{Algorithm description}\label{sec:algorithm:description}
Our algorithm simply repeats the following process $\log n$ times: 
Compute a clustering, such that for every node $v$ with constant probability its entire radius-$O(\maxDeg^{2r} \Lambda^2 \log^4 (\Lambda n / \lambda))$ neighborhood is contained within its cluster, where \(\Lambda\) and \(\lambda\) are later defined in this section, and are related to the description of the specific LOP of interest.
In each cluster we brute-force a maximal sequence of improving sets, each of improving ratio at least $R$.
We then apply the relabeling implied by this sequence in every cluster and inform all nodes of the cluster. 
Informally, we now consider the nodes in these clusters to be "safe".
To ensure that we do not "undo" the progress we have made in previous iterations, we have to increase the minimum improving ratio $R$ between every phase by a small value (roughly $\varepsilon$).
\begin{algorithm}
\caption{LOP algorithm (from the perspective of a node)}\label{alg:lop}
\begin{algorithmic}[1]
\Require Number of nodes \(n\), max-degree \(\maxDeg\), description of the LOP \((\problem, \pot)\).
\Ensure Solution to \((\problem, \pot)\).
\State Initialize \(\lambda \gets \) minimum improvement associated to \((\problem, \pot)\) \label{algo:line:min-improving-ratio}
\State Initialize \(R \gets \lambda/4\) \label{algo:line:initial-IR}
\State Initialize \(\varepsilon \gets \frac{\lambda}{100 c_1 \log n} \) \Comment{\(c_1\) is a constant given by the proof of \cref{thm:algorithm:lop}}\label{algo:line:epsilon}
\State Initialize any output label for the node and its incident half-edges \label{algo:line:initial-labeling}
\Statex \emph{The next loop identifies the ``phases'' of our algorithm. In the following, \(c\), \(c_1\), \(c_2\), and \(c_3\) are large enough constants. The constants \(c_2\) and \(c_3\) are given by \cref{lem:imprBalls,lem:imprChain}, while \(c_1\) is given by the proof of \cref{thm:algorithm:lop}, and \(c\) is given by \cref{lemma:preliminaries:mpx}.}
\For{$i = 1$ to $c_1 \log n$} \label{algo:line:phases} 
    \State Run MPX with \(\rho = \varepsilon^2 /(10 c \cdot c_2 \maxDeg^{2r} \Lambda^2 \log^2 (\Lambda n / \lambda))\) \Comment{\(d = O(\log n / \rho)\) by \cref{lemma:preliminaries:mpx}} \label{algo:line:mpx}
    \State \(C \gets\) cluster of the node \label{algo:line:cluster}
    \State \(H \gets \) graph induced by \(C\) \label{algo:line:induced-graph}
    \State \(\ell \gets \) labeling of \(H\) \label{algo:line:initial-labeling-graph}
    \If{the node is the cluster leader} \label{algo:line:if-leader}
        \State Choose any maximal sequence of \(R\)-improving sets \(\left((A_i,\ell_i)\right)_{i = 1}^k\) in \(H\) w.r.t.\ \( \ell\) such that \label{algo:line:choose-maximal-sequence}
        \State \hspace{\algorithmicindent} 1. \(\NN_{2r+1}[A_i] \subseteq C\) for all \(i \in [k]\) \label{algo:line:safe-neighborhood}
        \State \hspace{\algorithmicindent} 2. The diameter of each $A_i$ is at most $c_3 \maxDeg^{r} \Lambda \log(n)/\varepsilon$ \label{algo:line:diameter} 
        \State Adopt the new labeling according to \(\ell_k\) \label{algo:line:adopt-labeling-leader} 
        \State \(\ell \gets \ell_k\) \label{algo:line:update-labeling-leader}
        \State Broadcast \(\ell_k\) to all nodes in \(C\) \label{algo:line:broadcast-sequence}
    \Else \label{algo:line:not-leader}
        \State Wait for the broadcast by the leader \label{algo:line:wait-broadcast}
        \State Adopt the new labeling according to \(\ell_k\) \label{algo:line:adopt-labeling-not-leader} 
        \State \(\ell \gets \ell_k\) \label{algo:line:update-labeling-not-leader}
    \EndIf \label{algo:line:endif-leader}
    \State Set \(R \gets R + \frac{\lambda}{20 c_1 \log n}\) \label{algo:line:update-IR}
\EndFor \label{algo:line:end-phases}
\State \Return labels of the node and of the incident half-edges \label{algo:line:return-labels}
\end{algorithmic}
\end{algorithm}

We now give a formal description of the algorithm given in \cref{alg:lop}. It solves an LOP \((\problem, \pot)\) as defined in \Cref{def:preliminaries:lop}. 
It takes as input the number of nodes \(n\) and the description of the LOP \((\problem, \pot)\).
In more detail, it gets the full description of \(\problem = (\Vin, \Vout, \Ein, \Eout, \CC)\), where \(\CC\) is an \((2r,\maxDeg)\)-set of constraints over \((\Vin \times \Vout, \Ein \times \Eout)\), and the potential function \(\pot\) that assigns a non-negative real value to every element of \(\LL(\Vin \times \Vout, \Ein \times \Eout, r, \maxDeg)\).

For every labeled centered graph \((H,v_H) \in \LL(\Vin \times \Vout, \Ein \times \Eout, 2r, \maxDeg)\) that does not belong to \(\CC\), there is a way to reassign the labels to \(v_H\) and its incident half-edges so that the potential \(\graphPot(G, \ell)\) decreases and the centered graph \((H,v_H)\) with the new labeling belongs to \(\CC\). 
Let \(\ell'\) be the new labeling: \((\{v_H\}, \ell')\) is a minimal improving set with some positive improvement \(\hat{\lambda}\).
Let \(\lambda\) be the minimum over all possible values of \(\hat{\lambda}\).
We say that \(\lambda\) is the minimum improvement associated to \((\problem, \pot)\) (which is initialized in Line~\ref{algo:line:min-improving-ratio}).
Similarly, let \(\Lambda\) be the maximum possible decrease of the potential over any centered graph \((H,v_H) \in \LL(\Vin \times \Vout, \Ein \times \Eout, 2r, \maxDeg)\).
We say that \(\Lambda\) is the maximum improvement associated to \((\problem, \pot)\).
Note that, given the description of the LOP and the value of \(\maxDeg\), the nodes of the graph can compute the values \(\lambda\) and \(\Lambda\) locally.

The idea of the algorithm is as follows.
Set \(\varepsilon =  \lambda / (100c_1\log(n))\), where \(c_1\) is a large enough constant (Line~\ref{algo:line:epsilon}), and set a threshold improving ratio to be \(R = \lambda/4\) (Line~\ref{algo:line:initial-IR}).
We start with an arbitrary labeling of the nodes and half-edges of \(G\).
Then, we have \(\Theta(\log n)\) phases (from Line~\ref{algo:line:phases} to Line~\ref{algo:line:end-phases}), where in each phase we run the MPX clustering algorithm (Line~\ref{algo:line:mpx}) with parameters \(\rho = \varepsilon^2 / (10 c \cdot c_2 \maxDeg^{2r} \Lambda^2 \log^2 (\Lambda n / \lambda)) = O(\lambda^2 / (\maxDeg^{2r} \Lambda^2 \log^4 (\Lambda n / \lambda)))\) and \(d = O(\log n / \rho)\), where \(c, c_2\) are large enough constants that come from \cref{lemma:preliminaries:mpx,lem:imprChain}, respectively: the bound on \(d\) then holds by \cref{lemma:preliminaries:mpx}.
Each run of MPX takes \(O(d) = O(\maxDeg^{2r} (\Lambda/\lambda)^2 \log^5 (\Lambda n / \lambda))\) rounds.
We assume that each cluster \(C\) given by the MPX algorithm comes with a \emph{leader node} \(v_C\), that can be the center of the cluster itself. 
This anyway can be obtained in \(O(d)\) rounds after each MPX run.
Each leader node looks at the subgraph \(H\) of \(G\) induced by the cluster \(C\) it belongs to, and it computes a maximal sequence of \(R\)-improving sets \((A_i,\ell_i)_{i = 1}^k\) in \(H\) w.r.t.\ \(\ell\).
The sequence satisfies that \(\NN_{2r+1}[A_i] \subseteq C\) for all \(i \in [k]\) (from Line~\ref{algo:line:if-leader} to Line~\ref{algo:line:choose-maximal-sequence}), ensuring that the sequences of different clusters do not interfere with each other. 
Furthermore, we require the diameter of each set in the sequence to be at most \(c_3 \maxDeg^r \Lambda \log n / \varepsilon\) (the choice of the large constant \(c_3\) comes from \cref{lem:imprBalls}).
When we say that the sequence is maximal, we mean that in \(H\) labeled by \(\ell_k\) there is no minimal improving set \((A',\ell')\) such that \(\impRatio(A',\ell_k,\ell') > R\) and \(\NN_{2r+1}[A'] \subseteq H\).
Note that such a sequence must exist and terminate, as the potential of \(H\) is always reduced after every relabeling (and cannot be reduced arbitrarily): furthermore, the sequence can be computed in \(O(d)\) rounds by brute-force.
For this the leader simply collects the topology and labeling of its entire cluster and brute forces every possible sequence in local computation.
Then, \(v_C\) adopts the labeling defined by \(\ell_k\) and also broadcasts \(\ell_k\) to every node of the cluster \(C\) (Line~\ref{algo:line:adopt-labeling-leader}).
Every non-leader node just waits for the leader node to broadcast the new labeling (Lines~\ref{algo:line:not-leader} and \ref{algo:line:wait-broadcast}), and then it adopts the new labeling (Line~\ref{algo:line:adopt-labeling-not-leader}).
Then, the phase terminates and the nodes update the reference value of the improving ratio \(R\) to \(R + \lambda / (20 c_1 \log n)\) (Line~\ref{algo:line:update-IR}).
The algorithm terminates after \(\Theta(\log n)\) phases, and it returns the labels of the nodes and of the incident half-edges (Line~\ref{algo:line:return-labels}).

Note that the round complexity of the algorithm is the running time of each phase, that is, \(O(d)\), multiplied by the number of phases, that is, \(O(\log n)\).
Hence,
\[
    T(n) = O(\log n \cdot d) = O(\log n \cdot \maxDeg^{2r} (\Lambda/\lambda)^2 \log^5 (\Lambda n / \lambda)) = O(\maxDeg^{2r} (\Lambda/\lambda)^2 \log^6 (\Lambda n / \lambda)).
\]

We remind the reader that \(r\) is a constant, while the values \(\lambda, \Lambda\) can depend on the maximum degree \(\maxDeg\) of the input graph: 
if \(\maxDeg = O(1)\), then \(\Lambda = O(1)\) and \(\lambda = O(1)\), disappearing from asymptotic notation.
In the rest of the paper, we prove the following theorem.
\begin{theorem}\label{thm:algorithm:lop}
    Let \((\problem, \pot)\) be a LOP as defined in \Cref{def:preliminaries:lop}.
    The randomized \local algorithm described in \cref{alg:lop} solves \((\problem, \pot)\) in \(O(\maxDeg^{2r} (\Lambda/\lambda)^2 \log^6 (\Lambda n / \lambda))\) rounds w.h.p.
    If \(\maxDeg = O(1)\), then the running time of \cref{alg:lop}  is \(O(\log^6 (n))\) rounds.
\end{theorem}

Note that, for locally optimal cut, it holds that \(r = 1\), \(\lambda = 1\), and \(\Lambda = \maxDeg\).
Since \(\maxDeg \le n\), by \cref{thm:algorithm:lop}, we would obtain a randomized runtime of \(O(\maxDeg^{4} \log^6 (n))\) rounds.
However, as we argue in the appendix, for a natural subclass of LOPs, called Locally Optimal Problems on Edges (LOPEs), we can do better and obtain a better runtime (see \cref{def:appendix:lop-edges,thm:algorithm:lope}).
Note that the proof scheme of \cref{thm:algorithm:lope} is essentially the same as the one of \cref{thm:algorithm:lop} and follows the same lines: it just exploits some additional properties of LOPEs to get better inequalities.
The related result for locally optimal cut is the following.

\begin{corollary}\label{cor:algorithm:locally-optimal-cut}
    The randomized \local algorithm described in \cref{alg:lope} solves the locally optimal cut problem in \(O(\maxDeg^{2} \log^6 (n))\) rounds w.h.p.
\end{corollary}
\begin{proof}
    Just apply \cref{thm:algorithm:lope} to the locally optimal cut problem, which is a special case of LOPE with \(\lambda = 1\).
    Note that there is no need to rescale the potential function (as done in \cref{sec:app:lope}) of the locally optimal cut problem, as \(\Gamma = 1\) already.
\end{proof}

\subsection{Derandomization}

The algorithm described in \cref{alg:lop} can be derandomized using standard techniques.

Let us define the notion of a \((\alpha, d)\)-network decomposition of a graph \(G = (V,E)\).

\begin{definition}[Network decomposition]\label{def:analysis:network-decomposition}
    A \((\alpha,d)\)-network decomposition of a graph \(G = (V,E)\) is a partition of the nodes \(V\) into sets \(C_1, \ldots, C_k\), which we call \emph{clusters}, such that:
    \begin{itemize}[noitemsep]
        \item For every \(i\), the strong diameter of each \(C_i\) is at most \(d\).
        \item The clusters can be colored with \(\alpha\)   colors so that no two adjacent clusters share the same color. 
    \end{itemize}
\end{definition}  

From \cite{ghaffari-kuhn-maus-2017-on-the-complexity-of-local,ghaffari-harris-kuhn-2018-on-derandomizing-local,balliu-brandt-etal-2020-how-much-does-randomness-help,ghaffari2024ND}, the following holds.

\begin{theorem}\label{thm:analysis:derandomization}
    Suppose we are given an \(R(n)\)-rounds randomized \local algorithm \(\AA\) that solves a labeling problem \(\problem\) with high probability.
    Then, there exists a deterministic \local algorithm that solves \(\problem\) that runs in time
    \[
        D(n) = O\left(R(n)\mathrm{ND}(n) + R(n)\alpha(n)d(n)\right),
    \]
    where \(\mathrm{ND}(n)\) is the deterministic complexity of finding an \((\alpha(n), d(n))\)-network decomposition of the input graph.
\end{theorem}

\citeauthor{ghaffari2024ND} proved the following theorem.

\begin{theorem}[\cite{ghaffari2024ND}]\label{thm:analysis:network-decomposition}
    There exists an \(\tilde{O}(\log^2 n)\)-round  deterministic \local algorithm that computes a \((O(\log n), O(\log n))\)-network decomposition of a graph \(G = (V,E)\) of \(n\) nodes.
\end{theorem}

We immediately obtain the following corollary.

\begin{corollary}\label{cor:analysis:det-complexity}
    There exists a deterministic \local algorithm that solves any LOP \((\problem, \pot)\) in \(O(\maxDeg^{2r} (\Lambda/\lambda)^2 \log^6 (\Lambda n / \lambda) \log^2 (n) \poly(\log \log n))\) rounds.
    If \(\maxDeg = O(1)\), then the running time of the deterministic algorithm is \(\tilde{O}(\log^8 (n))\) rounds.
\end{corollary}

As for the specific case of locally optimal cut, by applying the derandomization technique to the algorithm from \cref{cor:algorithm:locally-optimal-cut}, we have the following.

\begin{corollary}\label{cor:analysis:det-complexity-locally-optimal-cut}
    There exists a deterministic \local algorithm that solves the locally optimal cut problem in \(O(\maxDeg^{2} \log^8 (n) \poly(\log \log n))\) rounds.
    If \(\maxDeg = O(1)\), then the running time of the deterministic algorithm is \(\tilde{O}(\log^8 (n))\) rounds.
\end{corollary} %
\section{Analysis of the algorithm}\label{sec:analysis}

In this section, we prove \Cref{thm:algorithm:lop}. We already analyzed the runtime of the algorithm presented in \cref{sec:algorithm:description}. Hence, in the following, we prove that the algorithm produces a proper solution to the given problem.

For the rest of the section, we fix an LOP of interest \((\problem, \pot)\) and any input graph \(G = (V,E)\) of \(n\) nodes that is already \((\Vin \times \Vout, \Ein \times \Eout)\)-labeled (hence it has an input and a -- possibly invalid -- output labeling).
Let \(\ell\) denote the labeling of \(G\).
Recall that \(\lambda\) is the minimum improvement associated to \((\problem, \pot)\), while \(\Lambda\) is the maximum improvement associated to \((\problem, \pot)\).
Trivially, \(\lambda \le \Lambda\).

We fix some standard notation.
\begin{definition}\label{def:best-labeling}
	For any subset of nodes $A \subseteq V$, we denote by $\ell^\star_A$ the best possible relabeling of the nodes of \(A\) and their incident half-edges, that is,  
	\[
		\ell^\star_A = \arg \max_{x \in \mathscr{L}} \{\imp(A,\ell,x)\},
	\]
	breaking ties arbitrarily,
	where $\mathscr{L}$ is the set of all possible relabelings of $G$ such that \(G\) is a \((\Vin \times \Vout, \Ein \times \Eout)\)-labeled graph where the input labels are untouched, and the output labels might change only on the nodes of \(A\) and their incident half-edges.
\end{definition}

The whole analysis of the algorithm relies on the following two lemmas (\cref{lem:phase-terminates,lem:improving-sets-close-to-borders}), the last of which requires involved technical claims.
The first lemma states that each phase terminates.

\begin{lemma}\label{lem:phase-terminates}
	For any phase \(j\) of the algorithm, let \(\ell_j\) be the labeling of the algorithm after finishing phase \(j-1\) (Lines~\ref{algo:line:update-labeling-leader}--\ref{algo:line:update-labeling-not-leader}) and \(R_j\) be the value of \(R\) at the end of phase \(j-1\) (Line~\ref{algo:line:update-IR}).
	Furthermore, let \(C^{(j)}_1, \ldots, C^{(j)}_k\) be the clusters that are the result of the MPX algorithm at phase \(j\) (Line~\ref{algo:line:mpx}),
	Then, for each cluster \(C^{(j)}_i\), any maximal sequence of \(R_j\)-improving sets is finite.
\end{lemma}
\begin{proof}
	Let \(H^{(j)}_i \) be the graph induced by cluster \(C^{(j)}_i\), and \(\ell_{H^{(j)}_i}\) the restriction of \(\ell_j\) to \(H^{(j)}_i\).
	Since \(\graphPot(H^{(j)}_i, \ell_{H^{(j)}_i})\) is a finite value, and each improving set decreases this value by at least some constant \(R_j > 0\), then the sequence cannot have more than \(\graphPot(H^{(j)}_i, \ell_{H^{(j)}_i})/R_j\) elements.
\end{proof}

The second lemma states that, at the end of some phase \(i\), all improving sets are ``close enough'' to the borders of the current and previously computed clusters.
Let us be more formal.

\begin{definition}[Border sets]\label{def:borders}
	For phase $i$ of the algorithm, we define 
	\[
		B_i = \left\{u \in V \st \exists \{u,v\} \in E, C_u \neq C_v\right\}
	\]
	where $C_u$ is the cluster of $u$ and $C_v$ the cluster of $v$ after running MPX in phase $i$ (Line~\ref{algo:line:mpx}). 	
\end{definition}

So $B_i$ is the set of all nodes that have an edge that is crossing different clusters after running MPX at phase \(i\).
The next lemma shows that after phase \(i\), all remaining improving sets of small diameter that have improving ratio at least \(R_i\) (the value of \(R\) initialized at the end of phase \(i-1\) at Line~\ref{algo:line:update-IR}) are close to all borders \(B_j\) for all \(j \le i\).
Note that after the last phase the intersection of all border nodes will be empty with high probability. 
So this lemma implies that at the end of our algorithm, there do not exist any improving sets with small diameter and good improving ratio anymore.

\begin{lemma}\label{lem:improving-sets-close-to-borders}
	At the end of phase \(i\) of the algorithm, all minimal improving sets of diameter at most \(c_3 \maxDeg^r \Lambda \log n / \varepsilon\) and improving ratio at least \(R_i\) (the value of \(R\) initialized at the end of phase \(i-1\) at Line~\ref{algo:line:update-IR}) are fully contained in the set
	\[
		\neighborhood_{t_1}[B_i] \cap \left( \bigcap_{j = 1}^{i-1} \neighborhood_{t_2}[ B_j]\right),
	\]
	for \(t_1 = c_3 \maxDeg^{r} \Lambda \log n / \varepsilon + 10(r+1)\) and \(t_2 = c_2 \maxDeg^{2r} \Lambda^2 \log^2 (\Lambda n/
	\lambda) / \varepsilon^2 + t_1\), where \(c_2\) and \(c_3\) are as in Lines~\ref{algo:line:mpx} and \ref{algo:line:diameter}, respectively.
\end{lemma}

The proof of \cref{lem:improving-sets-close-to-borders} relies on the following two lemmas.
The first one states that for any minimal improving set \(A\) of some improving ratio \(\beta\), for any node \(v \in A\), we can find, within some radius-\(O(\maxDeg^r \Lambda \log n / \varepsilon)\) neighborhood of \(v\), a minimal improving set \(A' \subseteq A\) with improving ratio at least \(\beta - \varepsilon\).
The proof of \cref{lem:imprBalls} is involved, and we defer it to \cref{sec:technical-claims}.

\begin{lemma}\label{lem:imprBalls}
	Let $\ell$ be a labeling of $G = (V,E)$ and let $(S, \ell_S)$ be a $\beta$-minimal improving set, for some $\beta > 0$. 
	Recall that \(\lambda\) is the minimum improvement associated to \((\problem, \pot)\).
	Then, for every $0 < \varepsilon < \min{\{\beta,\lambda\}}$ and any $v \in S$ there exists a large enough constant \(c_3\), a $k \le c_3 \maxDeg^r \Lambda \log n / \varepsilon$, and a subset $A \subseteq S \cap \neighborhood_k[v]$, such that $\impRatio(A, \ell, \ell_{A}^\star) \ge \beta - \varepsilon$.
	Furthermore, \((A,\ell_A^\star)\) is minimal.
\end{lemma}

The second lemma states that in a sequence of \(\beta\)-improving sets with small diameter, before any relabeling, around each node of the sequence there is a minimal improving set $A$ with improving ratio slightly less than \(\beta\). 
Importantly, this $A$ is an improving set w.r.t. the original labeling, before doing any relabeling of the sequence.
If, after the sequence, we have an improving set $A$ in some area which we previously already processed in phase $j$, then this lemma implies that we would have already seen a minimal improving set $A'$ in this same area in phase $j$. 
Because $A'$ will have a slightly worse improving ratio than $A$, we have to increase the improving ratio $R$ between every phase to make this argument useful.
The proof of \cref{lem:what-we-actually-use} is also quite involved, and we defer it to \cref{sec:technical-claims}.

\begin{lemma}\label{lem:what-we-actually-use}
	Let \((A_i, \ell_{i})_{i = 1}^h\) be a sequence of \(\beta\)-improving sets of diameter at most \(c_3 \maxDeg^r \Lambda \log n / \varepsilon\), for some \(\beta \ge \lambda\) and \(0 < \varepsilon < \lambda\), where \(c_3\) is the constant given by \cref{lem:imprBalls}, and \(\lambda\) is the minimum improvement associated to \((\problem, \pot)\).
	Then, for each \(i \in [h]\) and each \(v \in A_i \), there exists a large enough constant \(c_2\) and an improving set \((A,\ell_A)\) of diameter at most \(c_3 \maxDeg^r \Lambda \log n / \varepsilon\), such that \(\impRatio(A, \ell, \ell_A) \ge \beta - 2\varepsilon\), that is fully contained in the set \(\neighborhood_{c_2 \Delta^{2r} \Lambda^2 \log^2 (\Lambda n / \lambda)/\varepsilon^2}[v] \cap (\cup_{j \in [h]} A_j)\).
	Furthermore, \((A,\ell_A)\) is minimal.
\end{lemma}

We first proceed with the proof of \cref{lem:improving-sets-close-to-borders} and then with that of \cref{thm:algorithm:lop}.

\begin{proof}[Proof of \cref{lem:improving-sets-close-to-borders}]
	For any phase \(i > 0\), let \(\CC^{(i)}_{1}, \ldots, \CC^{(i)}_{h_i}\) be the clusters determined by MPX at phase \(i\) (Line~\ref{algo:line:mpx}), let \(\ell_i\) be the labeling of \(G\) at the end of phase \(i\) (Lines~\ref{algo:line:update-labeling-leader} and \ref{algo:line:update-labeling-not-leader}), and let \(R_i\) be the improving ratio initialized at the end of phase \(i - 1\) (Line~\ref{algo:line:update-IR}), with \(R_1 = R\) as in Line~\ref{algo:line:initial-IR}.
	We proceed by induction on the phase \(i\).

	If \(i = 1\), \(R_1 = R\) as initialized at Line~\ref{algo:line:initial-IR}.
		A minimal improving set \((A, \ell_A)\) of diameter at most \(c_3 \maxDeg^r \Lambda \log n / \varepsilon\) with \(\impRatio(A, \ell_1, \ell_A) \ge R_1\) must be such that there is no \(s \in [h_1]\) such that \(\neighborhood_{2r+1}[A] \subseteq C^{(1)}_s\), otherwise we are breaking maximality in Line~\ref{algo:line:choose-maximal-sequence}. 
	Hence, \(\dist(A, B_1) \le 2r\) and we have that \(A \subseteq \neighborhood_{c_3 \maxDeg^r \Lambda \log n / \varepsilon + 2r}[B_1]\), as the diameter of \(A\) is at most \(c_3 \maxDeg^r \Lambda \log n/\varepsilon\).
	Now, assume \(i > 1\) and the statement to be true for all phases \(j = 1, \ldots, i-1\).
	Let \((A,\ell_A)\) be a minimal improving set of diameter at most \(c_3 \maxDeg^r \Lambda \log n / \varepsilon\) such that \(\impRatio(A, \ell_{i}, \ell_A) \ge R_i\).
	For the same reason as in the case \(i = 1\), it holds that \(A \subseteq \neighborhood_{c_3 \maxDeg^r \Lambda \log n / \varepsilon + 2r}[B_i]\).
	Suppose, by contradiction, that \(A\) is not contained in 
	\[
		\bigcap_{j = 1}^{i-1} \neighborhood_{t_2}[ B_j].
	\]
	This implies that there is a \(j^\star \le i-1\) such that \(A\) is not contained in \(\neighborhood_{t_2}[ B_{j^\star}]\).
	Note that the diameter of \(A\) is at most \(c_3 \maxDeg^r \Lambda \log n / \varepsilon\) and at least one node of $A$ is not in  \(\neighborhood_{t_2}[ B_{j^\star}]\).
	As a result, the set \(\neighborhood_{c_2 \maxDeg^{2r} \Lambda^2 \log^2 (\Lambda n/ \lambda) / \varepsilon^2 + 2r + 1}[A]\) is fully contained in some cluster \(C^{(j^\star)}_s\), for some \(s \in [h_{j^\star}]\).

	We can now describe the family of all improving sets that are actually used by our algorithm after phase $j^\star$ (Line~\ref{algo:line:choose-maximal-sequence})  as a single improving sequence.
	We will define an ordering of all minimal improving sets that our algorithm has flipped.
	In order to do that, we first define an ordering of all clusters.
	We say that \(C^{(m_1)}_{s_1} < C^{(m_2)}_{s_2}\) if and only if:
	\begin{itemize}[noitemsep]
		\item \(m_1 < m_2\), or
		\item \(m_1 = m_2\) and \(s_1 < s_2\).
	\end{itemize}
	In each cluster \(C^{(m)}_s\) we also consider the natural ordering of the sets composing the sequence \((S^{(m)}_{s,t},\ell_{S^{(m)}_{s,t}})_{t \in [k_{s}^{(m)}]}\) of \(R_{m}\)-improving sets that were chosen by the leader of \(C^{(m)}_s\) (Line~\ref{algo:line:choose-maximal-sequence}), where \(k_s^{(m)}\) is the number of minimal improving sets that have been flipped inside \(C^{(m)}_s\) at phase \(m\).
	We can create a sequence of \(R_{j^\star + 1}\)-improving sets by concatenating all these maximal sequences 
	\[
		\left(\left((S^{(m)}_{s,t},\ell_{S^{(m)}_{s,t}})_{t \in [k_{s}^{(m)}]}\right)_{s \in [h_m]}\right)_{j^\star + 1 \le m \le i}
	\] 
	and at the end we put the improving set \((A, \ell_A)\).
	Observe that \((A, \ell_A)\) is such that \(\impRatio(A, \ell_{i}, \ell_A) \ge R_i \ge R_{j^\star + 1}\) and hence the whole sequence is an \(R_{j^\star + 1}\)-improving sequence.
	Let us rename this sequence as \((\hat{S}_t)_{t \in [t_i]}\), where \(t_i\) is the overall length of the defined sequence.

	Now, for each \(v \in A\), by \cref{lem:what-we-actually-use} we have that there exists a minimal improving set \((A', \ell_{A'})\) that is fully contained in the set 
	\[
		\neighborhood_{c_2 \maxDeg^{2r} \Lambda^2 \log^2 (\Lambda n/ \lambda) / \varepsilon^2}[v] \cap \left(\cup_{t \in [t_i]} \hat{S}_t\right),
	\]
	and such that \(\impRatio(A', \ell_{j^\star}, \ell_{A'}) \ge R_{j^\star+1} - 2\varepsilon\).
	By \cref{lem:imprBalls}, for each \(v \in A'\), there exists a minimal improving set \((A'', \ell_{A''})\) such that \(A'' \subseteq A' \cap \neighborhood_{c_3 \maxDeg^r \Lambda \log n / \varepsilon}[v]\) and \(\impRatio(A'', \ell_{j^\star}, \ell_{A''}) \ge R_{j^\star+1} - 3\varepsilon \ge  R_{j^\star}\).
	The fact that \[A'' \subseteq \neighborhood_{c_2 \maxDeg^{2r} \Lambda^2 \log^2 (\Lambda n/ \lambda) / \varepsilon^2}[v]\] implies that \(A'' \subseteq \neighborhood_{c_2 \maxDeg^{2r} \Lambda^2 \log^2 (\Lambda n/ \lambda) / \varepsilon^2}[A]\).
	Since \(\neighborhood_{c_2 \maxDeg^{2r} \Lambda^2 \log^2 (\Lambda n/ \lambda) / \varepsilon^2 + 2r + 1}[A]\) is fully contained in some cluster \(C^{(j^\star)}_s\), so it is \(\neighborhood_{2r+1}[A'']\).
	Hence, the sequence \((S^{(j^\star)}_{s,t},\ell_{S^{(j^\star)}_{s,t}})_{t \in [k_{s}^{(j^\star)}]}\) was not maximal, reaching a contradiction.
\end{proof}

With \cref{lem:improving-sets-close-to-borders,lem:phase-terminates} in hand, we can finally prove \cref{thm:algorithm:lop}.

\begin{proof}[Proof of \cref{thm:algorithm:lop}]
	By \cref{lem:phase-terminates}, we have that the algorithm terminates.
	As discussed in \cref{sec:algorithm:description}, the running time of the algorithm is \(O(\maxDeg^{2r} (\Lambda/\lambda)^2 \log^6 (\Lambda n / \lambda))\).
	We just have to argue about correctness.

	Let \(i = c_1 \log n\) be the last phase, and let \(\ell_i\) be the labeling obtained at the end of phase \(i\) (Lines~\ref{algo:line:update-labeling-leader} and \ref{algo:line:update-labeling-not-leader}).
	By contradiction, suppose that there is an error in the graph, that is, a centered graph \((H, v_H)\) of radius \(2r\) that does not belong to \(\CC\).
	This means that $v_H$ can change the output labels of itself and its incident edges to get an improvement of at least $\lambda$. 
	Let $\ell_v$ be that relabeling, so that we can get the improving set $(\{v_H\}, \ell_v)$, with $\IR(\{v_H\},\ell_i,\ell_v) \ge \lambda$.
	Let us now take the minimal version $(\{v_H\}, \ell_v^\star)$ of $(\{v_H\}, \ell_v)$. 

	Since \(R_{c_1 \log n} \le \lambda\) by definition, \cref{lem:improving-sets-close-to-borders} implies that \(\{v_H\}\) is fully contained in 
	\[
		\neighborhood_{t_1}[B_i] \cap \left( \bigcap_{j = 1}^{i-1} \neighborhood_{t_2}[ B_j]\right).
	\]

	Notice that, by \cref{lemma:preliminaries:mpx}, at each phase \(j \le i\), for every node \(v\) there is probability at least \(1/2\) that \(\neighborhood_{t_2 + 1}[v]\) is fully contained within some cluster of the \(j\)-th run of MPX.
	Note that \((1/2)^{c_1\log n} = 1/n^{c_1}\).
	Hence, with probability \(1 - 1/n^{c_1}\), \(\{v_H\}\) is not contained in 
	\[
		\neighborhood_{t_1}[B_i] \cap \left( \bigcap_{j = 1}^{i-1} \neighborhood_{t_2}[ B_j]\right),
	\] 
	reaching a contradiction with \cref{lem:improving-sets-close-to-borders}.

	Now, by the union bound, with probability \(1 - 1/n^{c_1 - 1}\), there is no node that is a center of a centered graph of radius \(2r\) that is invalid according to \(\CC\), yielding the thesis.
\end{proof}

\subsection{Technical claims and proofs}\label{sec:technical-claims}

Here we prove \cref{lem:imprBalls} and \cref{lem:what-we-actually-use}.
The first one basically says that, for any minimal improving set \(A\) of some improving ratio \(\beta\), for any node \(v \in A\), we can find, within some radius-\(O(\maxDeg^r \Lambda \log n / \varepsilon)\) neighborhood of \(v\), a minimal improving set \(A' \subseteq A\) with improving ratio at least \(\beta - \varepsilon\).
In order to prove it, we start with a preliminary lemma.

\begin{lemma}\label{lem:subsetImpr}
	Consider any \(S \subseteq V\), and let \(B \subseteq S\).
	If \(A = S \setminus B\), then
	\[
	\imp(S,\ell,\ell^\star_S) \le \imp(A, \ell, \ell^\star_{A}) + \imp(B, \ell, \ell^\star_{B}) + 2 \cdot \abs{\neighborhood_r[A] \cap \neighborhood_r[B]} \cdot \Lambda.
	\]
	Recall that $\Lambda$ is the maximum improvement associated to \((\problem, \pot)\), and the labeling $\ell^\star_M$ is the best possible relabeling of \(M\).
\end{lemma}
\begin{proof}
	We have that 
	\begin{align*}
		& \imp(S,\ell,\ell^\star_S) = \graphPot(G,\ell) - \graphPot(G,\ell^\star_S) \\
		= \ & \sum_{v\in V} \pot_{\ell}(G_r(v), v) - \pot_{\ell^\star_S}(G_r(v),v) \\
		= \ & \sum_{v\in \neighborhood_r[S]} \pot_{\ell}(G_r(v), v) - \pot_{\ell^\star_S}(G_r(v),v) + \sum_{v \in V(G) \setminus  \neighborhood_r[S]} \pot_{\ell}(G_r(v), v) - \pot_{\ell^\star_S}(G_r(v),v) \\
		= \ & \sum_{v\in  \neighborhood_r[S]} \pot_{\ell}(G_r(v), v) - \pot_{\ell^\star_S}(G_r(v),v),
	\end{align*}
	where the last equality follows from the fact that the local potential of a graph centered at node $v$ of radius \(r\) can only change if any label in the centered graph is modified, and hence $\pot_{\ell}(G_r(v), v) = \pot_{\ell^\star_S}(G_r(v),v)$ for all nodes at distance at least $r+1$ from $S$.

	Let $\ell_A$ be the labeling defined as follows: 
	On the nodes in \(A\) and their incident half-edges, it is equal to \(\ell^\star_S\), and on the rest of the graph it is equal to \(\ell\).
	We similarly define $\ell_{B}$ (we just replace \(B\) with \(A\) in the previous definition).

	We want to analyze the improvements of $(B,\ell_B)$ and $(A,\ell_A)$ compared to that of $(S,\ell^\star_S)$. 
	Consider any node $v$ such that its distance from $A$ (resp.\ \(B\)) is at least \(r+1\).
	The two labelings $\ell_B$ (resp.\ \(\ell_A\)) and $\ell^\star_S$ will yield  the same evaluation of the potential on the centered graph \((G_r(v),v)\). 

	Let 
	\(
		C = \neighborhood_r[A] \cap \neighborhood_r[B]
	\). 
	The following holds:
	\begin{align*}
		& \imp(B,\ell, \ell_B) \\
		= \ & \sum_{v \in \neighborhood_r[B]} \left[\pot_{\ell}(G_r(v),v) - \pot_{\ell_B}(G_r(v),v)\right] \\
		= \ &  \sum_{v \in \neighborhood_r[B] \setminus C} \left[\pot_{\ell}(G_r(v),v) - \pot_{\ell_B}(G_r(v),v)\right] + \sum_{v\in C \cap \neighborhood_r[B]} \left[\pot_{\ell}(G_r(v),v) - \pot_{\ell_B}(G_r(v),v)\right]\\
		= \ & \sum_{v \in \neighborhood_r[B] \setminus C} \left[\pot_{\ell}(G_r(v),v) - \pot_{\ell^\star_S}(G_r(v),v) \right] + \sum_{v\in C \cap \neighborhood_r[B]} \left[\pot_{\ell}(G_r(v),v) - \pot_{\ell_B}(G_r(v),v)\right].
	\end{align*}
	Notice that we can get the exact same equality when replacing $B$ with $A$:
	\begin{align*}
		& \imp(A,\ell, \ell_A)\\
		= \ & \sum_{v\in \neighborhood_r[A] \setminus C} \left[\pot_{\ell}(G_r(v),v) - \pot_{\ell^\star_S}(G_r(v),v) \right] + \sum_{v\in C \cap \neighborhood_r[A]} \left[\pot_{\ell}(G_r(v),v) - \pot_{\ell_{A}}(G_r(v),v)\right].
	\end{align*}
	Bringing those together we recover the improvement for all graphs centered in nodes of $\neighborhood_r[S] \setminus C$ of radius \(r\) when applying $\ell^\star_S$.
	\begin{align*}
		& \imp(A, \ell, \ell^\star_{A}) + \imp(B, \ell, \ell^\star_B) \\
		\ge \ & \imp(A,\ell,\ell_A) + \imp(B, \ell, \ell_B) \\
		= \ & \sum_{v\in \neighborhood_r[A] \setminus C} \left[\pot_{\ell}(G_r(v),v) - \pot_{\ell^\star_S}(G_r(v),v)\right] + \sum_{v\in C \cap \neighborhood_r[A]} \left[\pot_{\ell}(G_r(v),v) - \pot_{\ell_{A}}(G_r(v),v)\right]\\
		+ & \sum_{v \in \neighborhood_r[B] \setminus C} \left[\pot_{\ell}(G_r(v),v) - \pot_{\ell^\star_S}(G_r(v),v)\right] + \sum_{v\in C \cap \neighborhood_r[B]} \left[\pot_{\ell}(G_r(v),v) - \pot_{\ell_B}(G_r(v),v)\right] \\ 
		= \ & \sum_{v \in \neighborhood_r[S] \setminus C} \left[\pot_{\ell}(G_r(v),v) - \pot_{\ell^\star_S}(G_r(v),v)\right] \\
		- & \sum_{v \in C} \left[ \pot_{\ell_{A}}(G_r(v),v) + \pot_{\ell_B}(G_r(v),v) 
		- 2\pot_{\ell}(G_r(v),v) \right] \\ 
		= \ & \imp(S,\ell,\ell^\star_S) - \sum_{v\in C}[\pot_{\ell}(G_r(v), v) - \pot_{\ell^\star_S}(G_r(v),v)] \\
		- & \sum_{v \in C} \left[ \pot_{\ell_{A}}(G_r(v),v) + \pot_{\ell_B}(G_r(v),v) 
		- 2\pot_{\ell}(G_r(v),v) \right]  \\ 
		= \ & \imp(S,\ell,\ell^\star_S) - \sum_{v \in C} \left[\pot_{\ell_{A}}(G_r(v),v) + \pot_{\ell_B}(G_r(v),v) - \pot_{\ell^\star_S}(G_r(v), v) - \pot_{\ell}(G_r(v),v)\right] \\
		\ge \ &  \imp(S,\ell,\ell^\star_S) -2 \sum_{v \in C} \Lambda = \imp(S,\ell,\ell^\star_S) - 2 \abs{C}\Lambda .
	\end{align*}  
	The claim follows by the definition of \(C\) and by rearranging terms.
\end{proof}

With \cref{lem:subsetImpr} in hand, we can now prove \Cref{lem:imprBalls}.

\begin{proof}[Proof of \cref{lem:imprBalls}]
For any subset $A \subseteq S$ we again define $B = S \setminus A$.
Let \(x_S = \impRatio(S, \ell, \ell_S^\star)\). 
Starting with the inequality of \cref{lem:subsetImpr} and using that $(S, \ell^\star_S)$ is minimal, we have
\begin{align*}
	\imp(A,\ell, \ell^\star_A) &\ge \imp(S,\ell, \ell^\star_S) - \imp(B, \ell, \ell^\star_{B}) - 2  \abs{\neighborhood_r[A] \cap \neighborhood_r[B]} \cdot \Lambda\\
			&\ge x_S \abs{S} - x_S \abs{B} - 2 \abs{\neighborhood_r[A] \cap \neighborhood_r[B]}\cdot \Lambda\\
			&= x_S (\abs{S} - \abs{B}) - 2 \abs{\neighborhood_r[A] \cap \neighborhood_r[B]}\cdot \Lambda\\
			&=x_S \abs{A} - 2 \abs{\neighborhood_r[A] \cap \neighborhood_r[B]}\cdot \Lambda .
\end{align*}
So if $\abs{\neighborhood_r[A] \cap \neighborhood_r[B]} \le \frac{\varepsilon}{2 \Lambda } \abs{A}$, we have that $(A, \ell_A^\star)$ has improving ratio at least $\beta -\varepsilon$ as desired. 
Then, we can simply take the minimal improving set \((A',\ell^\star_{A'})\) such that \(A' \subseteq A\) with improving ratio at least \(\beta - \varepsilon\), and, if \(\diam(A')\) is small enough, we are done.

For the sake of readability, we set $y = 2 \Lambda $.
Fix $v\in S$ and consider the following ``candidate sets'':
\( C_i = \neighborhood_i[v]\cap S\) for \(i \ge 0\). 
Let \(n_i = \abs{\neighborhood_r[C_i]}\), and let \(k\) be the smallest integer \(i \ge 2r\) such that \(n_i \le n_{i-2r}(1 + \varepsilon/(y (\maxDeg^{r}+1)))\).

We now prove that such \(k\) must exist and must be at most \(4r y (\maxDeg^{r}+1) \ln n / \varepsilon + 2r + 1\).
It holds that \(n_{k-1} > n_{k - 1 - 2r}(1+\varepsilon/(y (\maxDeg^{r}+1))) > \ldots > (1+\varepsilon/(y (\maxDeg^{r}+1)))^{\floor{{(k-1)}/{2r}}}\).
By contradiction, assume \(k -1 > 4r y (\maxDeg^{r}+1) \ln n / \varepsilon + 2r\).
By the known inequality \((1+x/y)^y \ge \exp[xy / (x+y)]\) for all \(x>0, y>0\), we have that 
\begin{align*}
	& \left(1+\frac{\varepsilon}{y (\maxDeg^{r}+1)}\right)^{\floor{(k-1)/2r}} \\
	= \ & \left(1+\frac{\varepsilon}{y (\maxDeg^{r}+1)}\cdot \frac{\floor{(k-1)/2r}}{\floor{(k-1)/2r}}\right)^{\floor{(k-1)/2r}} \\ 
	\ge \ & \exp\left[\frac{\frac{\varepsilon}{y (\maxDeg^{r}+1)} \cdot \floor{(k-1)/2r}^2}{\frac{\varepsilon}{y (\maxDeg^{r}+1)} \cdot \floor{(k-1)/2r} + \floor{(k-1)/2r}}  \right] \\
	= \ & \exp\left[\frac{\frac{\varepsilon}{y (\maxDeg^{r}+1)} \cdot \floor{(k-1)/2r}}{\frac{\varepsilon}{y (\maxDeg^{r}+1)} + 1}  \right] \\
	\ge \ & \exp\left[\frac{\varepsilon}{2y (\maxDeg^{r}+1)} \cdot \floor{(k-1)/2r}  \right],
\end{align*}
where the last inequality follows from the fact that \(\varepsilon \le y (\maxDeg^{r}+1)\) (as \(\varepsilon < \lambda\), and \(\lambda \le \Lambda\)).
Since \((k-1)/2r > 2y (\maxDeg^{r}+1) \ln n / \varepsilon + 1\), we get that 
\[
	\exp\left[\frac{\varepsilon}{2y (\maxDeg^{r}+1)} \cdot \floor{(k-1)/2r}  \right] \ge \exp[\ln n] = n,
\]
and, hence, \(n_{k-1} > n\), which is a contradiction.

If \(C_k = S\), we are done. 
Assume that \(C_k\) is strictly contained in \(S\) and let \(A = C_k\) and \(B = S \setminus A\).
We need to upper bound the size of \(\neighborhood_r[A] \cap \neighborhood_r[B]\).
We first show that \(\neighborhood_r[A] \cap \neighborhood_r[B] \subseteq \neighborhood_r[C_k] \setminus \neighborhood_r[C_{k-2r}]\), and then we proceed bounding from above the cardinality of \(\neighborhood_r[C_k] \setminus \neighborhood_r[C_{k-2r}]\) (which is, by definition, \(n_k - n_{k-2r}\)).
Let \(u \in \neighborhood_r[A] \cap \neighborhood_r[B]\).
Trivially, \(u \in \neighborhood_r[C_k] = \neighborhood_r[A]\).
Furthermore, since \(u \in \neighborhood_r[B] = \neighborhood_r[S \setminus C_k]\), there exists a node \(w \in S \setminus C_k\) such that \(\dist(u, w) \le r\).
Since \(w \notin C_k\) but \(w \in S\), we have that \(\dist(v, w) \ge k + 1\).
By the triangle inequality, we have that 
\[
	\dist(w,v) \le \dist(u,w) + \dist(u, v),
\]
which implies that
\[	
	\dist(u,v) \ge \dist(w,v) - \dist(u, w) \ge k + 1 - r.
\]
Since all nodes in  \(\neighborhood_r[C_{k-2r}]\) have distance at most \(k - r\) from \(v\), \(u \notin \neighborhood_r[C_{k-2r}]\).
Now, by definition of \(C_k\), we have that \(n_k \le n_{k - 2r}(1+\varepsilon/(y( \maxDeg^{r}+1)))\), which implies that \(n_k - n_{k-2r} \le \varepsilon n_{k - 2r}/(y( \maxDeg^{r}+1))  \le \varepsilon n_k/(y (\maxDeg^{r}+1)) \).
Also, recall that \(n_k = \abs{\neighborhood_r[C_k]} = \abs{\neighborhood_r[A]}\), which implies that \(n_k \le \abs{A} (\maxDeg^r+1)\).
It follows that \(n_k - n_{k-2r} \le \frac{\varepsilon}{y} \abs{A} = \frac{\varepsilon}{2\Lambda} \abs{A}\).
Note that \(A' \subseteq \neighborhood_k[v]\), and \(k \le 4r y (\maxDeg^{r}+1) \ln n / \varepsilon + 2r + 1\).
By taking \(c_3 = 100 r\) we have that the weak diameter of \(A'\) in \(G\) is at most \( c_3 \maxDeg^r \Lambda \log n / \varepsilon\).
Since by assumption \(r\) is a constant, the proof is complete.
\end{proof}

Now we want to prove \cref{lem:what-we-actually-use}.
The key ingredient for \cref{lem:what-we-actually-use} is the next lemma, which states the following: in a sequence of \(\beta\)-improving sets in \(G\) w.r.t.\ some labeling \(\ell\), for any node \(v\) in any of the improving sets, there exists an improving set \((A,\ell_A)\) in the radius-\(O(\maxDeg^{2r} \Lambda^2 \log^2 (\Lambda n/ \lambda) / \varepsilon^2)\) neighborhood of \(v\) (that is fully contained in the union of the improving sets) that has improving ratio at least \(\beta - \varepsilon\) w.r.t.\ \(\ell\).
Now, we can take the minimal improving set \((A',\ell_{A'}^\star)\) with \(A' \subseteq A\) and, by \cref{lem:imprBalls}, we can find a set \(A'' \subseteq A'\) such that \((A'', \ell^\star_{A''})\) is a minimal improving set with improving ratio at least \(\beta - 2\varepsilon\) w.r.t.\ \(\ell\) and diameter at most \(O(\maxDeg^r \Lambda \log n / \varepsilon)\).
Hence, whenever there is a sequence of \(\beta\)-improving sets, before any relabeling, any node in the sequence can locally find small diameter improving sets with improving ratio at least \(\beta - 2\varepsilon\) w.r.t.\ \(\ell\).

\begin{lemma}\label{lem:imprChain}
	Let $0<\varepsilon<\lambda$ be any value, where \(\lambda\) is the minimum improvement associated to \((\problem, \pot)\), and let \(\beta \ge \lambda\).
	Let $((A_i,\ell_i))_{1\leq i \leq h}$ be a sequence of $\beta$-improving sets w.r.t.\ \(\ell\) in \(G\), where each minimal improving set has diameter at most $c_3 \maxDeg^r \Lambda \log n / \varepsilon$ (\(c_3\) as in Line~\ref{algo:line:diameter}).
	Then for all $1 \leq j \leq h$, for all \(v \in A_j\) there exists a large enough constant \(c_2 > 0\) and some $k \le c_2 \Delta^{2r} \Lambda^2 \log^2 (\Lambda n / \lambda) / \varepsilon^2$, such that there exists a set $A \subseteq \neighborhood_k[v] \cap (\bigcup_{1 \leq i \leq h} A_i)$ with $\impRatio(A, \ell, \ell^*_{A}) \ge \beta - \varepsilon $.
\end{lemma}
Before proving \cref{lem:imprChain}, we remark that \(c_2\) (Line~\ref{algo:line:mpx} of the algorithm) is given in the statement of \cref{lem:imprChain}.
\begin{proof}
We will employ an argument that is similar to that of \Cref{lem:subsetImpr}, but this time we will work on an auxiliary graph. 
We define the auxiliary graph $H = (V_H,E_H)$ where \(V_H\) contains a node \(\sigma(A_i)\) for each \(A_i\), \(i \in [h]\), and $\{\sigma(A_i),\sigma(A_j)\} \in E_H$ if and only if \(\dist_G(A_i,A_j) \le 2r\).

The reason for defining the edges like this is the following: if the distance of \(A_i\) and \(A_j\) is at most \(2r\), there are centered graphs of radius \(r\) of which both the relabeling of \(A_i\) and \(A_j\) might affect the potential. 
Conversely, if $\sigma(A_i)$ and $\sigma(A_j)$ are not adjacent in $H$, then no node in \(G\) can be affected by relabelings of both $A_i$ and $A_j$.

We define a weight function $w \colon V_H \to \nats$ for the nodes in $H$ by
$w(x) = \abs{\sigma^{-1}(x)}$,
that is, to any element \(x = \sigma(A_i)\) (for some \(i \in [h]\)), we assign a weight that is equal to the cardinality of \(A_i\).
With an abuse of notation,  we define $w(B) = \sum_{x \in B}w(x)$ for every subset $B$ of $V_H$.

Observe that the improvement of the entire sequence of \(\beta\)-improving sets $((A_i, \ell_i))_{1\le i\le h}$ is at least $\beta \cdot w(V_H)$, due to the definition of sequence of \(\beta\)-improving sets.
Now since the maximum improvement in the potential of any labeling in $G$ is at most $\Lambda \cdot n$, the overall improvement of the sequence itself cannot be larger than $ \Lambda \cdot n$. 
This implies that $ w(V_H) \le \Lambda \cdot n / \beta$.

Fix any \(j \in [h]\), and consider an arbitrary \(v \in A_j\). 
Let \(z = \sigma(A_j) \in V_H\).
We define the ``candidate sets'' with regard to \(z\) as \(C_i = \neighborhood_i[z]\) for \(i \ge 0\).
Let \(k\) be the minimum integer \(i\)  such that \(\frac{\varepsilon}{K} w(C_i) > w(C_{i+1} \setminus C_i) = w(C_{i+1}) - w(C_i)\), for \(K = (\beta - \varepsilon) + (\maxDeg^r + 1)\Lambda \).
We claim that \(k \le \ceil{3K \log (\Lambda n / \beta) / \varepsilon}\).
By contradiction, suppose that \(k > \ceil{3K \log (\Lambda n / \beta) / \varepsilon}\).
Let \(j = \ceil{3K \log (\Lambda n / \beta) / \varepsilon}\).
Then, we have that \(w(C_{i+1}) - w(C_i) \ge \frac{\varepsilon}{K} w(C_i)\) for all \(i \le j\).
Then, 
\[
	w(C_j) \ge (1 + \varepsilon/K) w(C_{j-1}) \ge (1 + \varepsilon/K)^{3K \log (\Lambda n / \beta) / \varepsilon}.
\]
Again, by exploiting that \((1 + x/y)^y \ge \exp[xy/(x+y)]\) for all positive \(x,y\), and that \(\varepsilon < \lambda \le \Lambda\), we get that 
\[
	w(C_j) \ge \exp\left[\frac{3 K \log (\Lambda n / \beta)}{\varepsilon + K}\right] \ge (\Lambda n / \beta)^{3/2}.
\]
This is a contradiction since \(w(C_j) \le w(V_H) \le \Lambda n / \beta\).

Let 
\begin{align*}
	A & = \{i \in [h] \mid A_i = \sigma^{-1}(x) \text{ for some } x \in C_{k+1}\}, \text{ and} \\
	A' & = \{i \in [h] \mid A_i =\sigma^{-1}(x) \text{ for some } x \in C_{k}\}.
\end{align*}
Furthermore, let \(s = \max (A)\). 
We also define \(\hat{A} = \cup_{i \in A} A_i\), and \(\hat{A}' = \cup_{i \in A'} A_i\).
Denote by $\ell_A$ the labeling obtained by setting $\ell_A = \ell_{s}$ for all nodes in $\hat{A}$ and their incident half-edges, and $\ell_A = \ell$ for all other nodes and half-edges.

Let \(i_1, \ldots, i_t\) denote the elements of \(A\), where \(i_1 < \ldots < i_t = s\).
We define $f_{i_0} = \ell$ and then a new  relabeling $f_{i_j}$ as the labeling that outputs $\ell_{i_{j}}$ on all the nodes in $A_{i_j}$ and their incident half-edges, but outputs $f_{i_{j-1}}$ everywhere else.
Now we can obtain $\ell_A$ by sequentially applying the relabelings $f_{i_1}, f_{i_2}, \ldots, f_{i_t} = \ell_A$.

Importantly, for each \(j \in [t]\), by the definition of \(f_{i_j}\) and \(E_H\), the labels of all nodes in \(\neighborhood_{2r}[\hat{A}']\) and their incident half-edges given by \(f_{i_j}\) are exactly those given by \(\ell_{i_j}\): in fact, if the labeling of a node \(u \in \neighborhood_{2r}[\hat{A}']\) (or its incident half-edges) is modified by \(\ell_{i_j}\), then \(u \in A_{i_j}\) and \(\dist_G(A_{i_j},\hat{A}') \le 2r\), implying that \(i_j \in A\).

As a result, if $i_j \in A'$, then
\[
	\imp(A_{i_j}, f_{i_{j-1}}, f_{i_j}) = \imp(A_{i_j}, \ell_{i_{j -1}}, \ell_{i_j}) \ge \imp(A_{i_j}, \ell_{i_{j} - 1}, \ell_{i_j}) \ge \beta \cdot \abs{A_{i_j}}.
\]  
On the other hand, if $i_j \in A \setminus A'$, any relabeling of \(A_{i_j}\) can worsen the potential by at most \(\Lambda (\maxDeg^r+1) \abs{A_{i_j}}\), since \(\abs{\neighborhood_r[A_{i_j}]} \le (\maxDeg^r+1) \abs{A_{i_j}}\). 
We obtain the following:
\begin{align*}
	\imp(\hat{A},\ell, \ell_A) = \ & \sum_{j=1}^{t} \imp(A_{i_j}, f_{i_{j-1}}, f_{i_j}) \\
	= \ & \sum_{i_j \in A'} \imp(A_{i_j}, f_{i_{j-1}}, f_{i_j}) + \sum_{i_j \in A\setminus A'} \imp(A_{i_j}, f_{i_{j-1}}, f_{i_j})\\
	\ge \ & \sum_{i_j \in A'} \imp(A_{i_j}, \ell_{i_j -1}, \ell_{i_j}) - \sum_{i_j \in A\setminus A'} (\maxDeg^r+1) \cdot \Lambda \cdot  \abs{A_{i_j}}\\
	\ge \ & \beta \cdot w(C_{k}) - (\maxDeg^r+1) \cdot \Lambda \cdot (w(C_{k+1}) - w(C_{k}))\\
	\ge \ & \beta\cdot  w(C_{k}) -   (\maxDeg^{r}+1) \cdot \Lambda \cdot \frac{\varepsilon}{K} w(C_{k})  \\
	= \ & w(C_{k}) \left(\beta - \frac{\varepsilon (\maxDeg^r+1) \Lambda}{K}\right).
\end{align*}
We now show that, by the definition of \(K\), it holds that
\[
	w(C_{k}) \left(\beta - \frac{\varepsilon (\maxDeg^r+1) \Lambda}{K}\right) \ge w(C_{k + 1}) \left(\beta - \varepsilon\right).
\]
Recall that \(w(C_{k + 1}) \le w(C_{k})(1 + \varepsilon/K)\).
Hence,
\begin{align*}
	w(C_{k}) \left(\beta - \frac{\varepsilon (\maxDeg^r+1) \Lambda}{K}\right) & \ge w(C_{k+1}) \cdot \frac{K}{K+\varepsilon} \cdot \left(\beta - \frac{\varepsilon (\maxDeg^r+1) \Lambda}{K}\right) \\
	& = w(C_{k+1}) \cdot \frac{\beta K - \varepsilon(K - (\beta - \varepsilon))}{K+\varepsilon} \\
	& = w(C_{k+1}) \cdot \frac{K(\beta - \varepsilon) + \varepsilon(\beta - \varepsilon)}{K+\varepsilon} \\
	& = w(C_{k+1}) \cdot \left(\beta - \varepsilon\right) \cdot \frac{K + \varepsilon}{K + \varepsilon} \\
	& = w(C_{k+1}) \cdot \left(\beta - \varepsilon\right).
\end{align*}

Note that \(C_{k+1} = \neighborhood_{k+1}[z]\).
Thus, the weak diameter of \(\hat{A}\) in \(G\) is at most \( (2k+1)\cdot(2r + c_3 \maxDeg^r \Lambda \log n / \varepsilon)\).
In particular, since \(k \le \ceil{3K \log (\Lambda n / \beta) / \varepsilon}\) and \(K \le \beta + (\maxDeg^r+1)\Lambda\), we have that the weak diameter of \(\hat{A}\) in \(G\) is at most \(100 ((\beta + \maxDeg^r \Lambda) \cdot \log (\Lambda n / \beta) / \varepsilon + 1)(r + c_3 \maxDeg^r \Lambda\log n / \varepsilon)\).
Since by assumption \(r\) is constant and \(\lambda \le \beta \le \Lambda\), we have that 
\[
	\hat{A} \subseteq \neighborhood_{c_2 \Delta^{2r} \Lambda^2 \log^2 (\Lambda n / \lambda) / \varepsilon^2}[v],
\]
for some large enough constant \(c_2 > 0\).
By noticing that \(\ell^\star_{\hat{A}}\) can only increase the improving ratio w.r.t.\ \(\ell_{{A}}\), we conclude the proof of \Cref{lem:imprChain}.
\end{proof}

We conclude with the proof of \cref{lem:what-we-actually-use}.

\begin{proof}[Proof of \cref{lem:what-we-actually-use}]
	We first apply \cref{lem:imprChain}, and we obtain an improving set \((A, \ell_A)\)   with improving ratio at least \(\beta - \varepsilon\) that is fully contained in \(\neighborhood_{c_2 \Delta^{2r} \Lambda^2 \log^2 (\Lambda n / \lambda) /\varepsilon^2}[v] \cap (\cup_{i \in [h]} A_i)\).
	Then, we take the minimal version of it, that is, \((A',\ell^\star_{A'})\), with \(A' \subseteq A\), and we apply \cref{lem:imprBalls}, which yields the thesis.
\end{proof}

\section{A lower bound for locally optimal cut}\label{sec:lb}

In this section we prove a lower bound of  $\Omega(\min\{\maxDeg,\sqrt{n}\})$  rounds for any algorithm that finds a locally optimal cut. 
Naturally this implies that any algorithm for generic local potential problems must also take at least $\Omega(\min\{\maxDeg,\sqrt{n}\})$ rounds.

\begin{definition}[Locally optimal cut]\label{def:lb:loc}
    Let $G = (V, E)$ be a graph.  
    A \emph{locally optimal cut} is an assignment of labels 
    $\ell : V \to \{-1, 1\}$ such that, for every node $v \in V$,  
    the number of neighbors of $v$ with a label different from $\ell(v)$  
    is at least as large as the number of neighbors that share $\ell(v)$, i.e.,
    \[
        \abs{\{\, u \in N(v) : \ell(u) \neq \ell(v) \,\}}
        \;\geq\;
        \abs{\{\, u \in N(v) : \ell(u) = \ell(v) \,\}}.
    \]
\end{definition}

We prove the following lower bound.

\begin{theorem}\label{thm:locLB}
    For any large enough \(n \in \nats\) and \(\maxDeg < n\), let \(\varepsilon = 1 - (\maxDeg / n)\).
    Any algorithm that finds a locally optimal cut in the randomized LOCAL model with probability $p \ge \frac{1}{2}$ on graphs of \(n\) nodes and maximum degree \(\maxDeg\) has locality $\Omega(\min\{\maxDeg,\sqrt{\varepsilon n}\})$.
    In particular, for any \(\maxDeg \le c n\) for some constant \(c < 1\), the locality is \(\Omega(\min\{\maxDeg,\sqrt{n}\})\).
\end{theorem} 

The proof idea is the following: We construct a graph that contains a component of diameter \(\Theta(\min\{\maxDeg, \sqrt{\varepsilon n}\})\). 
Such component contains two nodes \(u\) and \(v\) at distance \(\Theta(\min\{\maxDeg, \sqrt{\varepsilon n}\})\) satisfying the following property: in any valid solution to the locally optimal cut problem, the nodes \(u\) and \(v\) must have different labels.

The following lemma will be a key ingredient in our lower bound construction.

\begin{lemma}\label{lem:monoOutputs}
Let $G = (V, E)$ be a graph, and let $u, v \in V$ be two nodes of odd degree that share the same open neighborhood, i.e., $\neighborhood(u) = \neighborhood(v)$. Then, in any valid solution to the locally optimal cut problem, it must hold that $u$ and $v$ have the same label.
\end{lemma}

\begin{proof}
    Let $\ell$ be any valid labeling.  
    For a node $x \in V$, denote by 
    \[
        P_x = \{\, y \in \neighborhood(x) \mid \ell(y) = 1 \,\}
        \quad \text{and} \quad
        N_x = \{\, y \in \neighborhood(x) \mid \ell(y) = -1 \,\}
    \]
    the sets of neighbors of $x$ labeled $1$ and $-1$, respectively.  
    Since $u$ has odd degree, exactly one of the sets $P_u$ or $N_u$ must be larger.  
    By the definition of a locally optimal cut, node $u$ must be labeled $1$ if $|N_u| > |P_u|$, and $-1$ if $|P_u| > |N_u|$.  
    The same reasoning applies to $v$.  
    Because $\neighborhood(u) = \neighborhood(v)$, we have $P_u = P_v$ and $N_u = N_v$. 
    Hence, both $u$ and $v$ must have the same label.
\end{proof}

\subsection{The lower bound graph}

We build our lower bound graphs by merging two separate copies of the following construction.
For an integer parameter $k \ge 2$, we define a graph $H_k = (V, E)$ as follows.

\paragraph{Node set.}
The vertex set $V$ is partitioned into disjoint subsets, which we call \emph{layers}, arranged in the sequence
\[
S, L_2, R_2, L_3, R_3, \ldots, L_k, R_k.
\]
Here, $S$ is a singleton set ($|S| = 1$), and for each $i \in \{2, \ldots, k\}$, the sets $L_i$ and $R_i$ have cardinality $|L_i| = |R_i| = i$.

\paragraph{Edge set.}
The edge set $E$ is defined so that each pair of consecutive sets in the sequence
\[
S, L_2, R_2, L_3, R_3, \ldots, L_k, R_k
\]
forms a complete bipartite graph. 
Namely, for each $i \in \{2, \ldots, k\}$, all nodes in $L_i$ are connected to all nodes in $R_{i}$, and all nodes in $R_{i}$ are connected to all nodes in $L_{i+1}$ (if $i < k$).
As for the only node in \(S\), it is connected to all nodes in \(L_2\).

Hence, the sequence of layers induces the following pattern of complete bipartite graphs:
\[
K_{1,2},\; K_{2,2},\; K_{2,3},\; K_{3,3},\; K_{3,4},\; \ldots,\; K_{k-1,k-1},\; K_{k-1,k},\; K_{k,k}.
\]

\paragraph{Properties.} 
Intuitively, the graph can be seen as a ``path'' of layers whose sizes increase monotonically as
\[
1, 2, 2, 3, 3, 4, 4, \ldots, k-1, k-1, k, k,
\]
where adjacent layers are fully connected. 
Note that all nodes in a specific layer, except the layer $S$ and (possibly) the last layer, have an odd degree and share the same neighborhood.
The following lemma shows that the labels of nodes in any given layer are all the same, while the labels of nodes of consecutive layers must differ.

\begin{lemma}\label{lem:coloringHalf}
    Consider \(H_k\) to be the input graph.
    In any solution to the locally optimal cut problem, all nodes in any layer of \(H_k\) share the same label, while nodes in consecutive layers of \(H_k\) use different labels.
\end{lemma}
\begin{proof}
    Any node in any layer \(L_i\), \(2 \le i \le k\), and in any layer \(R_j\), \(2 \le j \le k-1\), has odd degree.
    Hence, by \cref{lem:monoOutputs}, all nodes in $L_i$ must share the same label for \(2 \le i \le k\) (the same holds in \(R_j\) for \(2 \le j \le k-1\)), implying that also nodes in $R_k$ must share the same label. 

    We now prove the second part of the statement.
    It is straightforward to notice that the node in $S$ needs to pick a label different from the nodes in $L_2$. 
    Similarly, the nodes in $R_k$ must have the opposite label of the nodes in $L_k$. 
    Now fix any three consecutive layers in the sequence
    \[
    S, L_2, R_2, L_3, R_3, \ldots, L_k, R_k
    \]
    We denote these consecutive layers as $A_{i-1}$, $A_i$, and $A_{i+1}$.
    Each node in $A_i$ is adjacent to all nodes in $A_{i-1}$ and $A_{i+1}$, and to no other nodes.  
    Because $|A_{i+1}| > |A_{i-1}|$, the majority of $A_i$'s neighbors are in $A_{i+1}$.  

    By contradiction, assume that the nodes in $A_i$ and the nodes in $A_{i+1}$ use the same label and, w.l.o.g., let that label be $+1$.  
    Then every node in $A_i$ has strictly more neighbors with label $+1$ (from $A_{i+1}$) than it could have with label $-1$ (from $A_{i-1}$).  
    This violates the locally optimal cut condition and, hence, the nodes in $A_i$ and those in $A_{i+1}$ cannot use the same label.
    We just proved that nodes in consecutive layers must be labeled differently.
\end{proof}

We now complete the lower bound construction by ``mirroring'' the graph $H_k$.

\paragraph{Mirrored construction.}
We consider a graph \(G_k\) constructed as follows.
Consider a copy \(H_k^{(1)}\) of \(H_k\) and a copy \(H_{k-1}^{(2)}\) of \(H_{k-1}\), where
\[
    S^{(1)}, L_2^{(1)}, R_2^{(1)}, L_3^{(1)}, R_3^{(1)}, \ldots, L_k^{(1)}, R_k^{(1)}
\]
are the layers of \(H_k^{(1)}\), and 
\[
    S^{(2)}, L_2^{(2)}, R_2^{(2)}, L_3^{(2)}, R_3^{(2)}, \ldots, L_{k-1}^{(2)}, R_{k-1}^{(2)}
\]
are the layers of \(H_{k-1}^{(2)}\).
Now, we connect nodes in \(R_k^{(1)}\) to all nodes in \(R_{k-1}^{(2)}\) to form a complete bipartite graph.
The structure of layers of \(G_k\) is then
\[
    S^{(1)}, L_2^{(1)}, R_2^{(1)}, L_3^{(1)}, R_3^{(1)}, \ldots, L_k^{(1)}, R_k^{(1)}, R_{k-1}^{(2)}, L_{k-1}^{(2)}, \ldots, R_3^{(2)}, L_3^{(2)}, R_2^{(2)}, L_2^{(2)}, S^{(2)},
\]
and the sizes of layers are as follows:
\[
1, 2, 2, 3, 3, 4, 4, \ldots, k-1, k-1, k, k, k-1, k-1, \ldots, 3, 3, 2, 2, 1.
\]
Note that the nodes in \(S^{(1)},S^{(2)}\) are at odd distance.
If we prove that, with this new construction, the layers \(L_k^{(1)}, R_k^{(1)}, R_{k-1}^{(2)}, L_{k-1}^{(2)}\) are  still monochromatic and must alternate colors, then we are asking the nodes in \(S^{(1)},S^{(2)}\) to have different labels.

\begin{lemma}\label{lem:lb_properties}
    For any $k\ge 2$, the graph $G_k$ has the following properties:
    \begin{enumerate}[label=(\roman*),noitemsep]
        \item The number of nodes in $G_k$ is $2 k^2 -2$.\label{prop:nodes}
        \item The maximum degree of \(G_k\) is \(2k-1\).\label{prop:degree}
        \item In any correct solution to the locally optimal cut problem, each layer is monochromatic and adjacent layers use different labels. \label{prop:monochromatic}
    \end{enumerate}
\end{lemma}
\begin{proof}
    Consider  claim (i).
    There are 2 layers with 1 node (the first and last), 2 layers with $k$ nodes (in the middle) and 4 layers of any other cardinality in \(\{2,\ldots,k-1\}\).
    Therefore, the following holds:
    \[
        \abs{V(G)} = 2 + 2k + 4\sum_{i=2}^{k-1}i = 4 \sum_{i=1}^{k}i - 2k -2= 4\frac{(k+1)k}{2} -2k  -2= 2k^2 -2.
    \]
    We now prove claim (ii).
    The maximum degree of \(G_k\) is reached at all nodes in \(L_k^{(1)}\) and in \(R_k^{(1)}\), which trivially is \(2k-1\).
    Finally, we prove claim (iii). 
    By \Cref{lem:coloringHalf}, we know that the layers $S^{(1)}, \ldots, L_k^{(1)}$ are monochromatic and must use alternating labels, and the same is true for the layers $R_k^{(1)}, R_{k-1}^{(2)}, L_{k-1}^{(2)}, \ldots, S^{(2)}$ (by symmetry). 
    We just need to show that the layers $L_k^{(1)}$ and $R_k^{(1)}$ use different labels. 
    Since $\abs{R_{k-1}^{(2)}} < \abs{R_k^{(1)}}$, the strict majority of the neighbors of the nodes in \(R_k^{(1)}\) are in $L_k^{(1)}$ and, hence, the layers $L_k^{(1)}$ and $R_k^{(1)}$ must use different labels. 
\end{proof}

We are ready to prove \Cref{thm:locLB} through an indistinguishability argument.

\begin{proof}[Proof of \Cref{thm:locLB}]
    Let $\AA$ be a $T(\maxDeg,n)$-round randomized LOCAL algorithm that solves locally optimal cut on graphs of \(n\) nodes and maximum degree \(\maxDeg\), with \(\maxDeg \le n(1-\varepsilon)\).
    We will show that, for all large enough values of \(n\), and for all large enough values of \(\maxDeg \le n(1-\varepsilon)\), if \(
    \AA\) has success probability at least \(1/2\), then it must hold that $T(\maxDeg,n) > \floor{\min\{(\maxDeg+1)/2, \sqrt{\varepsilon n}/2\}} / 2$.

    Let \(k = \floor{\min\{(\maxDeg+1)/2, \sqrt{\varepsilon n}/2\}}\).
    The lower bound graph \(G\) is constructed as follows.
    Take \(G_k\), and consider two paths \(P_1\) and \(P_2\) of length \(\floor{\varepsilon n / 4}\) and \(n - (\floor{\varepsilon n}/4 + \abs{V(G_k)} + \maxDeg - 1)\), respectively.
    By Property~\ref{prop:degree} of \cref{lem:lb_properties}, the maximum degree of \(G_k\) is at most \(2k - 1 \le \maxDeg\).
    Recall that \(G_k\) has exactly two nodes of degree \(2\): let them be \(u\) and \(v\). 
    We connect one endpoint of \(P_1\) to \(u\) and one endpoint of \(P_2\) to \(v\).
    Finally, we connect \(\maxDeg-1\) new nodes to the other endpoint of \(P_2\).
    By Property~\ref{prop:nodes} of \cref{lem:lb_properties}, observe that \(\abs{V(G_k)} \le 2k^2 - 2 \le \varepsilon n / 2\) by our choice of \(k\) and, hence, \(P_2\) is well-defined.
    Therefore, the constructed graph \(G\) has exactly \(n\) nodes and maximum degree \(\maxDeg\).

    It is easy to see that all the properties of \cref{lem:lb_properties} still hold for the induced subgraph \(G[V(G_k)]\), and, in particular, by Property~\ref{prop:monochromatic}, \(u\) and \(v\) must output different labels in any valid solution to the locally optimal cut problem on \(G\).
    
    By contradiction, assume that $T(\maxDeg,n) \le \floor{\min\{(\maxDeg+1)/2, \sqrt{\varepsilon n}/2\}} / 2$ (which is less than \(k/2\)), and that \(\AA\) succeeds with probability at least \(1/2\).
    By construction of \(G\), the radius-\(T\) views of \(u\) and \(v\) are isomorphic, since \(T = T(n, \maxDeg) \le k/2 = \Theta(\sqrt{n})\) and both \(P_1\) and \(P_2\) have length \(\Theta(n)\).  
    Hence, by indistinguishability, there is a probability of at least \(1/2\) that both nodes in \(S^{(1)}\) and \(S^{(2)}\) take the same output label, reaching a contradiction.
\end{proof}

In \cref{app:lb}, we show that the same method can be applied to prove a lower bound of \(\Omega(\min\{\maxDeg,\sqrt{n}\})\) rounds for the locally optimal cut problem even in the quantum-LOCAL model. 
The following is a corollary of \cref{thm:lb:online-local}.

\begin{corollary}\label{cor:lb:quantum-local}
    For any large enough \(n \in \nats\) and \(\maxDeg < n\), let \(\varepsilon = 1 - (\maxDeg / n)\).
    Any algorithm that finds a locally optimal cut in the  quantum-LOCAL model with probability $p \ge \frac{1}{2}$ on graphs of \(n\) nodes and maximum degree \(\maxDeg\) has locality $\Omega(\min\{\maxDeg,\sqrt{\varepsilon n}\})$.
    In particular, for any \(\maxDeg \le c n\) for some constant \(c < 1\), the locality is \(\Omega(\min\{\maxDeg,\sqrt{n}\})\).
\end{corollary} %
\section*{Acknowledgments}

Francesco d'Amore was supported by the project Decreto MUR n. 47/2025,
CUP: D13C25000750001.
This work was partially supported by the MUR (Italy) Department of Excellence 2023 - 2027 for GSSI.
This project was also discussed at the Research Workshop on Distributed Algorithms (RW-DIST 2025) in Freiburg, Germany; we would like to thank all workshop participants and organizers for inspiring discussions. %

\printbibliography

\appendix 

\section{Locally optimal problems on edges}
\label{sec:app:lope}

In this section, we give an alternative (but more restrictive) definition of Locally Optimal Problems. 
All these problems can be also described as Locally Optimal Problems according to \cref{def:preliminaries:lop}, but the opposite is not necessarily true.
First, we introduce the notion of labeled edge-graphs.

\begin{definition}[Labeled edge-graph]\label{def:appendix:labeled-edge}
    Let \(\VV,\EE\) be two finite sets of labels. 
    We denote by \(\LL_E(\VV,\EE)\) the family of all \((\VV,\EE)\)-labeled graphs \(G\) where \(G\) consists only of two nodes that are connected by an edge.
    An element of \(\LL_E(\VV,\EE)\) is called a labeled edge-graph.
\end{definition}

\begin{definition}[Locally Optimal Problems on Edges]\label{def:appendix:lop-edges}
    A Locally Optimal Problem on Edges (LOPE) is a pair \((\problem, \pot)\) where:
    \begin{itemize}
        \item \(\problem = (\Vin, \Ein, \Vout, \Eout, \CC)\) is an LCL problem as defined in \cref{def:preliminaries:lcl-problems}, where \(\CC\) is a \((r, \maxDeg)\)-set of constraints over \((\Vin \times \Vout, \Ein \times \Eout)\).
        \item \(\pot: \LL_E(\VV,\EE) \to \reals_{\ge 0}\) is a function that assigns a non-negative real value to every labeled edge-graph in \(\LL_E(\VV,\EE)\).
    \end{itemize}
    Moreover, \(\CC\) must satisfy the following property, which in essence says that a correct solution does not allow local improvements.
    For any \((\Vin \times \Vout, \Ein \times \Eout)\)-labeled graph \(G = (V,E)\), let us denote by \(\ell = (\lblnodes,\lbledges)\) the labeling of \(G\).
    We say that the potential associated to \((G,\ell)\) is \(\graphPot(G,\ell) = \sum_{\{u,v\} \in E}\pot_\ell(G[\{u,v\}])\), where \(\pot_\ell(G[\{u,v\}])\) is the evaluation of \(\pot\) on the edge-graph \(G[\{u,v\}]\) that is labeled by \(\ell\).
    For any node \(v \in V\), let the cost of $v$ be \(c_v := \sum_{u \in \neighborhood(v)}\pot_\ell(G[\{u,v\}])\).
    Then, the centered graph \((G_{r}(v),v)\) labeled with $\ell$ belongs to \(\CC\) if and only if we cannot change the labeling of $v$ to improve the overall potential of the graph. Formally, there is no labeling \(\ell'\) with the following properties, where, for any node \(u\in V\), \(c_u' := \sum_{w \in \neighborhood(u)}\pot_{\ell'}(G[\{u,w\}])\):
    \begin{itemize}[noitemsep]
        \item Labeling $\ell'$ differs from \(\ell\) only on the \emph{output labels} of \(v\) and of its incident half-edges.
		\item According to labeling $\ell'$, \(G_{r}(v)\) is still a \((\Vin \times \Vout, \Ein \times \Eout)\)-labeled graph.
		\item It holds that \( c_{v}' < c_{v} \) and, hence, \(\graphPot(G, \ell') < \graphPot(G, \ell)\).
	\end{itemize}
\end{definition}

It is easy to see that any LOPE can be expressed as a Locally Optimal Problem according to \cref{def:preliminaries:lop}.

\begin{lemma}
    Let \((\problem, \pot)\) be a LOPE according to \cref{def:appendix:lop-edges}.
    Then there exists a LOP \((\problem', \pot')\) according to \cref{def:preliminaries:lop} such that any solution of \((\problem, \pot)\) is also a solution of \((\problem', \pot')\) and vice versa.
\end{lemma}
\begin{proof}
    Problem \(\problem'\) is a tuple \((\Vin, \Ein, \Vout, \Eout, \CC')\) where \(\CC'\) is a \((2r, \maxDeg)\)-set of constraints over \((\Vin \times \Vout, \Ein \times \Eout)\).
    Function \(\pot'\) is a function \(\pot' : \LL(\Vin \times \Vout, \Ein \times \Eout, r, \maxDeg) \to \reals_{\ge 0}\), which assigns a non-negative real value to every \((\Vin \times \Vout, \Ein \times \Eout)\)-labeled centered graph \((H,v)\).
    For any \((\Vin \times \Vout, \Ein \times \Eout)\)-labeled centered graph \((H,v)\), we define \(\pot'((H,v)) := \sum_{u \in \neighborhood(v)} \pot(H[\{u,v\}])\), where \(H[\{u,v\}]\) is the labeled edge-graph induced by nodes \(u\) and \(v\) in \(H\).
    Finally, \(\CC'\) contains all the labeled centered graphs \((H, v)\) of radius \(2r\) such that \(H(v)\) belongs to \(\CC\) according to \cref{def:appendix:lop-edges}.
    By definition of \(\CC'\), it is immediate to see that any solution to \((\problem, \pot)\) is also a solution to \((\problem', \pot')\) and vice versa.
    It remains to prove that \((\problem', \pot')\) is a LOP according to \cref{def:preliminaries:lop}.
    We need to show that, for every \((\Vin \times \Vout, \Ein \times \Eout)\)-labeled graph \(G\) and for every node \(v\) of \(G\), the centered graph \((G_{2r}(v), v)\) belongs to \(\CC'\) if and only if there is no labeling \(\ell'\) that differs from the current labeling \(\ell\) of \(G\) only on the output labels of \(v\) and of its incident half-edges, such that \(G_{2r}(v)\) is still a \((\Vin \times \Vout, \Ein \times \Eout)\)-labeled graph and such that \(\graphPot((G, \ell')) < \graphPot((G, \ell))\).

    Assume there is a node \(v\) such that we can modify the labeling of \(v\) and that of its incident half-edges to obtain a new labeling \(\ell'\) such that \(G_{2r}(v)\) is still a \((\Vin \times \Vout, \Ein \times \Eout)\)-labeled graph and such that \(\graphPot((G, \ell')) < \graphPot((G, \ell))\).
    Then, by definition of the potential function \(\pot'\), we have that there is at least one edge \(\{u,v\}\) such that \(\pot(G[\{u,v\}])\) is strictly smaller according to labeling \(\ell'\) than according to labeling \(\ell\), which implies that the centered graph \((G_{r}(v), v)\) does not belong to \(\CC\) according to \cref{def:appendix:lop-edges}, and hence that the centered graph \((G_{2r}(v), v)\) does not belong to \(\CC'\).

    Conversely, assume that there is a node \(v\) such that the centered graph \((G_{2r}(v), v)\) does not belong to \(\CC'\).
    Then, by definition of \(\CC'\), we have that the centered graph \((G_{r}(v), v)\) does not belong to \(\CC\) according to \cref{def:appendix:lop-edges}, which implies that there is a neighbor \(u\) of \(v\) and a labeling \(\ell'\) that differs from the current labeling \(\ell\) of \(G\) only on \(v\) and on \((v,{u,v})\), such that \(G_{r}(v)\) is still a \((\Vin \times \Vout, \Ein \times \Eout)\)-labeled graph and such that \(\pot(G[\{u,v\}])\) is strictly smaller according to labeling \(\ell'\) than according to labeling \(\ell\), which implies that \(\graphPot((G, \ell')) < \graphPot((G, \ell))\).
    Hence, there exists a relabeling \(\ell'\) that differs from the current labeling \(\ell\) of \(G\) only on the output labels of \(v\) and of its incident half-edges, such that \(G_{2r}(v)\) is still a \((\Vin \times \Vout, \Ein \times \Eout)\)-labeled graph and such that \(\graphPot((G, \ell')) < \graphPot((G, \ell))\), proving the thesis.
\end{proof}

Note that the problem of finding a locally optimal cut (\cref{def:lb:loc}) is an example of a LOPE, where the potential function \(\pot\) simply outputs \(1\) if the two endpoints of the edge have the same output labels, and \(0\) otherwise.

We define the notion of maximum edge-improvement of a LOPE.
The maximum edge-improvement, denoted by \(\Gamma\) is the maximum value that the potential can get in any edge-graph.
If \(\Gamma = 0\), the problem is trivial, as any solution is a correct solution.
Note that the function \(\pot\) can always be rescaled so that \(\Gamma = 1\).
It suffices to take \(\pot' = \pot/\Gamma\).
We will assume that \(\Gamma = 1\) in the whole section.
Note that, for a \(\beta\)-improving set \(S\) of a \((\problem, \pot)\), it always holds that \(\beta \le \Gamma \maxDeg \le \maxDeg\).

Recall the value \(\lambda\) that was the minimum improvement of a LOP according to \cref{def:preliminaries:lop}.
We define \(\lambda\) to be the same for LOPEs, but with respect to the LOP \((\problem', \pot')\) that is obtained by applying the transformation in \cref{def:appendix:lop-edges} to \((\problem, \pot)\), where \(\pot\) is rescaled so that \(\Gamma = 1\).
Note that \(\lambda \le \Gamma \maxDeg = \maxDeg\) in general.

In what follows, we assume we consider a LOPE \((\problem, \pot)\) according to \cref{def:appendix:lop-edges}, and we denote by \(\Gamma\) the maximum edge-improvement of \((\problem, \pot)\), and \(\lambda\) the minimum improvement associated to \((\problem, \pot)\).
Note that, given the description of the LOPE and the value of \(\maxDeg\), the nodes of the graph can locally compute the values of \(\Gamma\) and \(\lambda\).
All the definitions of improvement and improving sets that we gave for LOPs also extend to the case of LOPEs with the new potential function \(\graphPot\).
We can rewrite \cref{lem:subsetImpr} as follows.

\begin{lemma}\label{lem:appendix:subsetImpr}
    Let \(S \subseteq V\) be any subset of nodes of a \((\Vin \times \Vout, \Ein \times \Eout)\)-labeled graph \(G = (V,E)\) and let \(\ell\) be any labeling of \(G\).
    Consider any \(S \subseteq V\), and let \(B \subseteq S\).
    If \(A = S \setminus B\), then it holds that
    \[
        \imp(S,\ell,\ell^\star_S) \le \imp(A,\ell,\ell^\star_A) + \imp(B,\ell,\ell^\star_B) + 2\abs{E(A,B)} \cdot \Gamma,
    \]
    where \(E(X,Y) = \{\{u,v\} : u \in X, v \in Y\}\). 
    Recall that \(\Gamma\) is the maximum edge-improvement of \((\problem, \pot)\), and that the labeling \(\ell^\star_M\) is the best possible relabeling of \(M\).
\end{lemma}
\begin{proof}
    The proof is almost identical to that of \cref{lem:subsetImpr}.
    For any subset of nodes \(X \subseteq V\), we denote by \(E(X)\) the set of edges with both endpoints in \(X\).

    We have that 
	\begin{align*}
		& \imp(S,\ell,\ell^\star_S) = \graphPot(G,\ell) - \graphPot(G,\ell^\star_S) \\
		= \ & \sum_{\{u,v\}\in E} \pot_{\ell}(G[\{u,v\}]) - \pot_{\ell^\star_S}(G[\{u,v\}]) \\
		= \ & \sum_{\{u,v\} \in E(S) \cup E(V,S)} \pot_{\ell}(G[\{u,v\}]) - \pot_{\ell^\star_S}(G[\{u,v\}]) + \sum_{\{u,v\} \in E(V \setminus S)} \pot_{\ell}(G[\{u,v\}]) - \pot_{\ell^\star_S}(G[\{u,v\}]) \\
		= \ & \sum_{\{u,v\} \in E(S) \cup E(V,S)} \pot_{\ell}(G[\{u,v\}]) - \pot_{\ell^\star_S}(G[\{u,v\}]),
	\end{align*}
	where the last equality follows from the fact that the local potential of an edge-graph can only change if any label of the edge-graph is changed, and that the labeling \(\ell^\star_S\) differs from \(\ell\) only on the nodes in \(S\) and on their incident half-edges.
	Let $\ell_A$ be the labeling defined as follows: 
	On the nodes in \(A\) and their incident half-edges, it is equal to \(\ell^\star_S\), and on the rest of the graph it is equal to \(\ell\).
	We similarly define $\ell_{B}$ (we just replace \(B\) with \(A\) in the previous definition).

	We want to analyze the improvements of $(B,\ell_B)$ and $(A,\ell_A)$ compared to that of $(S,\ell^\star_S)$. 
	Consider any two adjacent nodes \(u,v \notin A\) (resp.\ \(u,v \notin B\)).
	The two labelings $\ell_B$ (resp.\ \(\ell_A\)) and $\ell^\star_S$ will yield  the same evaluation of the potential on the edge-graph \(G[\{u,v\}]\).

	The following holds:
	\begin{align*}
		& \imp(B,\ell, \ell_B) \\
		= \ & \sum_{\{u,v\} \in E(B)\cup E(V,B)} \big[\pot_{\ell}(G[\{u,v\}]) - \pot_{\ell_B}(G[\{u,v\}])\big] \\
		= \ &  \sum_{\substack{\{u,v\} \in E(B) \text{ or} \\ \{u,v\} \in E(V\setminus A, B)}} \big[\pot_{\ell}(G[\{u,v\}]) - \pot_{\ell_B}(G[\{u,v\}])\big] + \sum_{\{u,v\} \in E(A,B)} \big[\pot_{\ell}(G[\{u,v\}]) - \pot_{\ell_B}(G[\{u,v\}])\big]\\
		= \ & \sum_{\substack{\{u,v\} \in E(B) \text{ or} \\ \{u,v\} \in E(V\setminus A, B)}} \big[\pot_{\ell}(G[\{u,v\}]) - \pot_{\ell^\star_S}(G[\{u,v\}]) \big] + \sum_{\{u,v\} \in E(A,B)} \big[\pot_{\ell}(G[\{u,v\}]) - \pot_{\ell_B}(G[\{u,v\}])\big].
	\end{align*}
	Notice that we can get the exact same equality when replacing $B$ with $A$:
	\begin{align*}
		& \imp(A,\ell, \ell_A)\\
		= \ & \sum_{\substack{\{u,v\} \in E(A) \text{ or} \\ \{u,v\} \in E(V\setminus B, A)}} \big[\pot_{\ell}(G[\{u,v\}]) - \pot_{\ell^\star_S}(G[\{u,v\}]) \big] + \sum_{\{u,v\} \in E(A,B)} \big[\pot_{\ell}(G[\{u,v\}]) - \pot_{\ell_{A}}(G[\{u,v\}])\big].
	\end{align*}
		Bringing those together we recover the improvement for all edge-graphs in \((E(S)\setminus E(A,B)) \cup E(V,S)\) when applying \(\ell^\star_S\).
	\begin{align*}
		& \imp(A, \ell, \ell^\star_{A}) + \imp(B, \ell, \ell^\star_B) \\
		\ge \ & \imp(A,\ell,\ell_A) + \imp(B, \ell, \ell_B) \\
		= \ & \sum_{\substack{\{u,v\} \in E(A) \text{ or} \\ \{u,v\} \in E(V\setminus B, A)}} \big[\pot_{\ell}(G[\{u,v\}]) - \pot_{\ell^\star_S}(G[\{u,v\}])\big] + \sum_{\{u,v\} \in E(A,B)} \big[\pot_{\ell}(G[\{u,v\}]) - \pot_{\ell_{A}}(G[\{u,v\}])\big]\\
		+ & \sum_{\substack{\{u,v\} \in E(B) \text{ or} \\ \{u,v\} \in E(V\setminus A, B)}} \big[\pot_{\ell}(G[\{u,v\}]) - \pot_{\ell^\star_S}(G[\{u,v\}])\big] + \sum_{\{u,v\} \in E(A,B)} \big[\pot_{\ell}(G[\{u,v\}]) - \pot_{\ell_B}(G[\{u,v\}])\big] \\ 
		= \ & \sum_{\{u,v\} \in (E(S)\setminus E(A,B)) \cup E(V,S) } \big[\pot_{\ell}(G[\{u,v\}]) - \pot_{\ell^\star_S}(G[\{u,v\}])\big] \\
		- & \sum_{\{u,v\} \in E(A,B)} \big[ \pot_{\ell_{A}}(G[\{u,v\}]) + \pot_{\ell_B}(G[\{u,v\}]) 
		- 2\pot_{\ell}(G[\{u,v\}]) \big] \\ 
		= \ & \imp(S,\ell,\ell^\star_S) - \sum_{\{u,v\} \in E(A,B)}[\pot_{\ell}(G[\{u,v\}]) - \pot_{\ell^\star_S}(G[\{u,v\}])] \\
		- & \sum_{\{u,v\} \in E(A,B)} \big[ \pot_{\ell_{A}}(G[\{u,v\}]) + \pot_{\ell_B}(G[\{u,v\}]) 
		- 2\pot_{\ell}(G[\{u,v\}]) \big].
	\end{align*}  
		The claim follows by the definition of \(E(A,B)\) and by rearranging terms.
    By definition of \(\Gamma\), we get that
    \begin{align*}
        & \imp(S,\ell,\ell^\star_S) - \sum_{\{u,v\} \in E(A,B)} \big[\pot_{\ell_{A}}(G[\{u,v\}]) + \pot_{\ell_B}(G[\{u,v\}]) - \pot_{\ell^\star_S}(G[\{u,v\}]) - \pot_{\ell}(G[\{u,v\}])\big] \\
        \ge \ & \imp(S,\ell,\ell^\star_S) - 2\abs{E(A,B)} \cdot \Gamma.
    \end{align*}
\end{proof}

Before proceeding, we introduce the following auxiliary lemma.
\begin{lemma}\label{lem:appendix:iterative-seq}
    Consider a sequence of positive integers \(n_1, n_2, \ldots\) such that \(n_1 \ge 1\) and \(n_k \ge n_{k-1} + \sqrt{\alpha \cdot n_{k-1}}\), for some real value \(\alpha \in (0,1)\).
    Then, for every \(k \ge 1\), it holds that \(\sqrt{n_k} \ge 1 + (k-1)\sqrt{\alpha} / 3 \).
\end{lemma}
\begin{proof}
    Let \(x_k = \sqrt{n_k}\).
    Then, \(x_1 \ge 1 \) and \(x_k^2 \ge x_{k-1}^2 + \sqrt{\alpha} x_{k-1}\).
    This implies that 
    \[
        x_{k} - x_{k-1} \ge x_{k-1}\big(\sqrt{1 + (\sqrt{\alpha})/ x_{k-1}} - 1\big).
    \]
    Since \(0 \le (\sqrt{\alpha})/ x_{k-1} \le 1\), by the known inequality \(\sqrt{1+x} \ge 1 + x/3\) for all \(x \in [0,1]\), we have that \(x_{k} - x_{k-1} \ge (\sqrt{\alpha})/3\), which implies that \(x_k \ge 1 + (k-1)(\sqrt{\alpha})/3\).
\end{proof}

Since \cref{lem:subsetImpr} is the main ingredient in the proof of \cref{lem:imprBalls}, and we now have its updated version (\cref{lem:appendix:subsetImpr} of \cref{lem:subsetImpr}), we can also rewrite \cref{lem:imprBalls} as follows.

\begin{lemma}\label{lem:appendix:imprBalls}
	Let $\ell$ be a labeling of $G = (V,E)$ and let $(S, \ell_S)$ be a $\beta$-minimal improving set, for some $\beta > 0$. 
	Then, for every $0 < \varepsilon < \min{\{\beta,\lambda\}}$ and any $v \in S$ there exists a large enough constant \(c_3 > 1\), a $k \le c_3 \min\{ \maxDeg \log n / \varepsilon, \sqrt{ n / \varepsilon}\}$, and a subset $A \subseteq S \cap \neighborhood_k[v]$, such that $\impRatio(A, \ell, \ell_{A}^\star) \ge \beta - \varepsilon$.
	Furthermore, \((A,\ell_A^\star)\) is minimal.
\end{lemma}
\begin{proof}
    The proof proceeds similarly to that of \cref{lem:imprBalls}.
    For any subset $A \subseteq S$ we again define $B = S \setminus A$.
    Let \(x_S = \impRatio(S, \ell, \ell_S^\star)\). 
    Starting with the inequality of \cref{lem:appendix:subsetImpr} and using that $(S, \ell^\star_S)$ is minimal, we have
    \begin{align*}
        \imp(A,\ell, \ell^\star_A) &\ge \imp(S,\ell, \ell^\star_S) - \imp(B, \ell, \ell^\star_{B}) - 2  \abs{E(A,B)} \cdot \Gamma\\
                &\ge x_S \abs{S} - x_S \abs{B} - 2 \abs{E(A,B)}\cdot \Gamma\\
                &= x_S (\abs{S} - \abs{B}) - 2 \abs{E(A,B)}\cdot \Gamma\\
                &=x_S \abs{A} - 2 \abs{E(A,B)}\cdot \Gamma .
    \end{align*}
    So if $\abs{E(A,B)} \le \frac{\varepsilon}{2 \Gamma } \abs{A}$, we have that $(A, \ell_A^\star)$ has improving ratio at least $\beta -\varepsilon$ as desired. 
    Then, we can simply take the minimal improving set \((A',\ell^\star_{A'})\) such that \(A' \subseteq A\) with improving ratio at least \(\beta - \varepsilon\), and, if \(\diam(A')\) is small enough, we are done.

    For the sake of readability, we set $y = 2 \Gamma $.
    Fix $v\in S$ and consider the following ``candidate sets'':
    \( C_i = \neighborhood_i[v]\cap S\) for \(i \ge 0\). 
    Let \(n_i = \abs{C_i}\), and let \(k\) be the smallest integer \(i \ge 1\) such that \(n_i \le n_{i-1}(1 + \varepsilon/(y (\maxDeg+1)))\).

    We now prove that such \(k\) must exist and must be at most \(4y (\maxDeg+1) \ln n / \varepsilon + 1\).
    It holds that \(n_{k-1} > n_{k - 2}(1+\varepsilon/(y (\maxDeg+1))) > \ldots > (1+\varepsilon/(y (\maxDeg+1)))^{k-1}\).
    By contradiction, assume \(k -1 > 4 y (\maxDeg+1) \ln n / \varepsilon + 1\).
    By the known inequality \((1+x/y)^y \ge \exp[xy / (x+y)]\) for all \(x>0, y>0\), we have that 
    \begin{align*}
        & \left(1+\frac{\varepsilon}{y (\maxDeg+1)}\right)^{(k-1)} \\
        = \ & \left(1+\frac{\varepsilon}{y (\maxDeg+1)}\cdot \frac{(k-1)}{(k-1)}\right)^{(k-1)} \\ 
        \ge \ & \exp\left[\frac{\frac{\varepsilon}{y (\maxDeg+1)} \cdot (k-1)^2}{\frac{\varepsilon}{y (\maxDeg+1)} \cdot (k-1) + (k-1)}  \right] \\
        = \ & \exp\left[\frac{\frac{\varepsilon}{y (\maxDeg+1)} \cdot (k-1)}{\frac{\varepsilon}{y (\maxDeg+1)} + 1}  \right] \\
        \ge \ & \exp\left[\frac{\varepsilon}{2y (\maxDeg+1)} \cdot (k-1)  \right],
    \end{align*}
    where the last inequality follows from the fact that \(\varepsilon \le y (\maxDeg+1)\) (as \(\varepsilon \le \lambda \le \Gamma \maxDeg\)).
    Since \((k-1) > 2y (\maxDeg+1) \ln n / \varepsilon + 1\), we get that 
    \[
        \exp\left[\frac{\varepsilon}{2y (\maxDeg+1)} \cdot (k-1) + 1 \right] > \exp[\ln n] = n,
    \]
    and, hence, \(n_{k-1} > n\), which is a contradiction.

    If \(C_{k-1} = S\), we are done. 
    Assume that \(C_{k-1}\) is strictly contained in \(S\) and let \(A = C_{k-1}\) and \(B = S \setminus A\).
    We need to upper bound the size of \(E(A,B)\).
    Trivially, \(\abs{E(A,B)} \le \maxDeg \cdot (n_{k} - n_{k-1})\), which implies that 
    \begin{align*}
        \abs{E(A,B)} & \le \maxDeg \cdot (n_{k} - n_{k-1}) \\
        & \le \maxDeg \cdot n_{k} \cdot \left(1 - \frac{y(\maxDeg + 1)}{y(\maxDeg + 1) + \varepsilon}\right)\\
        & \le \maxDeg \cdot n_{k} \cdot \frac{\varepsilon}{y(\maxDeg + 1) + \varepsilon} \\
        & \le \frac{\varepsilon}{2\Gamma} \cdot n_{k},
    \end{align*} 
    where the latter inequality holds since \(\maxDeg y \le y(\maxDeg + 1) + \varepsilon\).

    Now, we also want to prove that there exists a \(k \le 2 + 3\sqrt{yn/\varepsilon}\) such that \(A\) can be found within \(C_k\).
    Let \(m_i = \abs{E(C_i,S\setminus C_i)}\).
    Consider \(k\) such that \(k\) is the minimum integer satisfying \(n_k \le n_{k-1} + \sqrt{(\varepsilon/y) n_{k-1}}\).
    We can prove that such \(k\) is at most \(k < 1+ 3\sqrt{yn/\varepsilon}\).
    If such \(k\) does not exist, then \(n_k > n_{k-1} + \sqrt{(\varepsilon/y) n_{k-1}}\) for every \(k \ge 2\), and \(n_1 \ge 1\).
    Then, we can apply \cref{lem:appendix:iterative-seq} and get that \(\sqrt{n_k} \ge 1 + (k-1)\sqrt{\varepsilon/y} / 3\) for every \(k \ge 2\).
    To get the contradiction that \(n_k > n\), it suffices to take \((k-1) \ge 3\sqrt{yn/\varepsilon} \), which means \(k \ge 1+ 3\sqrt{yn/\varepsilon}\).

    Now take \(A = C_{k-1}\): we want to prove that \(m_{k-1} \le \varepsilon \abs{A}/ y\).
    We know that \(n_k \ge n_{k-1} + \sqrt{m_{k-1}}\), since \(m_{k-1}\) new edges must generate at least \(\sqrt{m_{k-1}}\) new nodes.
    If \(m_{k-1} > \varepsilon {n_{k-1}}/ y\),then \(n_k > n_{k-1} + \sqrt{\varepsilon {n_{k-1}}/ y}\), reaching a contradiction by definition of \(k\).

    By using that \(\Gamma = 1\), we get the thesis.
\end{proof}

Similarly, \cref{lem:imprChain} also can be adapted to the case of LOPEs. 

\begin{lemma}\label{lem:appendix:imprChain}
	Let \(0 < \beta \le \maxDeg\) be any value, and let $0<\varepsilon< \min\{\beta,\lambda\}$.
	Let $((A_i,\ell_i))_{1\leq i \leq h}$ be a sequence of $\beta$-improving sets w.r.t.\ \(\ell\) in \(G\), where each minimal improving set has diameter at most $c_3 \min\{ \maxDeg \log n / \varepsilon, \sqrt{ n / \varepsilon}\}$ (\(c_3\) as given by \cref{lem:appendix:imprBalls}).
	Then for all $1 \leq j \leq h$, for all \(v \in A_j\) there exists a large enough constant \(c_2 > 1\) and some $k \le c_2 \maxDeg^2  ({\log n} /{\varepsilon})^2$, such that there exists a set $A \subseteq \neighborhood_k[v] \cap (\bigcup_{1 \leq i \leq h} A_i)$ with $\impRatio(A, \ell, \ell^*_{A}) \ge \beta - \varepsilon $.
\end{lemma}
\begin{proof}
    We define the auxiliary graph $H = (V_H,E_H)$ where \(V_H\) contains a node \(\sigma(A_i)\) for each \(A_i\), \(i \in [h]\), and $\{\sigma(A_i),\sigma(A_j)\} \in E_H$ if and only if \(\dist_G(A_i,A_j) \le 1\).

    The reason for defining the edges like this is the following: if the distance of \(A_i\) and \(A_j\) is at most \(1\), there are edge-graphs that are affected by both the relabeling of \(A_i\) and \(A_j\). 
    Conversely, if $\sigma(A_i)$ and $\sigma(A_j)$ are not adjacent in $H$, then no edge-graph can be affected by relabelings of both $A_i$ and $A_j$.

    We define a weight function $w \colon V_H \to \nats$ for the nodes in $H$ by
    $w(x) = \abs{\sigma^{-1}(x)}$,
    that is, to any element \(x = \sigma(A_i)\) (for some \(i \in [h]\)), we assign a weight that is equal to the cardinality of \(A_i\).
    With an abuse of notation,  we define $w(B) = \sum_{x \in B}w(x)$ for every subset $B$ of $V_H$.

    Observe that the improvement of the entire sequence of \(\beta\)-improving sets $((A_i, \ell_i))_{1\le i\le h}$ is at least $\beta \cdot w(V_H)$, due to the definition of sequence of \(\beta\)-improving sets.
    Now since the maximum improvement in the potential of any labeling in $G$ is at most $\Gamma \cdot \abs{E}$, the overall improvement of the sequence itself cannot be larger than $\Gamma \cdot \abs{E}$. 
    This implies that $ w(V_H) \le \Gamma \cdot \abs{E} / \beta$.
    Note that \(1/(n-1) \le \Gamma / \beta\), hence \(\Gamma \cdot \abs{E} / \beta \ge 1\).

    Fix any \(j \in [h]\), and consider an arbitrary \(v \in A_j\). 
    Let \(z = \sigma(A_j) \in V_H\).
    We define the ``candidate sets'' with regard to \(z\) as \(C_i = \neighborhood_i[z]\) for \(i \ge 0\).
    Let \(k\) be the minimum integer \(i\)  such that \(\frac{\varepsilon}{K} w(C_i) > w(C_{i+1} \setminus C_i) = w(C_{i+1}) - w(C_i)\), for \(K = (\beta - \varepsilon) + \maxDeg\Gamma \).
    We claim that \(k \le \ceil{4K \log (\Gamma n^2 / \beta) / \varepsilon}\).
    By contradiction, suppose that \(k > \ceil{4K \log (\Gamma n^2 / \beta) / \varepsilon}\).
    Let \(j = \ceil{4K \log (\Gamma n^2 / \beta) / \varepsilon}\).
    Then, we have that \(w(C_{i+1}) - w(C_i) \ge \frac{\varepsilon}{K} w(C_i)\) for all \(i \le j\).
    Then, 
    \[
        w(C_j) \ge (1 + \varepsilon/K) w(C_{j-1}) \ge (1 + \varepsilon/K)^{4K \log (\Gamma n^2 / \beta) / \varepsilon}.
    \]
    Again, by exploiting that \((1 + x/y)^y \ge \exp[xy/(x+y)]\) for all positive \(x,y\), and that \(\varepsilon \le \min\{\beta, \lambda\} \le \Gamma \maxDeg\), we get that 
    \[
        w(C_j) \ge \exp\left[\frac{4 K \log (\Gamma n^2 / \beta)}{\varepsilon + K}\right] > (\Gamma n^2 / \beta)^{2}.
    \]
    This is a contradiction since \(w(C_j) \le w(V_H) \le \Gamma n^2 / \beta\).

    Let 
    \begin{align*}
        A & = \{i \in [h] \mid A_i = \sigma^{-1}(x) \text{ for some } x \in C_{k+1}\}, \text{ and} \\
        A' & = \{i \in [h] \mid A_i =\sigma^{-1}(x) \text{ for some } x \in C_{k}\}.
    \end{align*}
    Furthermore, let \(s = \max (A)\). 
    We also define \(\hat{A} = \cup_{i \in A} A_i\), and \(\hat{A}' = \cup_{i \in A'} A_i\).
    Denote by $\ell_A$ the labeling obtained by setting $\ell_A = \ell_{s}$ for all nodes in $\hat{A}$ and their incident half-edges, and $\ell_A = \ell$ for all other nodes and half-edges.

    Let \(i_1, \ldots, i_t\) denote the elements of \(A\), where \(i_1 < \ldots < i_t = s\).
    We define $f_{i_0} = \ell$ and then a new  relabeling $f_{i_j}$ as the labeling that outputs $\ell_{i_{j}}$ on all the nodes in $A_{i_j}$ and their incident half-edges, but outputs $f_{i_{j-1}}$ everywhere else.
    Now we can obtain $\ell_A$ by sequentially applying the relabelings $f_{i_1}, f_{i_2}, \ldots, f_{i_t} = \ell_A$.

    Importantly, for each \(j \in [t]\), by the definition of \(f_{i_j}\) and \(E_H\), the labels of all nodes in \(\neighborhood[\hat{A}']\) and their incident half-edges given by \(f_{i_j}\) are exactly those given by \(\ell_{i_j}\): in fact, if the labeling of a node \(u \in \neighborhood[\hat{A}']\) (or its incident half-edges) is modified by \(\ell_{i_j}\), then \(u \in A_{i_j}\) and \(\dist_G(A_{i_j},\hat{A}') \le 1\), implying that \(i_j \in A\).

    As a result, if $i_j \in A'$, then
    \[
        \imp(A_{i_j}, f_{i_{j-1}}, f_{i_j}) = \imp(A_{i_j}, \ell_{i_{j -1}}, \ell_{i_j}) \ge \imp(A_{i_j}, \ell_{i_{j} - 1}, \ell_{i_j}) \ge \beta \cdot \abs{A_{i_j}}.
    \]  
    On the other hand, if $i_j \in A \setminus A'$, any relabeling of \(A_{i_j}\) can worsen the potential by at most \(\Gamma \cdot \abs{E(A_{i_j}, V) \cup E(A_{i_j})} \le \Gamma \maxDeg \abs{A_{i_j}}\).
    We obtain the following:
    \begin{align*}
        \imp(\hat{A},\ell, \ell_A) = \ & \sum_{j=1}^{t} \imp(A_{i_j}, f_{i_{j-1}}, f_{i_j}) \\
        = \ & \sum_{i_j \in A'} \imp(A_{i_j}, f_{i_{j-1}}, f_{i_j}) + \sum_{i_j \in A\setminus A'} \imp(A_{i_j}, f_{i_{j-1}}, f_{i_j})\\
        \ge \ & \sum_{i_j \in A'} \imp(A_{i_j}, \ell_{i_j -1}, \ell_{i_j}) - \sum_{i_j \in A\setminus A'} \maxDeg \cdot \Gamma \cdot  \abs{A_{i_j}}\\
        \ge \ & \beta \cdot w(C_{k}) - \maxDeg \cdot \Gamma \cdot (w(C_{k+1}) - w(C_{k}))\\
        \ge \ & \beta\cdot  w(C_{k}) -   \maxDeg \cdot \Gamma \cdot \frac{\varepsilon}{K} w(C_{k})  \\
        = \ & w(C_{k}) \left(\beta - \frac{\varepsilon \maxDeg \Gamma}{K}\right).
    \end{align*}
    We now show that, by the definition of \(K\), it holds that
    \[
        w(C_{k}) \left(\beta - \frac{\varepsilon \maxDeg \Gamma}{K}\right) \ge w(C_{k + 1}) \left(\beta - \varepsilon\right).
    \]
    Recall that \(w(C_{k + 1}) \le w(C_{k})(1 + \varepsilon/K)\).
    Hence,
    \begin{align*}
        w(C_{k}) \left(\beta - \frac{\varepsilon \maxDeg \Gamma}{K}\right) & \ge w(C_{k+1}) \cdot \frac{K}{K+\varepsilon} \cdot \left(\beta - \frac{\varepsilon \maxDeg \Gamma}{K}\right) \\
        & = w(C_{k+1}) \cdot \frac{\beta K - \varepsilon(K - (\beta - \varepsilon))}{K+\varepsilon} \\
        & = w(C_{k+1}) \cdot \frac{K(\beta - \varepsilon) + \varepsilon(\beta - \varepsilon)}{K+\varepsilon} \\
        & = w(C_{k+1}) \cdot \left(\beta - \varepsilon\right) \cdot \frac{K + \varepsilon}{K + \varepsilon} \\
        & = w(C_{k+1}) \cdot \left(\beta - \varepsilon\right).
    \end{align*}

    Note that \(C_{k+1} = \neighborhood_{k+1}[z]\).
    Thus, the weak diameter of \(\hat{A}\) in \(G\) is at most \( (2k+1)\cdot(c_3 \min\{ \maxDeg \log n / \varepsilon, \sqrt{ n / \varepsilon}\})\).
    In particular, since \(k \le \ceil{4K \log (\Gamma n^2 / \beta) / \varepsilon}\) and \(K \le \beta + \maxDeg\Gamma\), we have that the weak diameter of \(\hat{A}\) in \(G\) is at most \(100 ((\beta + \maxDeg \Gamma) \cdot \log n / \varepsilon)(c_3 \min\{ \maxDeg \log n / \varepsilon, \sqrt{ n / \varepsilon}\})\).

    By using that \(\Gamma = 1\), and that \(0 < \beta <  \maxDeg\), we have that 
    \[
        \hat{A} \subseteq \neighborhood_{c_2 \maxDeg  \frac{\log n}{\varepsilon} \min\{ \maxDeg \frac{\log n}{\varepsilon}, \sqrt{ \frac{n}{\varepsilon}}\}}[v],
    \]
    for some large enough constant \(c_2 > 1\).
    Since \(\sqrt{n / \varepsilon}\) is the minimum when \(\maxDeg \ge \sqrt{n \varepsilon} \log n\), but for such values of \(\maxDeg\), the term
    \[
        c_2 \maxDeg  \frac{\log n}{\varepsilon} \min\left\{ \maxDeg \frac{\log n}{\varepsilon}, \sqrt{ \frac{n}{\varepsilon}}\right\} \ge n,
    \]
    we get the desired bound.
    By noticing that \(\ell^\star_{\hat{A}}\) can only increase the improving ratio w.r.t.\ \(\ell_{{A}}\), we conclude the proof.
\end{proof}

Now, we can update the statement of \cref{lem:what-we-actually-use}.

\begin{lemma}\label{lem:appendix:what-we-actually-use}
	Let \((A_i, \ell_{i})_{i = 1}^h\) be a sequence of \(\beta\)-improving sets for some \(0 < \beta \le \maxDeg\), each of which has diameter at most $c_3 \min\{ \maxDeg \log n / \varepsilon, \sqrt{ n / \varepsilon}\}$ (\(c_3\) as given by \cref{lem:appendix:imprBalls}), for \(0 < \varepsilon < \min\{\lambda,\beta\}\).
	Then, for each \(i \in [h]\) and each \(v \in A_i \), there exists a large enough constant \(c_2\) and an improving set \((A,\ell_A)\) of diameter at most $c_3 \min\{ \maxDeg \log n / \varepsilon, \sqrt{ n / \varepsilon}\}$, such that \(\impRatio(A, \ell, \ell_A) \ge \beta - 2\varepsilon\), that is fully contained in the set \(\neighborhood_{c_2 \maxDeg^2  ({\log n} /{\varepsilon})^2}[v] \cap (\cup_{j \in [h]} A_j)\).
	Furthermore, \((A,\ell_A)\) is minimal.
\end{lemma}
\begin{proof}
    We first apply \cref{lem:appendix:imprChain}, and we obtain an improving set \((A, \ell_A)\)   with improving ratio at least \(\beta - \varepsilon\) that is fully contained in \(\neighborhood_{c_2 \maxDeg^2  ({\log n} /{\varepsilon})^2}[v] \cap (\cup_{j \in [h]} A_j)\).
	Then, we take the minimal version of it, that is, \((A',\ell^\star_{A'})\), with \(A' \subseteq A\), and we apply \cref{lem:appendix:imprBalls}, which yields the thesis.
\end{proof}

We now state the adapted version of \cref{alg:lop}, which is the one we actually analyze in the case of LOPEs.

\begin{algorithm}
\caption{LOPE algorithm (from the perspective of a node)}\label{alg:lope}
\begin{algorithmic}[1]
\Require Number of nodes \(n\), max-degree \(\maxDeg\), description of the LOPE \((\problem, \pot)\).
\Ensure Solution to \((\problem, \pot)\).
\State Initialize \(\lambda \gets \) (scaled) minimum improvement associated to \((\problem, \pot)\) \label{algo:appendix:min-improving-ratio}
\State Initialize \(R \gets \lambda/8\) \label{algo:appendix:line:initial-IR}
\State Initialize \(\varepsilon \gets \frac{\lambda}{200 c_1  \log n} \) \Comment{\(c_1\) is a constant given by the proof of \cref{thm:algorithm:lop}}\label{algo:appendix:line:epsilon}
\State Initialize any output label for the node and its incident half-edges \label{algo:appendix:line:initial-labeling}
\Statex \emph{The next loop identifies the ``phases'' of our algorithm. In the following, \(c\), \(c_1\), \(c_2\), and \(c_3\) are large enough constants. The constants \(c_2\) and \(c_3\) are given by \cref{lem:appendix:imprBalls,lem:appendix:imprChain}, while \(c_1\) is given by the proof of \cref{thm:algorithm:lop}, and \(c\) is given by \cref{lemma:preliminaries:mpx}.}
\For{$i = 1$ to $c_1 \log n$} \label{algo:appendix:line:phases} 
    \State Run MPX with \(\rho = \frac{\varepsilon^2}{10 c \cdot c_2 \maxDeg^2  \log^2 n}\) \Comment{\(d = O(\log n / \rho)\) by \cref{lemma:preliminaries:mpx}} \label{algo:appendix:line:mpx}
    \State \(C \gets\) cluster of the node \label{algo:appendix:line:cluster}
    \State \(H \gets \) graph induced by \(C\) \label{algo:appendix:line:induced-graph}
    \State \(\ell \gets \) labeling of \(H\) \label{algo:appendix:line:initial-labeling-graph}
    \If{the node is the cluster leader} \label{algo:appendix:line:if-leader}
        \State Choose any maximal sequence of \(R\)-improving sets \(\left((A_i,\ell_i)\right)_{i = 1}^k\) in \(H\) w.r.t.\ \( \ell\) such that \label{algo:appendix:line:choose-maximal-sequence}
        \State \hspace{\algorithmicindent} 1. \(\neighborhood[A_i] \subseteq C\) for all \(i \in [k]\) \label{algo:appendix:line:safe-neighborhood}
        \State \hspace{\algorithmicindent} 2. The diameter of each $A_i$ is at most $c_3 \min\{ \maxDeg \log n / \varepsilon, \sqrt{ n / \varepsilon}\}$ \label{algo:appendix:line:diameter} 
        \State Adopt the new labeling according to \(\ell_k\) \label{algo:appendix:line:adopt-labeling-leader} 
        \State \(\ell \gets \ell_k\) \label{algo:appendix:line:update-labeling-leader}
        \State Broadcast \(\ell_k\) to all nodes in \(C\) \label{algo:appendix:line:broadcast-sequence}
    \Else \label{algo:appendix:line:not-leader}
        \State Wait for the broadcast by the leader \label{algo:appendix:line:wait-broadcast}
        \State Adopt the new labeling according to \(\ell_k\) \label{algo:appendix:line:adopt-labeling-not-leader} 
        \State \(\ell \gets \ell_k\) \label{algo:appendix:line:update-labeling-not-leader}
    \EndIf \label{algo:appendix:line:endif-leader}
    \State Set \(R \gets R + \frac{\lambda}{40 c_1 \log n}\) \label{algo:appendix:line:update-IR}
\EndFor \label{algo:appendix:line:end-phases}
\State \Return labels of the node and of the incident half-edges \label{algo:appendix:line:return-labels}
\end{algorithmic}
\end{algorithm}

We proceed by proving the correctness of \cref{alg:lope}.
First, we show that each phase of the algorithm terminates.
\begin{lemma}\label{lem:appendix:phase-terminates}
	For any phase \(j\) of the algorithm, let \(\ell_j\) be the labeling of the algorithm after finishing phase \(j-1\) (Lines~\ref{algo:appendix:line:update-labeling-leader}--\ref{algo:appendix:line:update-labeling-not-leader}) and \(R_j\) be the value of \(R\) at the end of phase \(j-1\) (Line~\ref{algo:appendix:line:update-IR}).
	Furthermore, let \(C^{(j)}_1, \ldots, C^{(j)}_k\) be the clusters that are the result of the MPX algorithm at phase \(j\) (Line~\ref{algo:appendix:line:mpx}),
	Then, for each cluster \(C^{(j)}_i\), any maximal sequence of \(R_j\)-improving sets is finite.
\end{lemma}
\begin{proof}
	Let \(H^{(j)}_i \) be the graph induced by cluster \(C^{(j)}_i\), and \(\ell_{H^{(j)}_i}\) the restriction of \(\ell_j\) to \(H^{(j)}_i\).
	Since \(\graphPot(H^{(j)}_i, \ell_{H^{(j)}_i})\) is a finite value, and each improving set decreases this value by at least some constant \(R_j > 0\), then the sequence cannot have more than \(\graphPot(H^{(j)}_i, \ell_{H^{(j)}_i})/R_j\) elements.
\end{proof}

Second, let us give an adapted version of \cref{lem:improving-sets-close-to-borders}.
By \(B_i\), we refer to the border sets as defined in \cref{def:borders} but with respect to the clusters determined by MPX at phase \(i\) (Line~\ref{algo:appendix:line:mpx}) of \cref{alg:lope}.
\begin{lemma}\label{lem:appendix:improving-sets-close-to-borders}
	At the end of phase \(i\) of the algorithm, all minimal improving sets of diameter at most \(c_3 \min\{\maxDeg \log n / \varepsilon, \sqrt{n / \varepsilon}\}\) and improving ratio at least \(R_i\) (the value of \(R\) initialized at the end of phase \(i-1\) at Line~\ref{algo:appendix:line:update-IR}) are fully contained in the set
	\[
		\neighborhood_{t_1}[B_i] \cap \left( \bigcap_{j = 1}^{i-1} \neighborhood_{t_2}[ B_j]\right),
	\]
	for \(t_1 = c_3 \min\{\maxDeg \log n / \varepsilon, \sqrt{n / \varepsilon}\} + 1\) and \(t_2 = c_2 \maxDeg^2  ({\log n} /{\varepsilon})^2 + t_1\), where \(c_2\) and \(c_3\) are as in Lines~\ref{algo:appendix:line:mpx} and \ref{algo:appendix:line:diameter}, respectively.
\end{lemma}
\begin{proof}
	For any phase \(i > 0\), let \(\CC^{(i)}_{1}, \ldots, \CC^{(i)}_{h_i}\) be the clusters determined by MPX at phase \(i\) (Line~\ref{algo:appendix:line:mpx}), let \(\ell_i\) be the labeling of \(G\) at the end of phase \(i\) (Lines~\ref{algo:appendix:line:update-labeling-leader} and \ref{algo:appendix:line:update-labeling-not-leader}), and let \(R_i\) be the improving ratio initialized at the end of phase \(i - 1\) (Line~\ref{algo:appendix:line:update-IR}), with \(R_1 = R\) as in Line~\ref{algo:appendix:line:initial-IR}.
	We proceed by induction on the phase \(i\).

	If \(i = 1\), \(R_1 = R\) as initialized at Line~\ref{algo:appendix:line:initial-IR}.
	A minimal improving set \((A, \ell_A)\) of diameter at most \(c_3 \min\{\maxDeg \log n / \varepsilon, \sqrt{n / \varepsilon}\}\) with \(\impRatio(A, \ell_1, \ell_A) \ge R_1\) must be such that there is no \(s \in [h_1]\) such that \(\neighborhood[A] \subseteq C^{(1)}_s\), otherwise we are breaking maximality in Line~\ref{algo:appendix:line:choose-maximal-sequence}. 
	Hence, \(\dist(A, B_1) \le 1\) and we have that \(A \subseteq \neighborhood_{c_3 \min\{\maxDeg \log n / \varepsilon, \sqrt{n / \varepsilon}\} + 1}[B_1]\), as the diameter of \(A\) is at most \(c_3 \min\{\maxDeg \log n / \varepsilon, \sqrt{n / \varepsilon}\}\).
	Now, assume \(i > 1\) and the statement to be true for all phases \(j = 1, \ldots, i-1\).
	Let \((A,\ell_A)\) be a minimal improving set of diameter at most \(c_3 \min\{\maxDeg \log n / \varepsilon, \sqrt{n / \varepsilon}\}\) such that \(\impRatio(A, \ell_{i}, \ell_A) \ge R_i\).
	For the same reason as in the case \(i = 1\), it holds that \(A \subseteq \neighborhood_{c_3 \min\{\maxDeg \log n / \varepsilon, \sqrt{n / \varepsilon}\} + 1}[B_i]\).
	Suppose, by contradiction, that \(A\) is not contained in 
	\[
		\bigcap_{j = 1}^{i-1} \neighborhood_{t_2}[ B_j].
	\]
	This implies that there is a \(j^\star \le i-1\) such that \(A\) is not contained in \(\neighborhood_{t_2}[ B_{j^\star}]\).
	Note that the diameter of \(A\) is at most \(c_3 \min\{\maxDeg \log n / \varepsilon, \sqrt{n / \varepsilon}\}\) and at least one node of $A$ is not in  \(\neighborhood_{t_2}[ B_{j^\star}]\).
	As a result, the set \(\neighborhood_{c_2 \maxDeg^2  ({\log n} /{\varepsilon})^2 + 1}[A]\) is fully contained in some cluster \(C^{(j^\star)}_s\), for some \(s \in [h_{j^\star}]\).

	We can now describe the family of all improving sets that are actually used by our algorithm after phase $j^\star$ (Line~\ref{algo:appendix:line:choose-maximal-sequence})  as a single improving sequence.
	We will define an ordering of all minimal improving sets that our algorithm has flipped.
	In order to do that, we first define an ordering of all clusters.
	We say that \(C^{(m_1)}_{s_1} < C^{(m_2)}_{s_2}\) if and only if:
	\begin{itemize}[noitemsep]
		\item \(m_1 < m_2\), or
		\item \(m_1 = m_2\) and \(s_1 < s_2\).
	\end{itemize}
	In each cluster \(C^{(m)}_s\) we also consider the natural ordering of the sets composing the sequence \((S^{(m)}_{s,t},\ell_{S^{(m)}_{s,t}})_{t \in [k_{s}^{(m)}]}\) of \(R_{m}\)-improving sets that were chosen by the leader of \(C^{(m)}_s\) (Line~\ref{algo:appendix:line:choose-maximal-sequence}), where \(k_s^{(m)}\) is the number of minimal improving sets that have been flipped inside \(C^{(m)}_s\) at phase \(m\).
	We can create a sequence of \(R_{j^\star + 1}\)-improving sets by concatenating all these maximal sequences 
	\[
		\left(\left((S^{(m)}_{s,t},\ell_{S^{(m)}_{s,t}})_{t \in [k_{s}^{(m)}]}\right)_{s \in [h_m]}\right)_{j^\star + 1 \le m \le i}
	\] 
	and at the end we put the improving set \((A, \ell_A)\).
	Observe that \((A, \ell_A)\) is such that \(\impRatio(A, \ell_{i}, \ell_A) \ge R_i \ge R_{j^\star + 1}\) and hence the whole sequence is an \(R_{j^\star + 1}\)-improving sequence.
	Let us rename this sequence as \((\hat{S}_t)_{t \in [t_i]}\), where \(t_i\) is the overall length of the defined sequence.

	Now, for each \(v \in A\), by \cref{lem:appendix:what-we-actually-use} we have that there exists a minimal improving set \((A', \ell_{A'})\) that is fully contained in the set 
	\[
		\neighborhood_{c_2 \maxDeg^2  ({\log n} /{\varepsilon})^2 + t_1}[v] \cap \left(\cup_{t \in [t_i]} \hat{S}_t\right),
	\]
	and such that \(\impRatio(A', \ell_{j^\star}, \ell_{A'}) \ge R_{j^\star+1} - 2\varepsilon\).
	By \cref{lem:appendix:imprBalls}, for each \(v \in A'\), there exists a minimal improving set \((A'', \ell_{A''})\) such that \(A'' \subseteq A' \cap \neighborhood_{c_3 \min\{\maxDeg \log n / \varepsilon, \sqrt{n / \varepsilon}\}}[v]\) and \(\impRatio(A'', \ell_{j^\star}, \ell_{A''}) \ge R_{j^\star+1} - 3\varepsilon \ge  R_{j^\star}\).
	The fact that \[A'' \subseteq \neighborhood_{c_2 \maxDeg^2  ({\log n} /{\varepsilon})^2}[v]\] implies that \(A'' \subseteq \neighborhood_{c_2 \maxDeg^2  ({\log n} /{\varepsilon})^2}[A]\).
	Since \(\neighborhood_{c_2 \maxDeg^2  ({\log n} /{\varepsilon})^2 + 1}[A]\) is fully contained in some cluster \(C^{(j^\star)}_s\), so it is \(\neighborhood[A'']\).
	Hence, the sequence \((S^{(j^\star)}_{s,t},\ell_{S^{(j^\star)}_{s,t}})_{t \in [k_{s}^{(j^\star)}]}\) was not maximal, reaching a contradiction.
\end{proof}

We now state the adapted version of \cref{thm:algorithm:lop}.
\begin{theorem}\label{thm:algorithm:lope}
    Let \((\problem, \pot)\) be a LOPE as defined in \Cref{def:appendix:lop-edges}.
    The randomized \local algorithm described in \cref{alg:lope} solves \((\problem, \pot)\) in \(O\left( (\maxDeg/\lambda)^2 \log ^6 n\right)\) rounds w.h.p.
    If \(\maxDeg = O(1)\), then the running time of \cref{alg:lope}  is \(O(\log^6 (n))\) rounds.
\end{theorem}
\begin{proof}
	By \cref{lem:appendix:phase-terminates}, we have that the algorithm terminates.
	The running time of the algorithm is the running time of MPX (Line~\ref{algo:appendix:line:mpx}) multiplied by the number of phases (Line~\ref{algo:appendix:line:phases}).
    By \cref{lemma:preliminaries:mpx}, this equals to 
    \[
        O\left( \log n \cdot  \maxDeg^2 \frac{\log ^3 n}{\varepsilon^2} \right) = O\left(\frac{\maxDeg^2}{\lambda^2} \log^6 n\right).
    \]

	Let \(i = c_1 \log n\) be the last phase, and let \(\ell_i\) be the labeling obtained at the end of phase \(i\) (Lines~\ref{algo:appendix:line:update-labeling-leader} and \ref{algo:appendix:line:update-labeling-not-leader}).
	By contradiction, suppose that there is an error in the graph, that is, a centered graph \((H, v_H)\) of radius \(r\) that does not belong to \(\CC\).
	This means that $v_H$ can change the output labels of itself and its incident edges to get an improvement of at least \(\lambda\). 
	Let $\ell_v$ be that relabeling, so that we can get the improving set $(\{v_H\}, \ell_v)$, with $\IR(\{v_H\},\ell_i,\ell_v) \ge \lambda$.
	Let us now take the minimal version $(\{v_H\}, \ell_v^\star)$ of $(\{v_H\}, \ell_v)$. 

	Since \(R_{c_1 \log n} \le \lambda\) by definition, \cref{lem:appendix:improving-sets-close-to-borders} implies that \(\{v_H\}\) is fully contained in 
	\[
		\neighborhood_{t_1}[B_i] \cap \left( \bigcap_{j = 1}^{i-1} \neighborhood_{t_2}[ B_j]\right).
	\]

	Notice that, by \cref{lemma:preliminaries:mpx}, at each phase \(j \le i\), for every node \(v\) there is probability at least \(1/2\) that \(\neighborhood_{t_2 + 1}[v]\) is fully contained within some cluster of the \(j\)-th run of MPX.
	Note that \((1/2)^{c_1\log n} = 1/n^{c_1}\).
	Hence, with probability \(1 - 1/n^{c_1}\), \(\{v_H\}\) is not contained in 
	\[
		\neighborhood_{t_1}[B_i] \cap \left( \bigcap_{j = 1}^{i-1} \neighborhood_{t_2}[ B_j]\right),
	\] 
	reaching a contradiction with \cref{lem:appendix:improving-sets-close-to-borders}.

	Now, by the union bound, with probability \(1 - 1/n^{c_1 - 1}\), there is no node that is a center of a centered graph of radius \(r\) that is invalid according to \(\CC\), yielding the thesis.
\end{proof}

We can derandomize \cref{alg:lope} by using \cref{thm:analysis:derandomization,thm:analysis:network-decomposition}, and we get the following corollary.

\begin{corollary}\label{cor:analysis:det-complexity-lope}
    There exists a deterministic \local algorithm that solves any LOPE \((\problem, \pot)\) in \(O((\maxDeg / \lambda)^2 \log^8 (n) \poly ( \log \log n))\) rounds.
    If \(\maxDeg = O(1)\), then the running time of the deterministic algorithm is \(\tilde{O}(\log^8 (n))\) rounds.
\end{corollary}

\section{Lower bound in super-quantum models}\label{app:lb}

In this section, we show how to obtain \cref{cor:lb:quantum-local}.

First, we have to give the definitions of the models we are considering and for which we have implications.

\paragraph{The quantum-LOCAL model.}
The quantum-LOCAL model is a generalization of the deterministic \local model, where each node/processor is now a quantum processor, and the communication links between nodes are quantum channels. 
Nodes can manipulate an arbitrary number of qubits, and can apply arbitrary unitary transformations on their local qubits, and can also perform measurements on them.
The communication between nodes happens in synchronous rounds, and in each round, each node can send an arbitrary amount of qubits to its neighbors, and receive an arbitrary amount of qubits from its neighbors.
The output of the algorithm is determined by the measurement outcomes of the qubits at each node at the end of the algorithm.
We require that the probability that the output of the algorithm is correct is at least \(1 - 1/n\), globally.

To date, it is difficult to provide lower bounds that apply specifically to the quantum-\local model.
However, we have at our disposal lower bound techniques that hold in more general models. 
We now introduce the randomized online-LOCAL model.

\paragraph{The randomized online-LOCAL model.}
In the randomized online-\local model, an adversary reveals the input graph to the algorithm in an online fashion.
We assume that the algorithm is given an infinite random bit string.
The adversary does not have access to the random bits of the algorithm, but knows the \emph{description} of the algorithm.
Then, it fixes an input graph \(G = (V,E)\), together with the inputs of each node. 
It also sets an adversarial ordering of the nodes \(v_1, \ldots, v_n\) of the graph, and reveals them one by one to the algorithm, together with their incident half-edges.
The adversary first reveals a node \(v_1\), and its \(T\)-hop neighborhood \(G_T(v_1)\) (together with the inputs of nodes in this neighborhood), where \(T\) is the locality/complexity of the algorithm.
Then, the algorithm must commit on \(v_1\) and its incident half-edges. 
After that, the adversary reveals the node \(v_2\) and its \(T\)-hop neighborhood \(G_T(v_2)\), and the algorithm must commit on \(v_2\) and its incident half-edges.
This process continues until all nodes of the graph are revealed and committed on.
The output of the algorithm is determined by the commitment of all nodes and their incident half-edges.
We require that the probability that the output of the algorithm is correct is at least \(1 - 1/n\), globally.

It has been proved that the randomized online-\local model is strictly more powerful than the quantum-\local model~\cite{akbari-coiteux-roy-etal-2025-online-locality-meets}.
We hence show that \cref{thm:locLB} holds also in the randomized online-\local model.

\begin{theorem}\label{thm:lb:online-local}
    For any large enough \(n \in \nats\) and \(\maxDeg < n\), let \(\varepsilon = 1 - (\maxDeg / n)\).
    Any algorithm that finds a locally optimal cut in the  randomized online-\local model with probability $p \ge \frac{1}{2}$ on graphs of \(n\) nodes and maximum degree \(\maxDeg\) has locality $\Omega(\min\{\maxDeg,\sqrt{\varepsilon n}\})$.
    In particular, for any \(\maxDeg \le c n\) for some constant \(c < 1\), the locality is \(\Omega(\min\{\maxDeg,\sqrt{n}\})\).
\end{theorem}
\begin{proof}
    Let $\AA$ be a $T(\maxDeg,n)$-round randomized online-\local algorithm that solves locally optimal cut on graphs of \(n\) nodes and maximum degree \(\maxDeg\), with \(\maxDeg \le n(1-\varepsilon)\).
    We will show that, for all large enough values of \(n\), and for all large enough values of \(\maxDeg \le n(1-\varepsilon)\), if \(
    \AA\) has success probability at least \(1/2\), then it must hold that $T(\maxDeg,n) > \floor{\min\{(\maxDeg+1)/2, \sqrt{\varepsilon n}/2\}} / 10$.

    Let \(k = \floor{\min\{(\maxDeg+1), \sqrt{\varepsilon n}\}}/2\).
    The lower bound graph \(G\) is constructed as follows.
    Take \(G_k\), and consider two paths \(P_1\) and \(P_2\) of length \(\floor{\varepsilon n / 4}\) and \(n - (\floor{\varepsilon n}/4 + \abs{V(G_k)} + \maxDeg - 1)\), respectively.
    By Property~\ref{prop:degree} of \cref{lem:lb_properties}, the maximum degree of \(G_k\) is at most \(2k - 1 \le \maxDeg\).
    Recall that \(G_k\) has exactly two nodes of degree \(2\): let them be \(u\) and \(v\). 
    We connect one endpoint of \(P_1\) to \(u\) and one endpoint of \(P_2\) to \(v\).
    Finally, we connect \(\maxDeg-1\) new nodes to the other endpoint of \(P_2\).
    By Property~\ref{prop:nodes} of \cref{lem:lb_properties}, observe that \(\abs{V(G_k)} \le 2k^2 - 2 \le \varepsilon n / 2\) by our choice of \(k\) and, hence, \(P_2\) is well-defined.
    Therefore, the constructed graph \(G\) has exactly \(n\) nodes and maximum degree \(\maxDeg\).

    It is easy to see that all the properties of \cref{lem:lb_properties} still hold for the induced subgraph \(G[V(G_k)]\), and, in particular, by Property~\ref{prop:monochromatic}, \(u\) and \(v\) must output different labels in any valid solution to the locally optimal cut problem on \(G\), and consecutive layers of \(G[V(G_k)]\) must be labeled differently.
    
    By contradiction, assume that $T(\maxDeg,n) \le \floor{\min\{(\maxDeg+1)/2, \sqrt{\varepsilon n}/2\}} / 10$ (which is less than \(k/10\)), and that \(\AA\) succeeds with probability at least \(1/2\).

    Consider two input graphs \(G'\) and \(G''\) that are topologically isomorphic to \(G\). 
    However, the copy of \(R^{(1)}_k\) in \(G'\) is such that the radius-\(T\) views of the nodes in \(R^{(1)}_k\) are indistinguishable to the radius-\(T\) views of the nodes in the copy of \(L^{(1)}_k\) in \(G''\).
    
    Consider two adversarial ordering sequences according to which the adversary can reveal the nodes of \(G'\) and \(G''\), respectively, to the algorithm.
    In the first sequence, the adversary first reveals the node in \(S^{(1)}\) and then the nodes in \(R^{(1)}_k\).
    In the second sequence, the adversary first reveals the node in \(S^{(1)}\) and then the nodes in \(L^{(1)}_k\).
    Note that the distance between \(S^{(1)}\) and both \(L^{(1)}_k\) and \(R^{(1)}_k\) is at least \(k \ge 10 T\), and the radius-\(T\) views of the nodes of \(G''\) in \(L^{(1)}_k\) and that of the nodes of \(G'\) in \(R^{(1)}_k\) are indistinguishable.
    Hence, given the random bit string of the algorithm, the algorithm is failing either in \(G'\) or in \(G''\) (as the outputs of the nodes in the copy of \(R^{(1)}_k\) in \(G'\) and the copy of \(L^{(1)}_k\) in \(G''\) must be the same, but also the output of the node in the copy of \(S^{(1)}\) in both graphs must be the same).
    If the success probability in \(G'\) is at least \(1/2\), then the algorithm fails in \(G''\) with probability at least \(1/2\), and vice versa, reaching a contradiction.
\end{proof}

The proof of \cref{cor:lb:quantum-local} follows by the fact that the randomized online-\local model is strictly more powerful than the quantum-\local model (and than the non-signaling model)~\cite{akbari-coiteux-roy-etal-2025-online-locality-meets}, and hence any lower bound that holds in the randomized online-\local model also holds in the quantum-\local model. 
\end{document}